\documentclass[12pt]{article}
\usepackage{lmodern}
\usepackage{microtype}
\usepackage{enumitem}
\usepackage{natbib}
\usepackage[boxed,vlined,linesnumbered]{algorithm2e}
\DontPrintSemicolon
\usepackage{float}
\usepackage{amsmath,amssymb,amsfonts,amsthm,mathtools}
\usepackage{xcolor}
\definecolor{Emerald}{HTML}{00674F}
\usepackage{hyperref}
\hypersetup{
    colorlinks=true,
    linkcolor=blue,
    citecolor=Emerald,
    filecolor=magenta,
    urlcolor=blue,
    pdftitle={Bandwidth-Free Inference for Recursive Nonlinear Impulse Responses},
    pdfpagemode=FullScreen
}
\makeatletter
\renewcommand{\hyper@natlinkstart}[1]{%
  \Hy@backout{#1}%
  \textbf\bgroup
  \hyper@linkstart{cite}{cite.#1}%
  \def\hyper@nat@current{#1}%
}
\renewcommand{\hyper@natlinkend}{%
  \hyper@linkend
  \egroup
}
\renewcommand{\hyper@natlinkbreak}[2]{%
  \hyper@linkend
  \egroup
  #1%
  \textbf\bgroup
  \hyper@linkstart{cite}{cite.#2}%
}
\makeatother
\usepackage{booktabs}
\usepackage[margin=2.25cm]{geometry}

\newtheorem{theorem}{Theorem}[section]
\newtheorem{proposition}[theorem]{Proposition}
\newtheorem{corollary}[theorem]{Corollary}
\newtheorem{lemma}[theorem]{Lemma}
\newtheorem{remark}[theorem]{Remark}
\newtheorem{assumption}[theorem]{Assumption}
\theoremstyle{definition}

\theoremstyle{plain}

\renewcommand{\Pr}{\mathbb{P}}
\renewcommand{\leq}{\leqslant}
\renewcommand{\geq}{\geqslant}

\newcommand{\diag}{\operatorname{diag}}

\newcommand{\dif}{\mathrm{d}}

\title{Bandwidth-Free Inference for Recursive Nonlinear Impulse Response Functions}
\author{Guilherme Vianna\footnote{University of S\~ao Paulo/Columbia University; \texttt{guilherme.dias.vianna@usp.br}.}}
\date{}

\begin{document}

\maketitle

\begin{abstract}
Recursive nonlinear impulse responses require an estimated innovation law whenever the impact shock is normalized by innovation ranks and future innovations are integrated out. The closest semiparametric recursive construction in the literature estimates the relevant innovation quantile functions smoothly and discusses a direct empirical-residual implementation without developing its complete first-order inference theory. We tackle this gap in a finite-dimensional nonlinear structural autoregression with unrestricted continuous marginal innovation distributions and a fixed normal-rank shock. Our estimator replaces each innovation quantile function with the empirical quantile of generated structural residuals and iterates the same structural transition. For any fixed collection of responses, we establish a joint \(\sqrt{T}\) asymptotic linear representation with four components: direct transition estimation, the effect of transition estimation on residual order statistics, ordinary innovation-quantile estimation, and the shifted impact quantile. After projection through the recursion, the quantile terms admit a residual-rank-and-spacing representation, yielding feasible inference without innovation-density estimation or quantile smoothing. We then characterize the propagated bias from smoothing, establish validity of a full recursive residual bootstrap, and derive the additional covariance contribution from a finite number of simulated paths, providing bandwidth-free inference for the empirical-residual version of the same normal-rank response used in the smooth recursive construction.
\end{abstract}



\section{Introduction}
\label{sec:introduction}

\subsection{The unresolved empirical-residual problem}
\label{sec:introduction-problem}

Impulse response functions summarize how a structural disturbance changes the future path of an economic system. In a nonlinear model, the response is commonly computed by iterating an estimated structural transition along shocked and unshocked paths and averaging their difference over the innovations that arrive after impact. When the structural innovation distributions are unknown, those distributions are part of the response calculation and their quantile functions determine both the impact shock in structural units and the future innovations through which the response propagates.

The closest semiparametric recursive analysis represents structural innovations by marginal quantile functions evaluated at Gaussian ranks and computes the response by iterating the nonlinear transition \citep{GourierouxLee2025NonlinearIRF}. Its direct implementation estimates the innovation quantile functions smoothly and even though it discusses drawing innovations from estimated residuals, it does not develop the corresponding first-order or bootstrap theory. We address this and propose a simple method: replace each unknown innovation quantile function with the empirical quantile of estimated structural residuals and use those order statistics in the recursive response.

The empirical-residual implementation is not covered by a standard parametric delta method. The residuals depend on estimated transition parameters; their order statistics are nonsmooth functions of those parameters; the impact shock evaluates one quantile function at a transformed rank; and each perturbation is carried through a nonlinear recursion. Treating the estimated residuals as observed innovations would omit a first-order channel. Treating the impact shock as fixed in structural units would instead change the population response.

We study a finite-dimensional nonlinear structural autoregression with unrestricted continuous marginal innovation distributions and an identified residual map. A selected innovation percentile is mapped to a standard-normal rank, shifted by a fixed amount, and mapped back through the same innovation quantile function. The shocked and unshocked paths start from the same state and receive the same future innovation ranks. The initial state, response variable, shocked component, shock size, future innovation law, and path coupling are fixed before any estimator is introduced. Accordingly, the empirical-quantile and smoothed-quantile procedures estimate one population response and differ only in how they estimate the unknown quantile functions.

The principal estimator first estimates the transition parameter and recovers the structural residuals. It then supplies the componentwise residual order statistics to the recursive response map. The central technical observation is that a fixed-horizon response does not require an unrestricted first-order approximation of the entire quantile process. After the recursive map is differentiated, quantile perturbations enter through a finite collection of weighted integrated-quantile functionals. These weights measure how innovations at different dates and ranks affect the terminal response. This projection permits an observation-level expansion and, after a change of variables, removes the innovation densities from feasible inference.

\subsection{Contributions}
\label{sec:introduction-contributions}

Our main contribution concerns the empirical-residual implementation of a fixed normal-rank recursive response.

First, we establish a joint \(\sqrt{T}\) asymptotic linear representation for the empirical-quantile estimator over any fixed collection of horizons, initial states, response variables, shocked components, and shock sizes. The representation separates direct transition estimation, the effect of transition estimation on generated residual order statistics, estimation of the ordinary impact and future innovation quantiles, and estimation of the shifted impact quantile. Primitive sufficient conditions are given for a scalar nonlinear location--scale autoregression, while the general structural result is stated under the corresponding componentwise generated-quantile expansion. After the quantile effects are projected through the recursive derivative, their influence contributions can be written using residual ranks and adjacent spacings. Feasible covariance estimation therefore requires neither innovation-density estimation nor a quantile-smoothing bandwidth. We are not aware of a prior result that combines these four first-order channels for this empirical-residual recursive estimator.

Next, we derive a same-target comparison with a smoothed-quantile estimator. The leading difference is the smoothing error weighted by the recursive sensitivity of the reported response to innovation quantiles at different dates and ranks. If the propagated bias is \(o\left(T^{-1/2}\right)\), the empirical and smoothed estimators have the same first-order distribution. A bias of order \(T^{-1/2}\) shifts the limiting distribution, while a larger bias dominates sampling uncertainty. This result isolates the effect of smoothing without changing the shock definition, the future innovation law, or the response being estimated.

Finally, we establish conditional validity of a full recursive residual bootstrap. Each replication regenerates a complete time series, re-estimates the transition, reconstructs the residual quantile functions, and recomputes the response. This re-estimation reproduces all four components of the influence representation. We separately characterize numerical integration error. When the number of simulated paths is proportional to the sample size, it contributes an additional covariance term; when it grows faster, the term is asymptotically negligible. Path-only resampling is shown to estimate numerical integration uncertainty conditional on the fitted model rather than sampling uncertainty in the recursive estimator.

\subsection{Relation to the literature}
\label{sec:introduction-literature}

Our object of study belongs to the literature on nonlinear structural impulse responses defined by comparing recursively simulated shocked and unshocked paths \citep{KoopPesaranPotter1996,GourierouxLee2025NonlinearIRF,Ballarin2025}. That literature establishes the role of the initial state, shock size, future innovations, and nonlinear transition. We take structural identification as maintained and focus on inference when the innovation distributions used in the recursion are estimated from generated residuals. Relative to the closest normal-rank construction, the main change is the empirical-residual estimator and its first-order, smoothing, bootstrap, and finite-simulation theory.

Nonparametric and semiparametric local projections estimate nonlinear responses through horizon-specific conditional-mean relations rather than by iterating a structural transition \citep{Jorda2005,PlagborgMollerWolf2021,JordaTaylor2025,GoncalvesHerreraKilianPesavento2024,GoncalvesHerreraKilianPesaventoHolban2026}. These methods address related questions through a different estimation route. The present results neither require nor imply a general ranking between recursive and local-projection procedures; they provide inference for the innovation-distribution component of a recursively defined structural response.

\subsection{Organization}
\label{sec:introduction-organization}

Section~2 defines the population response and the empirical-quantile estimator, with the smoothed estimator introduced as a paired comparison. Section~3 develops the generated-residual expansion, the recursive influence representation, and density-free feasible inference. Section~4 characterizes the effect of quantile smoothing under the same target. Section~5 studies the full recursive residual bootstrap and numerical integration error. Section~6 concludes. The appendices provide primitive conditions, derivations, proofs, and implementation conventions.

\section{Population response and empirical-quantile estimation}
\label{sec:model-estimators}

This section fixes the population impulse-response function before estimation. A scalar model first displays the two roles of the innovation quantile function. The vector formulation then defines the normal-rank structural response. The empirical-quantile estimator is introduced as the principal procedure, followed by a smoothed estimator that changes only the quantile estimate supplied to the same recursive map.

\subsection{A scalar model that displays the estimation problem}
\label{sec:scalar-example}

Consider the nonlinear location-scale autoregression
\begin{equation}
Y_t
=
\mu\left(Y_{t-1};\boldsymbol{\beta}_0\right)
+
\sigma\left(Y_{t-1};\boldsymbol{\beta}_0\right)U_t,
\qquad
\sigma\left(y;\boldsymbol{\beta}_0\right)>0,
\label{eq:scalar-model}
\end{equation}
where \(\boldsymbol{\beta}_0\) is finite-dimensional. The innovations are independent and identically distributed with continuous distribution function \(F_0\) and quantile function \(Q_0=F_0^{-1}\). The positivity of \(\sigma\) yields the residual map
\[
G\left(y_t,y_{t-1};\boldsymbol{\beta}\right)
=
\frac{y_t-\mu\left(y_{t-1};\boldsymbol{\beta}\right)}{\sigma\left(y_{t-1};\boldsymbol{\beta}\right)},
\]
so that \(U_t=G\left(Y_t,Y_{t-1};\boldsymbol{\beta}_0\right)\).

Let \(P_1=F_0\left(U_1\right)\), and define a shock of size \(\delta\in\mathbb{R}\) through
\begin{equation}
\tau_{\delta}\left(p\right)
=
\Phi\left(\Phi^{-1}\left(p\right)+\delta\right),
\qquad
p\in\left(0,1\right).
\label{eq:rank-shift}
\end{equation}
The impact innovations are \(Q_0\left(P_1\right)\) and \(Q_0\left(\tau_{\delta}\left(P_1\right)\right)\). Thus, \(\delta\) is measured on the standard-normal rank scale, while the change in structural innovation units depends on \(Q_0\) and the realized rank.

Fix an initial value \(y\) and let \(P_1,\ldots,P_h\) be independent uniform random variables. For a generic parameter \(\boldsymbol{\beta}\) and quantile function \(Q\), construct the unshocked and shocked paths from the same future ranks. Their difference satisfies
\begin{equation}
\begin{aligned}
D_1
&=
\sigma\left(y;\boldsymbol{\beta}\right)
\left[
Q\left(\tau_{\delta}\left(P_1\right)\right)-Q\left(P_1\right)
\right],\\
D_j
&=
\mu\left(Y_{j-1}^{\delta};\boldsymbol{\beta}\right)
-
\mu\left(Y_{j-1}^{0};\boldsymbol{\beta}\right)
+
\left[
\sigma\left(Y_{j-1}^{\delta};\boldsymbol{\beta}\right)
-
\sigma\left(Y_{j-1}^{0};\boldsymbol{\beta}\right)
\right]Q\left(P_j\right),
\qquad
j=2,\ldots,h.
\end{aligned}
\label{eq:scalar-response-recursion}
\end{equation}
Equation~\eqref{eq:scalar-response-recursion} displays the two uses of the unknown quantile function. It determines the impact difference and the future innovations through which that difference propagates. The scalar impulse-response function is \(\psi_h\left(y,\delta\right)=\mathbb{E}\left[D_h\right]\).

\subsection{Structural response in the general model}
\label{sec:general-response}

Let \(\boldsymbol{Y}_t\in\mathbb{R}^{d}\) satisfy
\begin{equation}
\boldsymbol{Y}_t
=
g\left(\boldsymbol{Y}_{t-1},\boldsymbol{U}_t;\boldsymbol{\beta}_0\right),
\label{eq:structural-transition}
\end{equation}
where \(\boldsymbol{U}_t\in\mathbb{R}^{n}\) is the vector of structural innovations. Higher-order dynamics are represented by augmenting the state. We assume that the structural representation is identified and that a known residual map \(G\) satisfies
\[
\boldsymbol{U}_t
=
G\left(\boldsymbol{Y}_t,\boldsymbol{Y}_{t-1};\boldsymbol{\beta}_0\right).
\]

The innovation vectors are independent and identically distributed over time. Their components are mutually independent, with continuous marginal distribution functions \(F_{j0}\) and quantile functions \(Q_{j0}\), for \(j=1,\ldots,n\). Define
\[
\boldsymbol{Q}_0\left(\boldsymbol{p}\right)
=
\left(
Q_{10}\left(p_1\right),
\ldots,
Q_{n0}\left(p_n\right)
\right)^\prime.
\]
Then \(\boldsymbol{U}_t=\boldsymbol{Q}_0\left(\boldsymbol{P}_t\right)\), where the components of \(\boldsymbol{P}_t\) are independent uniform random variables. For a selected component \(k\), let \(\boldsymbol{\tau}_{k,\delta}\left(\boldsymbol{p}\right)\) apply the map in equation~\eqref{eq:rank-shift} to \(p_k\) and leave the remaining coordinates unchanged.

Fix an initial state \(\boldsymbol{y}\), a response vector \(\boldsymbol{a}\), and a horizon \(h\). For \(r\in\left\{0,\delta\right\}\), set \(\boldsymbol{P}_1^0=\boldsymbol{P}_1\), \(\boldsymbol{P}_1^{\delta}=\boldsymbol{\tau}_{k,\delta}\left(\boldsymbol{P}_1\right)\), and \(\boldsymbol{P}_j^0=\boldsymbol{P}_j^{\delta}=\boldsymbol{P}_j\) for \(j\geq2\). For a generic \(\boldsymbol{\beta}\) and \(\boldsymbol{Q}\), define
\begin{equation}
\boldsymbol{Y}_0^r=\boldsymbol{y},
\qquad
\boldsymbol{Y}_j^r
=
g\left(\boldsymbol{Y}_{j-1}^r,\boldsymbol{Q}\left(\boldsymbol{P}_j^r\right);\boldsymbol{\beta}\right),
\qquad
j=1,\ldots,h.
\label{eq:paired-paths}
\end{equation}
The path response is
\[
D_h\left(\boldsymbol{P}_{1:h};\boldsymbol{\beta},\boldsymbol{Q},\boldsymbol{y},\boldsymbol{a},k,\delta\right)
=
\boldsymbol{a}^\prime\left(\boldsymbol{Y}_h^{\delta}-\boldsymbol{Y}_h^0\right),
\]
and the population impulse-response function is
\begin{equation}
\psi_h\left(\boldsymbol{y},\boldsymbol{a},k,\delta\right)
=
\mathbb{E}\left[
D_h\left(\boldsymbol{P}_{1:h};\boldsymbol{\beta}_0,\boldsymbol{Q}_0,\boldsymbol{y},\boldsymbol{a},k,\delta\right)
\right].
\label{eq:population-irf}
\end{equation}
The initial state, response variable, shocked component, shock map, and innovation law in equation~\eqref{eq:population-irf} remain fixed throughout the estimator comparison.

\subsection{Empirical-quantile estimator and smoothing comparator}
\label{sec:estimators}

Let \(\boldsymbol{\hat{\beta}}\) estimate \(\boldsymbol{\beta}_0\), and recover the structural residuals by
\[
\boldsymbol{\hat{U}}_t
=
G\left(\boldsymbol{Y}_t,\boldsymbol{Y}_{t-1};\boldsymbol{\hat{\beta}}\right),
\qquad
t=1,\ldots,T.
\]
For component \(j\), let \(\hat{U}_{j,\left(1\right)}\leq\cdots\leq\hat{U}_{j,\left(T\right)}\) denote the residual order statistics.

\subsubsection{Empirical-quantile estimator}

The empirical innovation quantile is
\[
\hat{Q}_{j}^{\mathrm{E}}\left(p\right)
=
\hat{U}_{j,\left(\lceil Tp\rceil\right)},
\qquad
p\in\left(0,1\right).
\]
This estimate uses the generated residual ranks directly and requires no smoothing parameter. Collect the componentwise estimates in \(\boldsymbol{\hat{Q}}^{\mathrm{E}}\).

\subsubsection{Paired smoothed estimator}

For comparison, apply a Gaussian rank smoother \(\mathcal{S}_{b_T}\) to the same residual order statistics,
\[
\hat{Q}_{j,b_T}^{\mathrm{S}}
=
\mathcal{S}_{b_T}\hat{Q}_{j}^{\mathrm{E}},
\]
where \(b_T\) is the bandwidth. Section~\ref{sec:smoothing} states the conditions on this operator. No smoothing choice enters the empirical estimator.

Let \(\boldsymbol{P}_{s,1:h}\), for \(s=1,\ldots,S\), be simulation draws independent of the sample. For \(m\in\left\{\mathrm{E},\mathrm{S}\right\}\), collect the componentwise quantile estimates in \(\boldsymbol{\hat{Q}}^{m}\) and define
\begin{equation}
\hat{\psi}_{h,S}^{m}\left(\boldsymbol{y},\boldsymbol{a},k,\delta\right)
=
\frac{1}{S}
\sum_{s=1}^{S}
D_h\left(
\boldsymbol{P}_{s,1:h};
\boldsymbol{\hat{\beta}},
\boldsymbol{\hat{Q}}^{m},
\boldsymbol{y},
\boldsymbol{a},
k,
\delta
\right).
\label{eq:unified-estimator}
\end{equation}
Both estimators in equation~\eqref{eq:unified-estimator} use the same transition estimate, response specification, and simulation ranks. Their difference is therefore attributable to the quantile estimate supplied to the recursive response map. Common simulation ranks reduce numerical variation in this paired comparison; they do not replace sampling inference.

\section{First-order inference with empirical residual quantiles}
\label{sec:main-theory}

For a fixed response specification, write
\[
\Psi_h\left(\boldsymbol{\beta},\boldsymbol{Q}\right)
=
\mathbb{E}\left[
D_h\left(
\boldsymbol{P}_{1:h};
\boldsymbol{\beta},
\boldsymbol{Q},
\boldsymbol{y},
\boldsymbol{a},
k,
\delta
\right)
\right]
\]
and let \(\Psi_{h,S}\left(\boldsymbol{\beta},\boldsymbol{Q}\right)\) denote the corresponding average over the \(S\) simulated rank paths. Thus, \(\psi_h=\Psi_h\left(\boldsymbol{\beta}_0,\boldsymbol{Q}_0\right)\) and \(\hat{\psi}_{h,S}^{\mathrm{E}}=\Psi_{h,S}\left(\boldsymbol{\hat{\beta}},\boldsymbol{\hat{Q}}^{\mathrm{E}}\right)\). The arguments \(\left(\boldsymbol{y},\boldsymbol{a},k,\delta\right)\) are suppressed in this section.

\begin{assumption}[First-order regularity]
\label{ass:first-order}
The following conditions hold for every response specification considered below.
\begin{enumerate}[label=(\roman*),leftmargin=2.2em]
\item The transition in equation~\eqref{eq:structural-transition} has a unique strictly stationary and ergodic solution. Its dependence and moments are sufficient for laws of large numbers and a joint central limit theorem for the influence sequences defined below. The maps \(g\) and \(G\) are continuously differentiable in the state, innovation, and parameter arguments used in the analysis. For any fixed maximum horizon, the derivatives of the paired paths are bounded by a random variable with a finite \(2+\eta\) moment for some \(\eta>0\).

\item The transition estimator is consistent and asymptotically linear:
\begin{equation}
\sqrt{T}\left(\boldsymbol{\hat{\beta}}-\boldsymbol{\beta}_0\right)
=
\frac{1}{\sqrt{T}}
\sum_{t=1}^{T}
\boldsymbol{L}_t
+
o_{\mathbb{P}}\left(1\right),
\qquad
\mathbb{E}\left[\boldsymbol{L}_t\right]
=
\boldsymbol{0},
\label{eq:beta-linearization}
\end{equation}
where \(\left\{\boldsymbol{L}_t\right\}\) is stationary, has a finite \(2+\eta\) moment, and satisfies the joint central limit theorem in part~\textup{(i)}.

\item Each marginal innovation distribution \(F_{j0}\) has a positive, continuously differentiable density \(f_{j0}\) on the interior of its support. The tails of \(Q_{j0}\), the reciprocal densities, and the path derivatives are controlled so that the derivative and influence terms defined below are square integrable. These conditions also make \(\Psi_h\) continuous and first-order differentiable in a weighted uniform norm that controls the ranks generated by \(\tau_{\delta}\).

\item For each innovation component, the empirical quantile function computed from generated residuals satisfies
\begin{equation}
\left\|
\sqrt{T}\left(\hat{Q}_{j}^{\mathrm{E}}-Q_{j0}\right)
-
\frac{1}{\sqrt{T}}
\sum_{t=1}^{T}
\left[
\Xi_{j,t}
+
\boldsymbol{r}_j\left(\cdot\right)^{\prime}\boldsymbol{L}_t
\right]
\right\|_j
=
o_{\mathbb{P}}\left(1\right),
\label{eq:generated-quantile-expansion}
\end{equation}
where \(\|\cdot\|_j\) is the norm in part~\textup{(iii)}, \(\boldsymbol{r}_j\) is a deterministic vector-valued function, and
\[
\Xi_{j,t}\left(p\right)
=
\frac{
p-
\mathbf{1}\left\{
U_{jt}\leq Q_{j0}\left(p\right)
\right\}
}{
f_{j0}\left(Q_{j0}\left(p\right)\right)
},
\qquad
p\in\left(0,1\right).
\]
For the smooth residual maps considered here, the generated-residual correction can be written as
\[
\boldsymbol{r}_j\left(p\right)
=
-
\frac{
\boldsymbol{\Gamma}_j\left(Q_{j0}\left(p\right)\right)
}{
f_{j0}\left(Q_{j0}\left(p\right)\right)
},
\qquad
\boldsymbol{\Gamma}_j\left(u\right)
=
\left.
\frac{\partial}{\partial\boldsymbol{\beta}}
\mathbb{E}\left[
\mathbf{1}\left\{
G_j\left(
\boldsymbol{Y}_t,
\boldsymbol{Y}_{t-1};
\boldsymbol{\beta}
\right)
\leq u
\right\}
\right]
\right|_{\boldsymbol{\beta}=\boldsymbol{\beta}_0}.
\]
Thus, \(\boldsymbol{r}_j\left(p\right)^{\prime}\boldsymbol{L}_t\) is the first-order effect of estimating \(\boldsymbol{\beta}_0\) before forming residual order statistics.

\item The number of response specifications and their horizons are fixed as \(T\) increases. The simulation draws are independent of the sample. Consistency requires \(S\rightarrow\infty\); the first-order results below impose \(S/T\rightarrow\infty\).
\end{enumerate}
\end{assumption}

Assumption~\ref{ass:first-order} separates the requirements imposed by the dynamic model, the transition estimator, the innovation distributions, and the generated residuals. The vector theorem below uses equation~\eqref{eq:generated-quantile-expansion} as a componentwise high-level condition. The next proposition verifies that condition under the primitive scalar assumptions and records the corresponding vector extension.

\begin{proposition}[Generated-residual empirical quantiles]
\label{prop:generated-residual-quantile-expansion}
Under Assumptions~\ref{ass:scalar-dynamics}--\ref{ass:scalar-estimator} and~\ref{ass:response-tail-continuity}, equation~\eqref{eq:generated-quantile-expansion} holds for the scalar model. In the vector model, the same conclusion holds jointly across components under the residual-process conditions in Lemma~\ref{lem:generated-residual-empirical-process}, the componentwise density conditions in Assumption~\ref{ass:first-order}\textup{(iii)}, and Assumption~\ref{ass:response-tail-continuity}.
\end{proposition}

Proposition~\ref{prop:generated-residual-quantile-expansion} is the bridge between residual empirical-process theory and the recursive response. Its proof is given in Appendix~\ref{app:residual-quantiles}; Appendix~\ref{app:assumptions} states the primitive scalar conditions in full.

We next introduce the derivatives needed for the first-order result. Let
\[
\Lambda_{h,s,j}^{r}
=
\left.
\frac{
\partial\left(
\boldsymbol{a}^{\prime}\boldsymbol{Y}_h^{r}
\right)
}{
\partial u_{j,s}^{r}
}
\right|_{\left(\boldsymbol{\beta}_0,\boldsymbol{Q}_0\right)},
\qquad
r\in\left\{0,\delta\right\},
\]
be the sensitivity of the date-\(h\) outcome along path \(r\) to structural innovation \(j\) at date \(s\). These sensitivities are obtained by differentiating the finite recursion in equation~\eqref{eq:paired-paths}. Also define the direct transition derivative
\[
\boldsymbol{A}_h
=
\mathbb{E}\left[
\left.
\frac{\partial}{\partial\boldsymbol{\beta}}
D_h\left(
\boldsymbol{P}_{1:h};
\boldsymbol{\beta},
\boldsymbol{Q}_0,
\boldsymbol{y},
\boldsymbol{a},
k,
\delta
\right)
\right|_{\boldsymbol{\beta}=\boldsymbol{\beta}_0}
\right],
\]
which holds the innovation quantile functions fixed.

For a collection of quantile perturbations \(\boldsymbol{q}=\left(q_1,\ldots,q_n\right)^{\prime}\), split the derivative with respect to \(\boldsymbol{Q}_0\) into
\begin{equation}
\begin{aligned}
\dot{\Psi}_{h}^{\mathrm{dist}}\left[\boldsymbol{q}\right]
&=
\mathbb{E}\left[
\sum_{j=1}^{n}
\left(
\mathbf{1}\left\{j\neq k\right\}
\Lambda_{h,1,j}^{\delta}
-
\Lambda_{h,1,j}^{0}
\right)
q_j\left(P_{j1}\right)
+
\sum_{s=2}^{h}
\sum_{j=1}^{n}
\left(
\Lambda_{h,s,j}^{\delta}
-
\Lambda_{h,s,j}^{0}
\right)
q_j\left(P_{js}\right)
\right],\\
\dot{\Psi}_{h}^{\mathrm{imp}}\left[\boldsymbol{q}\right]
&=
\mathbb{E}\left[
\Lambda_{h,1,k}^{\delta}
q_k\left(
\tau_{\delta}\left(P_{k1}\right)
\right)
\right].
\end{aligned}
\label{eq:response-quantile-derivatives}
\end{equation}
The first map covers the unshifted impact innovations and the common future innovations. The second isolates the shifted innovation in the shocked path at impact.

For later use, collect the empirical-quantile influence functions in
\[
\boldsymbol{\Xi}_t
=
\left(
\Xi_{1,t},
\ldots,
\Xi_{n,t}
\right)^{\prime},
\]
and define the generated-residual perturbation
\[
\boldsymbol{R}_t
=
\left(
p\mapsto
\boldsymbol{r}_1\left(p\right)^{\prime}\boldsymbol{L}_t,
\ldots,
p\mapsto
\boldsymbol{r}_n\left(p\right)^{\prime}\boldsymbol{L}_t
\right)^{\prime}.
\]

Before turning to the first-order result, consistency follows from the decomposition
\[
\begin{aligned}
\hat{\psi}_{h,S}^{\mathrm{E}}-\psi_h
&=
\left[
\Psi_{h,S}\left(
\boldsymbol{\hat{\beta}},
\boldsymbol{\hat{Q}}^{\mathrm{E}}
\right)
-
\Psi_h\left(
\boldsymbol{\hat{\beta}},
\boldsymbol{\hat{Q}}^{\mathrm{E}}
\right)
\right]\\
&\quad+
\left[
\Psi_h\left(
\boldsymbol{\hat{\beta}},
\boldsymbol{\hat{Q}}^{\mathrm{E}}
\right)
-
\Psi_h\left(
\boldsymbol{\beta}_0,
\boldsymbol{\hat{Q}}^{\mathrm{E}}
\right)
\right]\\
&\quad+
\left[
\Psi_h\left(
\boldsymbol{\beta}_0,
\boldsymbol{\hat{Q}}^{\mathrm{E}}
\right)
-
\Psi_h\left(
\boldsymbol{\beta}_0,
\boldsymbol{Q}_0
\right)
\right].
\end{aligned}
\]
The three terms are, respectively, Monte Carlo integration error, transition-estimation error, and residual-quantile estimation error.

\begin{proposition}[Consistency]
\label{prop:consistency}

Under Assumption~\ref{ass:first-order}\textup{(i)--(iv)}, if \(S\rightarrow\infty\), then for any fixed finite collection of response specifications,
\[
\max_{1\leq m\leq M}
\left|
\hat{\psi}_{m,S}^{\mathrm{E}}
-
\psi_m
\right|
\xrightarrow{\mathbb{P}}
0.
\]
If the smoother satisfies
\[
\left\|
\hat{Q}_{j,b_T}^{\mathrm{S}}
-
Q_{j0}
\right\|_j
\xrightarrow{\mathbb{P}}
0
\]
for every \(j\), the same conclusion holds for \(\hat{\psi}_{m,S}^{\mathrm{S}}\).
\end{proposition}

\begin{theorem}[Asymptotic linear representation]
\label{thm:main-linearization}

Suppose Assumption~\ref{ass:first-order} holds and \(S/T\rightarrow\infty\). For each fixed response specification,
\begin{equation}
\sqrt{T}\left(
\hat{\psi}_{h,S}^{\mathrm{E}}
-
\psi_h
\right)
=
\frac{1}{\sqrt{T}}
\sum_{t=1}^{T}
Z_{h,t}
+
o_{\mathbb{P}}\left(1\right),
\label{eq:empirical-irf-linearization}
\end{equation}
where the observation-level influence contribution is
\begin{equation}
\begin{aligned}
Z_{h,t}
&=
Z_{h,t}^{\mathrm{tr}}
+
Z_{h,t}^{\mathrm{res}}
+
Z_{h,t}^{\mathrm{dist}}
+
Z_{h,t}^{\mathrm{imp}},\\
Z_{h,t}^{\mathrm{tr}}
&=
\boldsymbol{A}_h^{\prime}\boldsymbol{L}_t,\\
Z_{h,t}^{\mathrm{res}}
&=
\dot{\Psi}_{h}^{\mathrm{dist}}\left[
\boldsymbol{R}_t
\right]
+
\dot{\Psi}_{h}^{\mathrm{imp}}\left[
\boldsymbol{R}_t
\right],\\
Z_{h,t}^{\mathrm{dist}}
&=
\dot{\Psi}_{h}^{\mathrm{dist}}\left[
\boldsymbol{\Xi}_t
\right],\\
Z_{h,t}^{\mathrm{imp}}
&=
\dot{\Psi}_{h}^{\mathrm{imp}}\left[
\boldsymbol{\Xi}_t
\right].
\end{aligned}
\label{eq:influence-decomposition}
\end{equation}

More generally, let \(m=1,\ldots,M\) index any fixed collection of horizons, initial states, response vectors, shocked components, and shock sizes. Stack the corresponding estimators, population responses, and influence contributions in \(\boldsymbol{\hat{\psi}}_{S}^{\mathrm{E}}\), \(\boldsymbol{\psi}\), and \(\boldsymbol{Z}_t\). Then
\begin{equation}
\sqrt{T}\left(
\boldsymbol{\hat{\psi}}_{S}^{\mathrm{E}}
-
\boldsymbol{\psi}
\right)
=
\frac{1}{\sqrt{T}}
\sum_{t=1}^{T}
\boldsymbol{Z}_t
+
o_{\mathbb{P}}\left(1\right).
\label{eq:joint-empirical-irf-linearization}
\end{equation}
\end{theorem}

Theorem~\ref{thm:main-linearization} preserves the covariance among all four sources of uncertainty because they are assembled at the observation level before the long-run covariance is computed. In particular, the direct and generated-residual terms share the transition influence \(\boldsymbol{L}_t\), while the distribution and impact terms are formed from the same residual observation.

\begin{corollary}[Gaussian limit]
\label{cor:main-gaussian-limit}

Under the conditions of Theorem~\ref{thm:main-linearization},
\begin{equation}
\sqrt{T}\left(
\boldsymbol{\hat{\psi}}_{S}^{\mathrm{E}}
-
\boldsymbol{\psi}
\right)
\xrightarrow{\mathrm{d}}
\mathcal{N}\left(
\boldsymbol{0},
\boldsymbol{\Omega}
\right),
\qquad
\boldsymbol{\Omega}
=
\sum_{\ell=-\infty}^{\infty}
\operatorname{Cov}\left(
\boldsymbol{Z}_0,
\boldsymbol{Z}_{\ell}
\right).
\label{eq:empirical-irf-limit}
\end{equation}
For one response, the asymptotic variance is the corresponding diagonal element of \(\boldsymbol{\Omega}\). If \(\left\{\boldsymbol{Z}_t\right\}\) is a martingale difference sequence, the sum reduces to
\[
\mathbb{E}\left[
\boldsymbol{Z}_t
\boldsymbol{Z}_t^{\prime}
\right].
\]
\end{corollary}

\subsection{Interpretation of the influence decomposition}
\label{sec:main-theory-interpretation}

Equation~\eqref{eq:influence-decomposition} keeps four uses of estimated objects separate while preserving their covariance at the observation level. Table~\ref{tab:influence-components} summarizes the source and scope of each component.

\begin{table}[t]
\centering
\caption{Sources of first-order uncertainty}
\label{tab:influence-components}
\begin{tabular}{p{2.4cm}p{7.0cm}p{5.0cm}}
\toprule
Component & Source & Vanishes when \\
\midrule
\(Z_{h,t}^{\mathrm{tr}}\) & The transition parameter changes the shocked and unshocked recursive paths directly. & \(\boldsymbol{\beta}_0\) is known. \\
\addlinespace
\(Z_{h,t}^{\mathrm{res}}\) & The transition estimate changes the generated residual values and their order statistics before the response is simulated. & \(\boldsymbol{\beta}_0\) is known, or \(\boldsymbol{Q}_0\) is treated as known. \\
\addlinespace
\(Z_{h,t}^{\mathrm{dist}}\) & The unshifted impact innovations and the common future innovations use estimated quantiles. & \(\boldsymbol{Q}_0\) is known. \\
\addlinespace
\(Z_{h,t}^{\mathrm{imp}}\) & The shocked impact innovation uses the estimated quantile evaluated at the shifted rank. & \(\boldsymbol{Q}_0\) is known, or the impact intervention is fixed in structural units. \\
\bottomrule
\end{tabular}
\end{table}

The direct and generated-residual components share the transition influence \(\boldsymbol{L}_t\), while the distribution and shifted-impact components are formed from the same residual observation. Computing their covariance after separate aggregation would therefore lose first-order cross terms. The decomposition also clarifies the nested cases: known transition parameters remove the first two components; known innovation distributions leave only direct transition uncertainty; and a fixed additive structural shock removes the separate shifted-impact component without generally removing distribution uncertainty from future innovations.

\subsection{Density-free feasible inference}
\label{sec:main-theory-feasible}

Theorem~\ref{thm:main-linearization} yields feasible inference once its observation-level contributions are estimated. A direct use of the empirical-quantile influence function in equation~\eqref{eq:generated-quantile-expansion} appears to require the innovation densities. Here those densities cancel after the quantile effects are integrated over the ranks entering the impulse-response function.

Let \(\rho_{\delta}=\tau_{\delta}^{-1}=\tau_{-\delta}\). If \(\phi\) denotes the standard normal density, then
\[
\rho_{\delta}\left(p\right)
=
\Phi\left(\Phi^{-1}\left(p\right)-\delta\right),
\qquad
\rho_{\delta}'\left(p\right)
=
\frac{\phi\left(\Phi^{-1}\left(p\right)-\delta\right)}
{\phi\left(\Phi^{-1}\left(p\right)\right)}.
\]
For each innovation component, define
\begin{equation}
\begin{aligned}
\omega_{h,j}^{\mathrm{dist}}\left(p\right)
&=
\mathbb{E}\left[
\left(
\mathbf{1}\left\{j\neq k\right\}\Lambda_{h,1,j}^{\delta}
-
\Lambda_{h,1,j}^{0}
\right)
\mathrel{\big|}
P_{j1}=p
\right]
+
\sum_{s=2}^{h}
\mathbb{E}\left[
\Lambda_{h,s,j}^{\delta}-\Lambda_{h,s,j}^{0}
\mathrel{\big|}
P_{js}=p
\right],\\
\omega_{h,k}^{\mathrm{imp}}\left(p\right)
&=
\mathbb{E}\left[
\Lambda_{h,1,k}^{\delta}
\mathrel{\big|}
P_{k1}=p
\right].
\end{aligned}
\label{eq:propagation-weights}
\end{equation}
The first weight collects the effects of the unshifted impact innovations and the common future innovations. The second collects the effect of the shifted impact innovation. Consequently,
\[
\dot{\Psi}_{h}^{\mathrm{dist}}\left[\boldsymbol{q}\right]
=
\sum_{j=1}^{n}
\int_{0}^{1}
\omega_{h,j}^{\mathrm{dist}}\left(p\right)
q_j\left(p\right)
\,\mathrm{d}p,
\qquad
\dot{\Psi}_{h}^{\mathrm{imp}}\left[\boldsymbol{q}\right]
=
\int_{0}^{1}
\omega_{h,k}^{\mathrm{imp}}\left(p\right)
q_k\left(\tau_{\delta}\left(p\right)\right)
\,\mathrm{d}p.
\]

Let
\[
\boldsymbol{J}_{j,t}
=
\left.
\frac{\partial}{\partial\boldsymbol{\beta}}
G_j\left(
\boldsymbol{Y}_t,
\boldsymbol{Y}_{t-1};
\boldsymbol{\beta}
\right)
\right|_{\boldsymbol{\beta}=\boldsymbol{\beta}_0}
\]
denote the derivative of structural residual \(j\) with respect to the transition parameter.

\begin{proposition}[Influence contributions without density estimation]
\label{prop:density-free-influence}

Suppose Assumption~\ref{ass:first-order} holds and the generated-residual correction is induced by the differentiable residual map \(G\). Then the residual, innovation-distribution, and shifted-impact terms in equation~\eqref{eq:influence-decomposition} satisfy
\begin{equation}
\begin{aligned}
\boldsymbol{B}_{h}^{\mathrm{res}}
&=
\sum_{j=1}^{n}
\mathbb{E}\left[
\omega_{h,j}^{\mathrm{dist}}\left(P_{jt}\right)
\boldsymbol{J}_{j,t}
\right]
+
\mathbb{E}\left[
\omega_{h,k}^{\mathrm{imp}}\left(
\rho_{\delta}\left(P_{kt}\right)
\right)
\rho_{\delta}'\left(P_{kt}\right)
\boldsymbol{J}_{k,t}
\right],\\
Z_{h,t}^{\mathrm{res}}
&=
\left(
\boldsymbol{B}_{h}^{\mathrm{res}}
\right)^{\prime}
\boldsymbol{L}_t,\\
Z_{h,t}^{\mathrm{dist}}
&=
\sum_{j=1}^{n}
\int_{\mathbb{R}}
\omega_{h,j}^{\mathrm{dist}}\left(
F_{j0}\left(u\right)
\right)
\left[
F_{j0}\left(u\right)
-
\mathbf{1}\left\{U_{jt}\leq u\right\}
\right]
\,\mathrm{d}u,\\
Z_{h,t}^{\mathrm{imp}}
&=
\int_{\mathbb{R}}
\omega_{h,k}^{\mathrm{imp}}\left(
\rho_{\delta}\left(
F_{k0}\left(u\right)
\right)
\right)
\rho_{\delta}'\left(
F_{k0}\left(u\right)
\right)
\left[
F_{k0}\left(u\right)
-
\mathbf{1}\left\{U_{kt}\leq u\right\}
\right]
\,\mathrm{d}u.
\end{aligned}
\label{eq:density-free-influence}
\end{equation}
\end{proposition}

The result follows by changing variables from ranks to innovation units in the two quantile derivatives. For the generated-residual term, differentiability of \(G\) gives
\[
\boldsymbol{r}_j\left(p\right)
=
\mathbb{E}\left[
\boldsymbol{J}_{j,t}
\mathrel{\big|}
U_{jt}=Q_{j0}\left(p\right)
\right].
\]
Appendix~\ref{app:residual-quantiles} gives the formal argument.

We now construct sample analogues. For a transition estimator defined by the estimating equation
\[
\frac{1}{T}
\sum_{t=1}^{T}
\boldsymbol{S}_t\left(
\boldsymbol{\hat{\beta}}
\right)
=
\boldsymbol{0},
\]
let
\[
\boldsymbol{\hat{H}}
=
\frac{1}{T}
\sum_{t=1}^{T}
\frac{
\partial
\boldsymbol{S}_t\left(
\boldsymbol{\hat{\beta}}
\right)
}{
\partial\boldsymbol{\beta}^{\prime}
},
\qquad
\boldsymbol{\hat{L}}_t
=
-
\boldsymbol{\hat{H}}^{-1}
\boldsymbol{S}_t\left(
\boldsymbol{\hat{\beta}}
\right).
\]
Other regular transition estimators enter through their estimated influence contributions. Holding \(\boldsymbol{\hat{Q}}^{\mathrm{E}}\) fixed, estimate the direct transition derivative by
\[
\boldsymbol{\hat{A}}_h
=
\frac{1}{S}
\sum_{s=1}^{S}
\left.
\frac{\partial}{\partial\boldsymbol{\beta}}
D_h\left(
\boldsymbol{P}_{s,1:h};
\boldsymbol{\beta},
\boldsymbol{\hat{Q}}^{\mathrm{E}},
\boldsymbol{y},
\boldsymbol{a},
k,
\delta
\right)
\right|_{\boldsymbol{\beta}=\boldsymbol{\hat{\beta}}}.
\]
The derivative and the path sensitivities in equation~\eqref{eq:propagation-weights} are obtained by automatic differentiation through equation~\eqref{eq:paired-paths}, or by the equivalent finite recursive derivatives in Appendix~\ref{app:recursive-derivatives}. To estimate a propagation weight at rank \(p\), fix the relevant simulated rank at \(p\) and average over the remaining ranks. Thus, the conditional expectations in equation~\eqref{eq:propagation-weights} are numerical integrals over inputs controlled by the researcher rather than nonparametric regressions on observed data.

Let \(\hat{R}_{jt}\) be the rank of \(\hat{U}_{jt}\) among the component-\(j\) residuals, set
\[
\hat{P}_{jt}
=
\frac{\hat{R}_{jt}-1/2}{T},
\]
and define
\[
\Delta\hat{U}_{j,r}
=
\hat{U}_{j,\left(r+1\right)}
-
\hat{U}_{j,\left(r\right)}.
\]
Also let
\[
\boldsymbol{\hat{J}}_{j,t}
=
\left.
\frac{\partial}{\partial\boldsymbol{\beta}}
G_j\left(
\boldsymbol{Y}_t,
\boldsymbol{Y}_{t-1};
\boldsymbol{\beta}
\right)
\right|_{\boldsymbol{\beta}=\boldsymbol{\hat{\beta}}}.
\]
Estimate the generated-residual coefficient by
\[
\boldsymbol{\hat{B}}_{h}^{\mathrm{res}}
=
\frac{1}{T}
\sum_{t=1}^{T}
\left[
\sum_{j=1}^{n}
\hat{\omega}_{h,j}^{\mathrm{dist}}\left(
\hat{P}_{jt}
\right)
\boldsymbol{\hat{J}}_{j,t}
+
\hat{\omega}_{h,k}^{\mathrm{imp}}\left(
\rho_{\delta}\left(
\hat{P}_{kt}
\right)
\right)
\rho_{\delta}'\left(
\hat{P}_{kt}
\right)
\boldsymbol{\hat{J}}_{k,t}
\right].
\]
The four estimated contributions are
\begin{equation}
\begin{aligned}
\hat{Z}_{h,t}^{\mathrm{tr}}
&=
\boldsymbol{\hat{A}}_h^{\prime}
\boldsymbol{\hat{L}}_t,
\qquad
\hat{Z}_{h,t}^{\mathrm{res}}
=
\left(
\boldsymbol{\hat{B}}_{h}^{\mathrm{res}}
\right)^{\prime}
\boldsymbol{\hat{L}}_t,\\
\hat{Z}_{h,t}^{\mathrm{dist}}
&=
\sum_{j=1}^{n}
\sum_{r=1}^{T-1}
\Delta\hat{U}_{j,r}
\hat{\omega}_{h,j}^{\mathrm{dist}}\left(
\frac{r}{T}
\right)
\left[
\frac{r}{T}
-
\mathbf{1}\left\{
\hat{R}_{jt}\leq r
\right\}
\right],\\
\hat{Z}_{h,t}^{\mathrm{imp}}
&=
\sum_{r=1}^{T-1}
\Delta\hat{U}_{k,r}
\hat{\omega}_{h,k}^{\mathrm{imp}}\left(
\rho_{\delta}\left(
\frac{r}{T}
\right)
\right)
\rho_{\delta}'\left(
\frac{r}{T}
\right)
\left[
\frac{r}{T}
-
\mathbf{1}\left\{
\hat{R}_{kt}\leq r
\right\}
\right],\\
\hat{Z}_{h,t}
&=
\hat{Z}_{h,t}^{\mathrm{tr}}
+
\hat{Z}_{h,t}^{\mathrm{res}}
+
\hat{Z}_{h,t}^{\mathrm{dist}}
+
\hat{Z}_{h,t}^{\mathrm{imp}}.
\end{aligned}
\label{eq:estimated-influence-contributions}
\end{equation}
The two spacing sums evaluate the integrals in equation~\eqref{eq:density-free-influence} over the intervals between adjacent residual order statistics.

For a fixed collection of \(M\) horizons or response variables, apply equation~\eqref{eq:estimated-influence-contributions} to each response and stack the results in
\[
\boldsymbol{\hat{Z}}_t
=
\left(
\hat{Z}_{1,t},
\ldots,
\hat{Z}_{M,t}
\right)^{\prime}.
\]
Let
\[
\overline{\boldsymbol{\hat{Z}}}
=
\frac{1}{T}
\sum_{t=1}^{T}
\boldsymbol{\hat{Z}}_t,
\qquad
\boldsymbol{\tilde{Z}}_t
=
\boldsymbol{\hat{Z}}_t
-
\overline{\boldsymbol{\hat{Z}}},
\]
and define
\[
\boldsymbol{\hat{\Gamma}}_{\ell}
=
\frac{1}{T}
\sum_{t=\ell+1}^{T}
\boldsymbol{\tilde{Z}}_t
\boldsymbol{\tilde{Z}}_{t-\ell}^{\prime}.
\]
A feasible long-run covariance estimator is
\begin{equation}
\boldsymbol{\hat{\Omega}}
=
\boldsymbol{\hat{\Gamma}}_0
+
\sum_{\ell=1}^{q_T}
\mathcal{K}\left(
\frac{\ell}{q_T+1}
\right)
\left(
\boldsymbol{\hat{\Gamma}}_{\ell}
+
\boldsymbol{\hat{\Gamma}}_{\ell}^{\prime}
\right),
\label{eq:feasible-long-run-covariance}
\end{equation}
where \(\mathcal{K}\) is a standard HAC weighting function and \(q_T\) satisfies the corresponding lag-truncation conditions. When \(\left\{\boldsymbol{Z}_t\right\}\) is a martingale difference sequence, set \(q_T=0\), and equation~\eqref{eq:feasible-long-run-covariance} reduces to the sample covariance.

\begin{corollary}[Feasible Wald inference]
\label{cor:feasible-wald-inference}

Suppose the conditions of Corollary~\ref{cor:main-gaussian-limit} hold, the nuisance estimates above are consistent, and the numerical integration error is \(o_{\mathbb{P}}\left(1\right)\). Then
\[
\boldsymbol{\hat{\Omega}}
\xrightarrow{\mathbb{P}}
\boldsymbol{\Omega}.
\]
For response \(m\), an asymptotic pointwise \(1-\alpha\) confidence interval is
\begin{equation}
\mathcal{I}_{m,1-\alpha}^{\mathrm{pt}}
=
\left[
\hat{\psi}_{m,S}^{\mathrm{E}}
-
z_{1-\alpha/2}
\sqrt{
\frac{\hat{\Omega}_{mm}}{T}
},
\quad
\hat{\psi}_{m,S}^{\mathrm{E}}
+
z_{1-\alpha/2}
\sqrt{
\frac{\hat{\Omega}_{mm}}{T}
}
\right].
\label{eq:pointwise-wald-interval}
\end{equation}

For simultaneous inference over the fixed collection, let
\[
\boldsymbol{\hat{V}}
=
\diag\left(
\hat{\Omega}_{11},
\ldots,
\hat{\Omega}_{MM}
\right),
\qquad
\boldsymbol{\hat{C}}
=
\boldsymbol{\hat{V}}^{-1/2}
\boldsymbol{\hat{\Omega}}
\boldsymbol{\hat{V}}^{-1/2},
\]
and let \(c_{1-\alpha}\left(\boldsymbol{\hat{C}}\right)\) be the \(1-\alpha\) quantile of
\[
\max_{1\leq m\leq M}
\left|G_m\right|
\]
for
\[
\boldsymbol{G}
\sim
\mathcal{N}\left(
\boldsymbol{0},
\boldsymbol{\hat{C}}
\right).
\]
Simultaneous intervals are
\begin{equation}
\mathcal{I}_{m,1-\alpha}^{\mathrm{sim}}
=
\left[
\hat{\psi}_{m,S}^{\mathrm{E}}
-
c_{1-\alpha}\left(
\boldsymbol{\hat{C}}
\right)
\sqrt{
\frac{\hat{\Omega}_{mm}}{T}
},
\quad
\hat{\psi}_{m,S}^{\mathrm{E}}
+
c_{1-\alpha}\left(
\boldsymbol{\hat{C}}
\right)
\sqrt{
\frac{\hat{\Omega}_{mm}}{T}
}
\right],
\qquad
m=1,\ldots,M.
\label{eq:simultaneous-wald-intervals}
\end{equation}
\end{corollary}

The first-stage influence contributions and residual-map Jacobians are computed analytically from the transition estimator and structural model. Automatic differentiation is applied only to the smooth recursive paths; the order statistics are handled by equation~\eqref{eq:estimated-influence-contributions}. Expectations over rank paths are numerical integrals. Thus, the influence-based procedure requires no innovation-density estimator or quantile-smoothing bandwidth. Under the general weak-dependence formulation, the HAC lag choice concerns serial covariance in the observation-level contributions.

The joint result applies to a fixed number of horizons and response variables. Growing-horizon inference is outside the present claims. The intervals above also impose \(S/T\rightarrow\infty\); Section~\ref{sec:simulation-error} adds simulation uncertainty when \(S\) is proportional to \(T\).

\section{What smoothing changes under the same target}
\label{sec:smoothing}

The comparison in this section keeps the transition estimate and the simulated rank paths fixed. Thus, the paired difference between the two estimators is generated only by replacing \(\boldsymbol{\hat{Q}}^{\mathrm{E}}\) with \(\boldsymbol{\hat{Q}}_{b_T}^{\mathrm{S}}=\mathcal{S}_{b_T}\boldsymbol{\hat{Q}}^{\mathrm{E}}\), applied componentwise, in equation~\eqref{eq:unified-estimator}. For each innovation component,
\begin{equation}
\hat{Q}_{j,b_T}^{\mathrm{S}}-\hat{Q}_{j}^{\mathrm{E}}
=
\left(\mathcal{S}_{b_T}-\operatorname{Id}\right)Q_{j0}
+
\left(\mathcal{S}_{b_T}-\operatorname{Id}\right)
\left(\hat{Q}_{j}^{\mathrm{E}}-Q_{j0}\right).
\label{eq:quantile-smoothing-decomposition}
\end{equation}
The first term is the deterministic approximation error introduced by smoothing. The second records how smoothing changes the sampling error of the empirical quantile function.

\begin{assumption}[Rank smoothing]
\label{ass:smoothing}

The operator \(\mathcal{S}_{b}\) is a componentwise boundary-corrected Gaussian rank smoother with \(b_T\rightarrow0\). On interior ranks, its population version satisfies
\[
\left(\mathcal{S}_{b}q\right)\left(p\right)
=
\int_{\mathbb{R}}
\phi\left(v\right)
q\left(p+bv\right)
\,\mathrm{d}v,
\]
where \(\phi\) is the standard normal density. For each innovation component, \(Q_{j0}\) is twice continuously differentiable over the ranks used by the response, and
\[
\left\|
\mathcal{S}_{b}Q_{j0}
-
Q_{j0}
-
\frac{b^2}{2}Q_{j0}^{\prime\prime}
\right\|_j
=
o\left(b^2\right)
\]
under the weighted norm in Assumption~\ref{ass:first-order}. Moreover,
\[
\sqrt{T}
\left\|
\left(\mathcal{S}_{b_T}-\operatorname{Id}\right)
\left(\hat{Q}_{j}^{\mathrm{E}}-Q_{j0}\right)
\right\|_j
=
o_{\mathbb{P}}\left(1\right),
\qquad
j=1,\ldots,n.
\]
The boundary correction and tail conditions make the same expansion valid for the ranks reached by \(\tau_{\delta}\).
\end{assumption}

Assumption~\ref{ass:smoothing} separates the local smoothing bias from the first-order empirical-quantile fluctuation. To express its effect on the impulse-response function, use the propagation weights in equation~\eqref{eq:propagation-weights} and define
\begin{equation}
B_h^{\mathrm{S}}
=
\frac{1}{2}
\left[
\sum_{j=1}^{n}
\int_{0}^{1}
\omega_{h,j}^{\mathrm{dist}}\left(p\right)
Q_{j0}^{\prime\prime}\left(p\right)
\,\mathrm{d}p
+
\int_{0}^{1}
\omega_{h,k}^{\mathrm{imp}}\left(p\right)
Q_{k0}^{\prime\prime}\left(\tau_{\delta}\left(p\right)\right)
\,\mathrm{d}p
\right].
\label{eq:smoothing-bias-coefficient}
\end{equation}
The first integral covers the unshifted impact innovations and the common future innovations. The second covers the rank-shifted impact innovation.

\begin{theorem}[Paired smoothing expansion]
\label{thm:smoothing-expansion}

Suppose Assumptions~\ref{ass:first-order} and~\ref{ass:smoothing} hold, and let \(S\rightarrow\infty\). At any fixed horizon,
\begin{equation}
\hat{\psi}_{h,S}^{\mathrm{S}}
-
\hat{\psi}_{h,S}^{\mathrm{E}}
=
b_T^2B_h^{\mathrm{S}}
+
R_{h,T,S}^{\mathrm{S}},
\label{eq:smoothing-expansion}
\end{equation}
where
\begin{equation}
R_{h,T,S}^{\mathrm{S}}
=
o_{\mathbb{P}}\left(T^{-1/2}+b_T^2\right)
+
O_{\mathbb{P}}\left(
S^{-1/2}
\left(T^{-1/2}+b_T^2\right)
\right).
\label{eq:smoothing-remainder}
\end{equation}
The result holds jointly for any fixed collection of horizons, initial states, response variables, shocked components, and shock sizes after stacking the corresponding coefficients \(B_h^{\mathrm{S}}\).
\end{theorem}

The first term in equation~\eqref{eq:smoothing-remainder} contains the stochastic part of equation~\eqref{eq:quantile-smoothing-decomposition} and the higher-order terms from the nonlinear recursion. The second is the conditional Monte Carlo error in the paired estimator difference. Common simulation draws make this error proportional to the distance between the two quantile estimates. Consequently, \(S\rightarrow\infty\) is sufficient for the paired expansion, although inference for either estimator without a simulation correction continues to require \(S/T\rightarrow\infty\).

Because both procedures use the same \(\boldsymbol{\hat{\beta}}\) and the same generated residuals, there is no separate direct first-stage term in equation~\eqref{eq:smoothing-expansion}. The common transition and residual-quantile uncertainty remains in the sampling variance of each estimator; smoothing changes their leading difference through \(b_T^2B_h^{\mathrm{S}}\).

Equation~\eqref{eq:smoothing-bias-coefficient} also distinguishes pointwise quantile approximation from its effect after recursive propagation. The local error is weighted by the sensitivity of the date-\(h\) response to an innovation at that rank. When the relevant norms are finite,
\[
\left|B_h^{\mathrm{S}}\right|
\leq
\frac{1}{2}
\left[
\sum_{j=1}^{n}
\left\|\omega_{h,j}^{\mathrm{dist}}\right\|_1
\left\|Q_{j0}^{\prime\prime}\right\|_{\infty}
+
\left\|\omega_{h,k}^{\mathrm{imp}}\right\|_1
\left\|Q_{k0}^{\prime\prime}\right\|_{\infty}
\right].
\]
Thus, a small local error can matter when the recursion is persistent or strongly state dependent, while errors at different ranks may offset one another. If \(B_h^{\mathrm{S}}=0\), the second-order bias cancels after propagation and the next nonzero term determines the relevant bandwidth condition.

\begin{corollary}[Bandwidth regimes]
\label{cor:smoothing-regimes}

Let \(\lambda_T=\sqrt{T}b_T^2\), let \(\boldsymbol{B}^{\mathrm{S}}\) stack the coefficients in equation~\eqref{eq:smoothing-bias-coefficient} for a fixed collection of responses, and suppose \(S/T\rightarrow\infty\).
\begin{enumerate}[label=(\roman*),leftmargin=2.2em]
\item If \(\lambda_T\rightarrow0\), then
\[
\sqrt{T}
\left(
\boldsymbol{\hat{\psi}}_{S}^{\mathrm{S}}
-
\boldsymbol{\hat{\psi}}_{S}^{\mathrm{E}}
\right)
\xrightarrow{\mathbb{P}}
\boldsymbol{0}.
\]
The empirical and smoothed estimators have the same first-order distribution and the same feasible covariance matrix.

\item If \(\lambda_T\rightarrow\lambda\in\left(0,\infty\right)\), then
\begin{equation}
\sqrt{T}
\left(
\boldsymbol{\hat{\psi}}_{S}^{\mathrm{S}}
-
\boldsymbol{\psi}
\right)
\xrightarrow{\mathrm{d}}
\mathcal{N}
\left(
\lambda\boldsymbol{B}^{\mathrm{S}},
\boldsymbol{\Omega}
\right).
\label{eq:smoothed-local-bias-limit}
\end{equation}
The smoothed estimator has the same first-order covariance as the empirical estimator, but its limiting distribution is shifted by the propagated smoothing bias.

\item If \(\lambda_T\rightarrow\infty\) and \(B_h^{\mathrm{S}}\neq0\), then
\begin{equation}
b_T^{-2}
\left(
\hat{\psi}_{h,S}^{\mathrm{S}}
-
\psi_h
\right)
\xrightarrow{\mathbb{P}}
B_h^{\mathrm{S}}.
\label{eq:smoothing-bias-dominates}
\end{equation}
The smoothed estimator remains consistent because \(b_T\rightarrow0\), but smoothing bias dominates its \(T^{-1/2}\) sampling error.
\end{enumerate}
\end{corollary}

For the second-order smoother used here, first-order equivalence requires \(b_T=o\left(T^{-1/4}\right)\). A bandwidth satisfying \(b_T\sim cT^{-1/4}\) produces the mean shift in equation~\eqref{eq:smoothed-local-bias-limit} with \(\lambda=c^2\), while \(T^{-1/4}=o\left(b_T\right)\) produces the bias-dominated behavior in equation~\eqref{eq:smoothing-bias-dominates}. For an order-\(r\) smoother, the same argument replaces \(b_T^2\) with \(b_T^r\), so the boundary between negligible and first-order bias is \(T^{-1/(2r)}\).

For one response, let \(\sigma_h^2\) be the corresponding diagonal element of \(\boldsymbol{\Omega}\). Under the local-bias regime, a nominal \(1-\alpha\) Wald interval centered at \(\hat{\psi}_{h,S}^{\mathrm{S}}\) and using the standard error from Section~\ref{sec:main-theory-feasible} has limiting coverage
\begin{equation}
\Phi\left(
z_{1-\alpha/2}
-
\frac{\lambda B_h^{\mathrm{S}}}{\sigma_h}
\right)
-
\Phi\left(
-z_{1-\alpha/2}
-
\frac{\lambda B_h^{\mathrm{S}}}{\sigma_h}
\right).
\label{eq:smoothed-wald-coverage}
\end{equation}
This equals \(1-\alpha\) when the propagated bias vanishes and is smaller otherwise. In the bias-dominated regime, the same uncorrected interval has limiting coverage zero. Bias correction is possible if \(b_T^2B_h^{\mathrm{S}}\) can be estimated with error \(o_{\mathbb{P}}\left(T^{-1/2}\right)\), but it requires additional smoothness estimation that the empirical estimator does not use.

\begin{corollary}[Additive shocks and affine transitions]
\label{cor:smoothing-special-cases}

The following cases verify the two channels in equation~\eqref{eq:smoothing-bias-coefficient}.
\begin{enumerate}[label=(\roman*),leftmargin=2.2em]
\item Suppose the impact intervention adds a fixed \(\xi\boldsymbol{e}_k\) to the structural innovation. Write \(\Lambda_{h,s,j}^{\xi}\) for the path derivative along the shocked path. The leading smoothing coefficient becomes
\[
\frac{1}{2}
\sum_{j=1}^{n}
\int_{0}^{1}
\left[
\mathbb{E}\left[
\Lambda_{h,1,j}^{\xi}-\Lambda_{h,1,j}^{0}
\mathrel{\big|}
P_{j1}=p
\right]
+
\sum_{s=2}^{h}
\mathbb{E}\left[
\Lambda_{h,s,j}^{\xi}-\Lambda_{h,s,j}^{0}
\mathrel{\big|}
P_{js}=p
\right]
\right]
Q_{j0}^{\prime\prime}\left(p\right)
\,\mathrm{d}p.
\]
There is no separate shifted-impact term because the amount added in structural innovation units does not depend on an estimated shifted quantile.

\item Suppose
\[
g\left(\boldsymbol{y},\boldsymbol{u};\boldsymbol{\beta}_0\right)
=
\boldsymbol{c}_0
+
\boldsymbol{A}_0\boldsymbol{y}
+
\boldsymbol{B}_0\boldsymbol{u}.
\]
Under the rank-normalized shock,
\begin{equation}
B_h^{\mathrm{S}}
=
\frac{1}{2}
\boldsymbol{a}^{\prime}
\boldsymbol{A}_0^{h-1}
\boldsymbol{B}_0
\boldsymbol{e}_k
\int_{0}^{1}
\left[
Q_{k0}^{\prime\prime}\left(\tau_{\delta}\left(p\right)\right)
-
Q_{k0}^{\prime\prime}\left(p\right)
\right]
\,\mathrm{d}p.
\label{eq:affine-smoothing-bias}
\end{equation}
Common future innovations cancel from the path difference, so smoothing matters only through the impact innovation. Under a fixed additive shock, the affine response does not depend on \(\boldsymbol{Q}_0\), and the empirical and smoothed estimators coincide exactly when they use the same transition estimate.
\end{enumerate}
\end{corollary}

Alternative rank maps and counterfactual future innovation distributions define different population responses rather than alternative estimators of equation~\eqref{eq:population-irf}. They are therefore outside the same-target smoothing comparison.

\section{Full recursive bootstrap and numerical integration}
\label{sec:inference}

\subsection{Full-model re-estimation}
\label{sec:bootstrap}

The residual bootstrap must reproduce every estimated object that enters the recursive impulse response. Throughout this section, \(\boldsymbol{\hat{Q}}^{\mathrm{E}}\) includes any location or scale normalization used to identify the structural innovations. When no finite-sample adjustment is required, it is the empirical quantile function defined in Section~\ref{sec:estimators}.

\begin{algorithm}[t]
\small
\caption{Full recursive residual bootstrap}
\label{alg:recursive-residual-bootstrap}
\KwIn{The sample \(\left\{\boldsymbol{Y}_t\right\}_{t=0}^{T}\), the fitted transition \(\boldsymbol{\hat{\beta}}\), the residual quantiles \(\boldsymbol{\hat{Q}}^{\mathrm{E}}\), the response specification \(\left(\boldsymbol{y},\boldsymbol{a},k,\delta,h\right)\), the response-simulation ranks \(\left\{\boldsymbol{P}_{s,1:h}\right\}_{s=1}^{S}\), a burn-in length \(\ell_T\), and \(B\) bootstrap replications.}
\For{\(b=1,\ldots,B\)}{
Draw \(W_{jt,b}^{*}\stackrel{\mathrm{iid}}{\sim}\operatorname{Unif}\left(0,1\right)\) independently over \(j=1,\ldots,n\) and \(t=1-\ell_T,\ldots,T\), and set \(\boldsymbol{U}_{t,b}^{*}=\boldsymbol{\hat{Q}}^{\mathrm{E}}\left(\boldsymbol{W}_{t,b}^{*}\right)\).\;
Set \(\boldsymbol{Y}_{-\ell_T,b}^{*}=\boldsymbol{Y}_0\) and generate
\[
\boldsymbol{Y}_{t,b}^{*}
=
g\left(
\boldsymbol{Y}_{t-1,b}^{*},
\boldsymbol{U}_{t,b}^{*};
\boldsymbol{\hat{\beta}}
\right)
\]
recursively for \(t=1-\ell_T,\ldots,T\); retain \(\left\{\boldsymbol{Y}_{t,b}^{*}\right\}_{t=0}^{T}\).\;
Re-estimate the transition by the original procedure to obtain \(\boldsymbol{\hat{\beta}}_{b}^{*}\).\;
Compute
\[
\boldsymbol{\hat{U}}_{t,b}^{*}
=
G\left(
\boldsymbol{Y}_{t,b}^{*},
\boldsymbol{Y}_{t-1,b}^{*};
\boldsymbol{\hat{\beta}}_{b}^{*}
\right)
\]
and reconstruct \(\boldsymbol{\hat{Q}}_{b}^{\mathrm{E},*}\) from the bootstrap residual order statistics.\;
Using the original ranks \(\left\{\boldsymbol{P}_{s,1:h}\right\}_{s=1}^{S}\), compute \(\hat{\psi}_{h,S,b}^{\mathrm{E},*}\) from equation~\eqref{eq:unified-estimator}; when studentization is used, recompute \(\boldsymbol{\hat{\Omega}}_{b}^{*}\) by the procedure in Section~\ref{sec:main-theory-feasible}.\;
}
\KwOut{The bootstrap estimates \(\left\{\hat{\psi}_{h,S,b}^{\mathrm{E},*},\boldsymbol{\hat{\Omega}}_{b}^{*}\right\}_{b=1}^{B}\).}
\end{algorithm}

The ranks used to generate each bootstrap time series are redrawn in every replication. In contrast, the response-simulation ranks are held fixed across the original and bootstrap estimates. Holding them fixed removes avoidable numerical noise from the comparison when \(S/T\rightarrow\infty\). Because the maintained model assumes independence across structural components, the algorithm draws the component ranks independently. Contemporaneously dependent innovations would instead require resampling from an estimated joint distribution.

Let \(\mathcal{F}_T\) contain the observed sample and the fixed response-simulation ranks, and let \(\Pr^*\) and \(\mathbb{E}^*\) denote probability and expectation conditional on \(\mathcal{F}_T\).

\begin{assumption}[Bootstrap regularity]
\label{ass:bootstrap}

The following conditions hold.

\begin{enumerate}[label=(\roman*),leftmargin=2.2em]
\item The stability, differentiability, and moment conditions in Assumption~\ref{ass:first-order} hold uniformly over a neighborhood of \(\boldsymbol{\beta}_0\) and over marginal innovation laws in a neighborhood of \(\boldsymbol{Q}_0\). The same location and scale normalizations are imposed in the original and bootstrap samples. The initialization and burn-in convention in Algorithm~\ref{alg:recursive-residual-bootstrap} affect the bootstrap estimator by \(o_{\mathbb{P}^*}\left(T^{-1/2}\right)\) in probability.

\item Conditional on \(\mathcal{F}_T\), the re-estimated transition parameter satisfies
\[
\sqrt{T}
\left(
\boldsymbol{\hat{\beta}}^*
-
\boldsymbol{\hat{\beta}}
\right)
=
\frac{1}{\sqrt{T}}
\sum_{t=1}^{T}
\boldsymbol{L}_t^*
+
o_{\mathbb{P}^*}\left(1\right),
\]
and its conditional law converges in probability to the same Gaussian limit as the right-hand side of equation~\eqref{eq:beta-linearization}.

\item Jointly with part~\textup{(ii)}, the empirical quantile process formed from the re-estimated bootstrap residuals satisfies the conditional counterpart of equation~\eqref{eq:generated-quantile-expansion}:
\[
\left\|
\sqrt{T}
\left(
\hat{Q}_{j}^{\mathrm{E},*}
-
\hat{Q}_{j}^{\mathrm{E}}
\right)
-
\frac{1}{\sqrt{T}}
\sum_{t=1}^{T}
\left[
\Xi_{j,t}^*
+
\boldsymbol{r}_j^*\left(\cdot\right)^{\prime}
\boldsymbol{L}_t^*
\right]
\right\|_j
=
o_{\mathbb{P}^*}\left(1\right)
\]
for \(j=1,\ldots,n\). The stacked bootstrap influence array satisfies the corresponding conditional central limit theorem, and its long-run covariance converges to \(\boldsymbol{\Omega}\).

\item The number of response specifications is fixed and \(S/T\rightarrow\infty\).
\end{enumerate}
\end{assumption}

Primitive sufficient conditions are stated in Appendix~\ref{app:assumptions}, and the proof is given in Appendix~\ref{app:proofs-inference}. The key requirement is joint reproduction of the first-stage expansion and the residual empirical process.

\begin{theorem}[Validity of the full recursive residual bootstrap]
\label{thm:bootstrap-validity}

Suppose Assumptions~\ref{ass:first-order} and~\ref{ass:bootstrap} hold. For any fixed collection of response specifications,
\begin{equation}
\begin{aligned}
\sqrt{T}
\left(
\boldsymbol{\hat{\psi}}_{S}^{\mathrm{E},*}
-
\boldsymbol{\hat{\psi}}_{S}^{\mathrm{E}}
\right)
&=
\frac{1}{\sqrt{T}}
\sum_{t=1}^{T}
\boldsymbol{Z}_t^*
+
o_{\mathbb{P}^*}\left(1\right),\\
\boldsymbol{Z}_t^*
&=
\boldsymbol{Z}_t^{\mathrm{tr},*}
+
\boldsymbol{Z}_t^{\mathrm{res},*}
+
\boldsymbol{Z}_t^{\mathrm{dist},*}
+
\boldsymbol{Z}_t^{\mathrm{imp},*},
\end{aligned}
\label{eq:bootstrap-linearization}
\end{equation}
where the four terms are the bootstrap analogues of the contributions in equation~\eqref{eq:influence-decomposition}. Conditional on \(\mathcal{F}_T\),
\begin{equation}
\sqrt{T}
\left(
\boldsymbol{\hat{\psi}}_{S}^{\mathrm{E},*}
-
\boldsymbol{\hat{\psi}}_{S}^{\mathrm{E}}
\right)
\xrightarrow{\mathrm{d}^*}
\mathcal{N}\left(
\boldsymbol{0},
\boldsymbol{\Omega}
\right)
\quad
\text{in probability}.
\label{eq:bootstrap-conditional-limit}
\end{equation}

For one response with \(\Omega_{hh}>0\), let
\[
\hat{\sigma}_h^2
=
\hat{\Omega}_{hh},
\qquad
\left(
\hat{\sigma}_h^*
\right)^2
=
\hat{\Omega}_{hh}^*,
\]
and define
\begin{equation}
R_{h,T}
=
\frac{
\sqrt{T}
\left(
\hat{\psi}_{h,S}^{\mathrm{E}}
-
\psi_h
\right)
}{
\hat{\sigma}_h
},
\qquad
R_{h,T}^*
=
\frac{
\sqrt{T}
\left(
\hat{\psi}_{h,S}^{\mathrm{E},*}
-
\hat{\psi}_{h,S}^{\mathrm{E}}
\right)
}{
\hat{\sigma}_h^*
}.
\label{eq:bootstrap-studentized-root}
\end{equation}
Then
\[
\sup_{x\in\mathbb{R}}
\left|
\Pr^*\left(
R_{h,T}^*
\leq x
\right)
-
\Pr\left(
R_{h,T}
\leq x
\right)
\right|
\xrightarrow{\mathbb{P}}
0.
\]
If every marginal variance is positive, the same conclusion holds for the maximum absolute studentized statistic over any fixed collection of responses.
\end{theorem}

Re-estimation of the transition generates the direct parameter term and the generated-residual term in equation~\eqref{eq:bootstrap-linearization}. Reconstructing the residual quantiles generates the future-distribution and shifted-impact terms. Consequently, the proof establishes a conditional version of the complete expansion in Theorem~\ref{thm:main-linearization}, followed by the bootstrap functional delta method for the finite-horizon recursive map \citep{VanDerVaartWellner2023,BeutnerZaehle2016}. Bootstrap quantiles of \(R_{h,T}^*\) therefore yield valid percentile-\(t\) intervals, and the fixed-dimensional maximum statistic yields simultaneous intervals.

We next formalize what is learned when only the response-simulation paths are redrawn. Define the fitted response with exact integration by
\[
\hat{\psi}_h^{\infty}
=
\Psi_h\left(
\boldsymbol{\hat{\beta}},
\boldsymbol{\hat{Q}}^{\mathrm{E}}
\right),
\]
and let
\[
\hat{\psi}_{h,S}^{\dagger}
=
\frac{1}{S}
\sum_{s=1}^{S}
D_h\left(
\boldsymbol{P}_{s,1:h}^{\dagger};
\boldsymbol{\hat{\beta}},
\boldsymbol{\hat{Q}}^{\mathrm{E}},
\boldsymbol{y},
\boldsymbol{a},
k,
\delta
\right),
\]
where the ranks \(\boldsymbol{P}_{s,1:h}^{\dagger}\) are newly drawn while the fitted transition and residual quantiles remain fixed. Let \(\Pr^{\dagger}\) denote probability over these ranks conditional on \(\mathcal{F}_T\).

\begin{proposition}[Path-only resampling]
\label{prop:path-only-resampling}

Suppose Assumption~\ref{ass:first-order} holds and the fitted path response has a finite conditional second moment. Conditional on \(\mathcal{F}_T\),
\begin{equation}
\sqrt{S}
\left(
\hat{\psi}_{h,S}^{\dagger}
-
\hat{\psi}_h^{\infty}
\right)
\xrightarrow{\mathrm{d}^{\dagger}}
\mathcal{N}\left(
0,
\mathcal{V}_h
\right)
\quad
\text{in probability},
\label{eq:path-only-limit}
\end{equation}
where
\[
\mathcal{V}_h
=
\operatorname{Var}
\left[
D_h\left(
\boldsymbol{P}_{1:h};
\boldsymbol{\beta}_0,
\boldsymbol{Q}_0,
\boldsymbol{y},
\boldsymbol{a},
k,
\delta
\right)
\right].
\]
Consequently, if \(S/T\rightarrow\infty\), then
\[
\sqrt{T}
\left(
\hat{\psi}_{h,S}^{\dagger}
-
\hat{\psi}_h^{\infty}
\right)
\xrightarrow{\mathbb{P}^{\dagger}}
0.
\]
If \(S/T\rightarrow\kappa\in\left(0,\infty\right)\), then the same quantity converges conditionally to
\[
\mathcal{N}\left(
0,
\frac{\mathcal{V}_h}{\kappa}
\right).
\]
Neither limit contains the sampling variance \(\Omega_{hh}\).
\end{proposition}

Accordingly, path-only intervals quantify numerical integration error conditional on the fitted transition and residual distribution. They omit transition estimation, generated residuals, and estimation of both uses of the innovation quantiles. When \(S/T\rightarrow\infty\), such intervals collapse even though the estimator retains sampling uncertainty of order \(T^{-1/2}\). When \(S/T\) has a finite limit, they recover only the additional simulation component derived in Section~\ref{sec:simulation-error}. Centering path-only draws at the original finite-\(S\) estimate does not restore sampling variation; it combines the numerical errors from two simulated averages. Path-only resampling is therefore a Monte Carlo diagnostic, while sampling intervals require the influence-function procedure or Algorithm~\ref{alg:recursive-residual-bootstrap}.

\subsection{A finite number of simulated paths}
\label{sec:simulation-error}

The estimator in equation~\eqref{eq:unified-estimator} combines sampling uncertainty with numerical error from approximating the population expectation by \(S\) simulated paths. The preceding results impose \(S/T\rightarrow\infty\), which makes the second component negligible. We now allow \(S\) to grow at the same rate as \(T\).

Let \(m=1,\ldots,M\) index a fixed collection of response specifications, let \(H=\max_{1\leq m\leq M}h_m\), and use the same master rank path \(\boldsymbol{P}_{s,1:H}\) to evaluate all responses in simulation draw \(s\). Collect the corresponding path responses in
\[
\boldsymbol{D}\left(
\boldsymbol{P}_{1:H};
\boldsymbol{\beta},
\boldsymbol{Q}
\right)
=
\left(
D_1\left(
\boldsymbol{P}_{1:h_1};
\boldsymbol{\beta},
\boldsymbol{Q}
\right),
\ldots,
D_M\left(
\boldsymbol{P}_{1:h_M};
\boldsymbol{\beta},
\boldsymbol{Q}
\right)
\right)^{\prime},
\]
where the fixed initial states, response variables, shocked components, and shock sizes are suppressed. Define
\begin{equation}
\boldsymbol{\Omega}_{\mathrm{MC}}
=
\operatorname{Var}\left[
\boldsymbol{D}\left(
\boldsymbol{P}_{1:H};
\boldsymbol{\beta}_0,
\boldsymbol{Q}_0
\right)
\right].
\label{eq:mc-covariance}
\end{equation}
The off-diagonal elements of \(\boldsymbol{\Omega}_{\mathrm{MC}}\) record the covariance across horizons and response variables induced by the common rank paths.

\begin{theorem}[Sampling and simulation uncertainty]
\label{thm:finite-simulation-limit}

Suppose Assumption~\ref{ass:first-order}\textup{(i)--(iv)} holds, the collection of responses is fixed, and
\[
\mathbb{E}\left[
\left\|
\boldsymbol{D}\left(
\boldsymbol{P}_{1:H};
\boldsymbol{\beta}_0,
\boldsymbol{Q}_0
\right)
\right\|^{2+\eta}
\right]
<
\infty
\]
for some \(\eta>0\). If the simulation draws are independent of the observed sample and
\[
\frac{S}{T}
\rightarrow
\kappa
\in
\left(0,\infty\right),
\]
then
\begin{equation}
\begin{aligned}
\sqrt{T}\left(
\boldsymbol{\hat{\psi}}_{S}^{\mathrm{E}}
-
\boldsymbol{\psi}
\right)
&=
\frac{1}{\sqrt{T}}
\sum_{t=1}^{T}
\boldsymbol{Z}_t\\
&\quad+
\sqrt{\frac{T}{S}}
\frac{1}{\sqrt{S}}
\sum_{s=1}^{S}
\left[
\boldsymbol{D}\left(
\boldsymbol{P}_{s,1:H};
\boldsymbol{\beta}_0,
\boldsymbol{Q}_0
\right)
-
\boldsymbol{\psi}
\right]
+
o_{\mathbb{P}}\left(1\right).
\end{aligned}
\label{eq:finite-simulation-expansion}
\end{equation}
Consequently,
\begin{equation}
\sqrt{T}\left(
\boldsymbol{\hat{\psi}}_{S}^{\mathrm{E}}
-
\boldsymbol{\psi}
\right)
\xrightarrow{\mathrm{d}}
\mathcal{N}\left(
\boldsymbol{0},
\boldsymbol{\Omega}
+
\kappa^{-1}
\boldsymbol{\Omega}_{\mathrm{MC}}
\right).
\label{eq:finite-simulation-limit}
\end{equation}
The sampling and simulation components in equation~\eqref{eq:finite-simulation-expansion} are asymptotically independent.
\end{theorem}

When \(S/T\rightarrow\infty\), the simulation term vanishes and equation~\eqref{eq:finite-simulation-limit} reduces to Corollary~\ref{cor:main-gaussian-limit}. When \(S/T\rightarrow\kappa\in\left(0,\infty\right)\), simulation error contributes at the same order as estimation error. If \(S/T\rightarrow0\) and \(\boldsymbol{\Omega}_{\mathrm{MC}}\) is nonzero, simulation error dominates on the \(\sqrt{T}\) scale, and the leading rate becomes \(\sqrt{S}\).

For each simulated path, define the fitted response vector
\[
\boldsymbol{\hat{D}}_{s}^{\mathrm{E}}
=
\left(
D_m\left(
\boldsymbol{P}_{s,1:h_m};
\boldsymbol{\hat{\beta}},
\boldsymbol{\hat{Q}}^{\mathrm{E}}
\right)
\right)_{m=1}^{M},
\qquad
\boldsymbol{\hat{\psi}}_{S}^{\mathrm{E}}
=
\frac{1}{S}
\sum_{s=1}^{S}
\boldsymbol{\hat{D}}_{s}^{\mathrm{E}}.
\]
The simulation covariance is estimated by
\[
\boldsymbol{\hat{\Omega}}_{\mathrm{MC}}^{\mathrm{E}}
=
\frac{1}{S-1}
\sum_{s=1}^{S}
\left(
\boldsymbol{\hat{D}}_{s}^{\mathrm{E}}
-
\boldsymbol{\hat{\psi}}_{S}^{\mathrm{E}}
\right)
\left(
\boldsymbol{\hat{D}}_{s}^{\mathrm{E}}
-
\boldsymbol{\hat{\psi}}_{S}^{\mathrm{E}}
\right)^{\prime}.
\]
Combining this matrix with the sampling covariance estimator from equation~\eqref{eq:feasible-long-run-covariance} gives
\begin{equation}
\boldsymbol{\hat{\Omega}}_{T,S}^{\mathrm{tot}}
=
\boldsymbol{\hat{\Omega}}
+
\frac{T}{S}
\boldsymbol{\hat{\Omega}}_{\mathrm{MC}}^{\mathrm{E}},
\qquad
\operatorname{Var}\left(
\boldsymbol{\hat{\psi}}_{S}^{\mathrm{E}}
\right)
\approx
\frac{\boldsymbol{\hat{\Omega}}}{T}
+
\frac{\boldsymbol{\hat{\Omega}}_{\mathrm{MC}}^{\mathrm{E}}}{S}.
\label{eq:finite-simulation-correction}
\end{equation}
Thus, the pointwise standard error for response \(m\) is
\[
\operatorname{se}_{m,T,S}
=
\left(
\frac{\hat{\Omega}_{mm}}{T}
+
\frac{\hat{\Omega}_{\mathrm{MC},mm}^{\mathrm{E}}}{S}
\right)^{1/2}.
\]
Pointwise and simultaneous intervals are obtained by replacing \(\boldsymbol{\hat{\Omega}}\) with \(\boldsymbol{\hat{\Omega}}_{T,S}^{\mathrm{tot}}\) in equations~\eqref{eq:pointwise-wald-interval} and~\eqref{eq:simultaneous-wald-intervals}.

Common random numbers are particularly useful when comparing the empirical and smoothed estimators. For \(r,q\in\left\{\mathrm{E},\mathrm{S}\right\}\), define the cross-covariance
\[
\boldsymbol{\hat{\Omega}}_{\mathrm{MC}}^{rq}
=
\frac{1}{S-1}
\sum_{s=1}^{S}
\left(
\boldsymbol{\hat{D}}_{s}^{r}
-
\boldsymbol{\hat{\psi}}_{S}^{r}
\right)
\left(
\boldsymbol{\hat{D}}_{s}^{q}
-
\boldsymbol{\hat{\psi}}_{S}^{q}
\right)^{\prime}.
\]
Let
\[
\boldsymbol{\hat{\Delta}}_{s}
=
\boldsymbol{\hat{D}}_{s}^{\mathrm{S}}
-
\boldsymbol{\hat{D}}_{s}^{\mathrm{E}},
\qquad
\boldsymbol{\hat{\Delta}}_{S}
=
\boldsymbol{\hat{\psi}}_{S}^{\mathrm{S}}
-
\boldsymbol{\hat{\psi}}_{S}^{\mathrm{E}}.
\]
The simulation covariance of the paired estimator difference is
\begin{equation}
\begin{aligned}
\boldsymbol{\hat{\Omega}}_{\mathrm{MC}}^{\Delta}
&=
\boldsymbol{\hat{\Omega}}_{\mathrm{MC}}^{\mathrm{SS}}
+
\boldsymbol{\hat{\Omega}}_{\mathrm{MC}}^{\mathrm{EE}}
-
\boldsymbol{\hat{\Omega}}_{\mathrm{MC}}^{\mathrm{SE}}
-
\boldsymbol{\hat{\Omega}}_{\mathrm{MC}}^{\mathrm{ES}}\\
&=
\frac{1}{S-1}
\sum_{s=1}^{S}
\left(
\boldsymbol{\hat{\Delta}}_{s}
-
\boldsymbol{\hat{\Delta}}_{S}
\right)
\left(
\boldsymbol{\hat{\Delta}}_{s}
-
\boldsymbol{\hat{\Delta}}_{S}
\right)^{\prime}.
\end{aligned}
\label{eq:paired-mc-covariance}
\end{equation}
Its contribution on the \(\sqrt{T}\) scale is
\[
\frac{T}{S}
\boldsymbol{\hat{\Omega}}_{\mathrm{MC}}^{\Delta}.
\]

Under Assumption~\ref{ass:smoothing}, both path-response functions converge to the same population function as \(b_T\rightarrow0\). With common rank paths,
\[
\boldsymbol{\hat{\Omega}}_{\mathrm{MC}}^{\mathrm{EE}},
\boldsymbol{\hat{\Omega}}_{\mathrm{MC}}^{\mathrm{SS}},
\boldsymbol{\hat{\Omega}}_{\mathrm{MC}}^{\mathrm{ES}},
\boldsymbol{\hat{\Omega}}_{\mathrm{MC}}^{\mathrm{SE}}
\xrightarrow{\mathbb{P}}
\boldsymbol{\Omega}_{\mathrm{MC}},
\qquad
\boldsymbol{\hat{\Omega}}_{\mathrm{MC}}^{\Delta}
\xrightarrow{\mathbb{P}}
\boldsymbol{0}.
\]
Consequently, common random numbers remove first-order simulation noise from the smoothing comparison. With independent rank paths, the cross-covariances are zero and the simulation covariance of the difference converges to \(2\boldsymbol{\Omega}_{\mathrm{MC}}\).

For response \(m\), a useful numerical diagnostic is
\[
\hat{\pi}_{m,\mathrm{MC}}
=
\frac{
\hat{\Omega}_{\mathrm{MC},mm}^{\mathrm{E}}/S
}{
\hat{\Omega}_{mm}/T
+
\hat{\Omega}_{\mathrm{MC},mm}^{\mathrm{E}}/S
}.
\]
This quantity reports the fraction of the estimated variance due to numerical integration. In practice, \(S\) can be increased until this fraction is small for the reported responses and the estimates and standard errors are stable. The sufficient asymptotic condition is \(S/T\rightarrow\infty\), but no fixed path count or fixed ratio applies uniformly because simulation variance depends on the horizon, persistence, nonlinearity, initial state, and shock size.

Theorem~\ref{thm:bootstrap-validity} conditions on fixed response-simulation ranks and imposes \(S/T\rightarrow\infty\). When \(S/T\rightarrow\kappa\in\left(0,\infty\right)\), that bootstrap continues to reproduce the sampling component \(\boldsymbol{\Omega}\), while equation~\eqref{eq:finite-simulation-correction} supplies the additional simulation component.

Appendix~\ref{app:implementation} summarizes implementation and reporting conventions supported by the preceding results.

\section{Conclusion}
\label{sec:conclusion}

This paper studies inference for a recursively defined nonlinear impulse response when the structural innovation distributions are estimated from generated residuals. The population response is fixed by a normal-rank impact shock, a common future innovation law, and paired shocked and unshocked paths. The empirical-quantile estimator then changes only the estimate of the innovation quantile functions. This common target is necessary for interpreting differences between empirical and smoothed procedures as estimation effects.

First, the empirical-residual estimator is \(\sqrt{T}\)-asymptotically linear at fixed horizons. Its observation-level influence contribution separates direct transition estimation, the effect of transition estimation on residual order statistics, estimation of ordinary innovation quantiles, and estimation of the shifted impact quantile. After the quantile terms are projected through the recursive response, they can be evaluated from residual ranks and spacings. Feasible inference therefore does not require estimating the innovation density or selecting a smoothing bandwidth.

Next, the same-target smoothing comparison shows that the relevant bias is the quantile-smoothing error after it has been weighted by recursive propagation. Smoothing is first-order irrelevant only when that propagated bias is smaller than \(T^{-1/2}\). At the boundary it shifts the Gaussian limit, and above the boundary it dominates sampling uncertainty. The result distinguishes a local approximation property of the quantile estimator from its effect on the reported nonlinear response.

Finally, valid residual-bootstrap inference requires regeneration and re-estimation of the complete recursive model. Resampling response paths while holding the fitted transition and residual distribution fixed measures numerical integration error, not sampling uncertainty. When the number of simulated paths is proportional to the sample size, numerical integration contributes a separate covariance component; when it grows faster, this component vanishes at the \(\sqrt{T}\) scale. The analysis is deliberately fixed-horizon and maintains independent structural innovation components. Growing horizons, dependent components, and high-dimensional transition estimates require additional arguments beyond the present results.

\appendix

\section{Assumptions and primitive sufficient conditions}
\label{app:assumptions}

This appendix gives primitive conditions for the scalar location--scale model in equation~\eqref{eq:scalar-model} and a high-level formulation for the vector model in equation~\eqref{eq:structural-transition}. The conditions are stated for a fixed finite collection of response specifications. Let \(M\) denote the number of reported responses, let \(H\) be their largest horizon, and let \(\mathcal{D}\) collect their shock sizes. Assumption~\ref{ass:first-order} is a compact version of the conditions below.

\subsection{Weighted control of the innovation quantiles}
\label{app:assumptions-weighted-quantiles}

The rank shift requires quantile estimates over the full interval \(\left(0,1\right)\). For an innovation component \(j\), choose \(\zeta_j\in\left(0,1/2\right)\) and define
\begin{equation}
\vartheta_j\left(p\right)
=
\frac{f_{j0}\left(Q_{j0}\left(p\right)\right)}
{\left[p\left(1-p\right)\right]^{1/2-\zeta_j}},
\qquad
\left\|q\right\|_j
=
\sup_{0<p<1}
\vartheta_j\left(p\right)
\left|q\left(p\right)\right|.
\label{eq:weighted-quantile-norm}
\end{equation}
For a vector of perturbations, set
\[
\left\|\boldsymbol{q}\right\|
=
\max_{1\leq j\leq n}
\left\|q_j\right\|_j.
\]
The weight in equation~\eqref{eq:weighted-quantile-norm} allows quantile errors to increase near the endpoints while controlling the empirical-quantile influence function, since
\[
\vartheta_j\left(p\right)
\Xi_{j,t}\left(p\right)
=
\frac{
p-\mathbf{1}\left\{
U_{jt}\leq Q_{j0}\left(p\right)
\right\}
}{
\left[p\left(1-p\right)\right]^{1/2-\zeta_j}
}.
\]
Weighted empirical-process and quantile-process results under norms of this form are standard; see \citet{VanDerVaart1998}, \citet{Kosorok2008}, and \citet{VanDerVaartWellner2023}.

\begin{assumption}[Continuity of the response in quantile directions]
\label{ass:response-tail-continuity}

For every response specification
\(\left(h,\boldsymbol{y},\boldsymbol{a},k,\delta\right)\),
the propagation weights in equation~\eqref{eq:propagation-weights} satisfy
\begin{equation}
\sum_{j=1}^{n}
\int_{0}^{1}
\frac{
\left|
\omega_{h,j}^{\mathrm{dist}}\left(p\right)
\right|
}{
\vartheta_j\left(p\right)
}
\,\dif p
+
\int_{0}^{1}
\frac{
\left|
\omega_{h,k}^{\mathrm{imp}}
\left(
\rho_{\delta}\left(p\right)
\right)
\right|
\rho_{\delta}'\left(p\right)
}{
\vartheta_k\left(p\right)
}
\,\dif p
<
\infty.
\label{eq:response-tail-integrability}
\end{equation}
The same integrals remain finite when the absolute propagation weights are replaced by the corresponding conditional \(L^{2+\eta}\) norms. Moreover, the generated-residual correction satisfies
\[
\sup_{0<p<1}
\vartheta_j\left(p\right)
\left\|
\boldsymbol{r}_j\left(p\right)
\right\|
<
\infty,
\qquad
j=1,\ldots,n.
\]
\end{assumption}

Assumption~\ref{ass:response-tail-continuity} makes the two derivative maps in equation~\eqref{eq:response-quantile-derivatives} continuous under the norm in equation~\eqref{eq:weighted-quantile-norm}. The second integral includes the Jacobian of the inverse rank shift and therefore imposes the stronger tail requirement associated with the shocked impact innovation.

\subsection{Primitive conditions for the scalar location--scale model}
\label{app:assumptions-scalar}

Write
\[
g\left(y,u;\boldsymbol{\beta}\right)
=
\mu\left(y;\boldsymbol{\beta}\right)
+
\sigma\left(y;\boldsymbol{\beta}\right)u.
\]
The following conditions are sufficient for the stability, differentiability, and first-stage parts of Assumption~\ref{ass:first-order}.

\begin{assumption}[Scalar dynamics]
\label{ass:scalar-dynamics}

There is a compact neighborhood \(\mathcal{B}_0\) of \(\boldsymbol{\beta}_0\) and constants \(\underline{\sigma}>0\), \(C<\infty\), \(m\geq0\), and \(q>2+\eta\) such that the following conditions hold.

\begin{enumerate}[label=(\roman*),leftmargin=2.2em]

\item The functions \(\mu\) and \(\sigma\) have continuous mixed derivatives in \(y\) and \(\boldsymbol{\beta}\) through order two on \(\mathbb{R}\times\mathcal{B}_0\). Their derivatives through order two are bounded by
\(C\left(1+\left|y\right|^m\right)\).
For the fixed maximum horizon \(H\), the finite products of these derivative envelopes along the shocked and unshocked paths have finite \(2+\eta\) moments, and
\[
\inf_{y\in\mathbb{R}}
\inf_{\boldsymbol{\beta}\in\mathcal{B}_0}
\sigma\left(y;\boldsymbol{\beta}\right)
\geq
\underline{\sigma}.
\]

\item The innovations are independent and identically distributed, independent of the history generated by \(\left\{Y_s:s<t\right\}\), and normalized by
\[
\mathbb{E}\left[U_t\right]
=
0,
\qquad
\mathbb{E}\left[U_t^2\right]
=
1.
\]
For every \(\delta\in\mathcal{D}\), if
\(P\sim\operatorname{Unif}\left(0,1\right)\), then
\[
\mathbb{E}\left[
\left|U_t\right|^q
\right]
+
\mathbb{E}\left[
\left|
Q_0\left(
\tau_{\delta}\left(P\right)
\right)
\right|^q
\right]
<
\infty.
\]

\item Define the random Lipschitz coefficient
\[
L\left(u\right)
=
\sup_{\boldsymbol{\beta}\in\mathcal{B}_0}
\sup_{y\neq\tilde{y}}
\frac{
\left|
g\left(y,u;\boldsymbol{\beta}\right)
-
g\left(\tilde{y},u;\boldsymbol{\beta}\right)
\right|
}{
\left|y-\tilde{y}\right|
}.
\]
The contraction and intercept moments satisfy
\begin{equation}
\mathbb{E}\left[
L\left(U_t\right)^q
\right]
<
1,
\qquad
\mathbb{E}\left[
\sup_{\boldsymbol{\beta}\in\mathcal{B}_0}
\left|
g\left(0,U_t;\boldsymbol{\beta}\right)
\right|^q
\right]
<
\infty.
\label{eq:scalar-contraction}
\end{equation}

\item The parameter \(\boldsymbol{\beta}_0\) is the unique value in \(\mathcal{B}_0\) satisfying the population restrictions used by the transition estimator. The location and scale normalizations in part~\textup{(ii)} are imposed by those restrictions; no additional residual recentering or rescaling is applied after estimation.

\end{enumerate}
\end{assumption}

The contraction condition is stronger than necessary, but it is transparent and covers the smooth-transition specifications used in the Monte Carlo analysis. It yields a unique causal stationary solution, finite \(q\)th moments, and geometrically decaying dependence on remote innovations. These properties provide laws of large numbers and central limit theorems for smooth functions of the state; see \citet{Rio2017} for general weak-dependence results. When the state derivatives of \(\mu\) and \(\sigma\) are bounded, the path-derivative requirement in part~\textup{(i)} follows from a \(2+\eta\) moment for the ordinary and shifted innovations. Polynomially growing derivatives require correspondingly higher state and innovation moments. If residuals are separately recentered or rescaled in finite samples, the corresponding location and scale estimates must be added to the first-stage parameter vector and to equation~\eqref{eq:beta-linearization}.

\begin{assumption}[Scalar innovation distribution]
\label{ass:scalar-innovations}

The innovation distribution has a positive continuously differentiable density \(f_0\) on the interior of its support. The density is eventually monotone in each unbounded tail and satisfies
\begin{equation}
\sup_{0<p<1}
p\left(1-p\right)
\frac{
\left|
f_0'\left(
Q_0\left(p\right)
\right)
\right|
}{
f_0\left(
Q_0\left(p\right)
\right)^2
}
<
\infty.
\label{eq:quantile-density-tail-condition}
\end{equation}
For some \(\zeta\in\left(0,1/2\right)\), equation~\eqref{eq:response-tail-integrability} holds with \(n=1\), and the residual Jacobian defined below satisfies
\[
\sup_{0<p<1}
\vartheta_1\left(p\right)
\left\|
\mathbb{E}\left[
\boldsymbol{J}_t
\mathrel{\big|}
U_t=Q_0\left(p\right)
\right]
\right\|
<
\infty.
\]
\end{assumption}

Condition~\eqref{eq:quantile-density-tail-condition} controls the change in the quantile density as the probability rank approaches zero or one. The admissible value of \(\zeta\) depends on both the innovation tails and the growth of the propagation weights. For example, when the propagation weights remain bounded near the endpoints, a Gaussian innovation permits any \(\zeta<1/2\), whereas a Student-\(t\) innovation with \(\nu\) degrees of freedom requires
\(\zeta<1/2-1/\nu\).
If a propagation weight grows at rate
\(\left|Q_0\left(p\right)\right|^m\),
the corresponding Student-\(t\) condition becomes
\[
\zeta
<
\frac{1}{2}
-
\frac{m+1}{\nu}.
\]
A fixed normal-rank shift changes these bounds only by a sub-polynomial tail factor, so the strict inequalities provide the required margin.

\begin{assumption}[Scalar transition estimator]
\label{ass:scalar-estimator}

The estimator \(\boldsymbol{\hat{\beta}}\) satisfies
\[
\frac{1}{T}
\sum_{t=1}^{T}
\boldsymbol{S}_t\left(
\boldsymbol{\hat{\beta}}
\right)
=
o_{\mathbb{P}}\left(
T^{-1/2}
\right).
\]
The population equation
\[
\mathbb{E}\left[
\boldsymbol{S}_t\left(
\boldsymbol{\beta}
\right)
\right]
=
\boldsymbol{0}
\]
identifies \(\boldsymbol{\beta}_0\). The map
\(\boldsymbol{\beta}\mapsto
\boldsymbol{S}_t\left(\boldsymbol{\beta}\right)\)
is continuously differentiable on \(\mathcal{B}_0\), a uniform law of large numbers applies to its derivative, and
\[
\boldsymbol{H}_0
=
\mathbb{E}\left[
\frac{
\partial
\boldsymbol{S}_t\left(
\boldsymbol{\beta}_0
\right)
}{
\partial\boldsymbol{\beta}^{\prime}
}
\right]
\]
is nonsingular. The sequence
\(\left\{
\boldsymbol{S}_t\left(
\boldsymbol{\beta}_0
\right)
\right\}\)
has a finite \(2+\eta\) moment and satisfies a central limit theorem jointly with the residual empirical process used in Appendix~\ref{app:residual-quantiles}.
\end{assumption}

Under Assumption~\ref{ass:scalar-estimator}, equation~\eqref{eq:beta-linearization} holds with
\[
\boldsymbol{L}_t
=
-
\boldsymbol{H}_0^{-1}
\boldsymbol{S}_t\left(
\boldsymbol{\beta}_0
\right).
\]
This formulation covers smooth likelihood, nonlinear least-squares, and regular estimating-equation procedures. Standard arguments for asymptotically linear estimators are given by \citet{VanDerVaart1998}.

\begin{proposition}[Consequences of the scalar primitive conditions]
\label{prop:scalar-primitive-conditions}

Under Assumptions~\ref{ass:scalar-dynamics}--\ref{ass:scalar-estimator}, the process in equation~\eqref{eq:scalar-model} has a unique causal strictly stationary and ergodic solution with a finite \(q\)th moment. The transition and residual maps satisfy the differentiability and moment conditions in Assumption~\ref{ass:first-order}\textup{(i)}, and the transition estimator satisfies Assumption~\ref{ass:first-order}\textup{(ii)}.

At the true parameter, the derivative of the generated residual is
\begin{equation}
\boldsymbol{J}_t
=
\left.
\frac{
\partial
G\left(
Y_t,
Y_{t-1};
\boldsymbol{\beta}
\right)
}{
\partial\boldsymbol{\beta}
}
\right|_{\boldsymbol{\beta}=\boldsymbol{\beta}_0}
=
-
\frac{
\partial
\mu\left(
Y_{t-1};
\boldsymbol{\beta}_0
\right)
/\partial\boldsymbol{\beta}
}{
\sigma\left(
Y_{t-1};
\boldsymbol{\beta}_0
\right)
}
-
U_t
\frac{
\partial
\sigma\left(
Y_{t-1};
\boldsymbol{\beta}_0
\right)
/\partial\boldsymbol{\beta}
}{
\sigma\left(
Y_{t-1};
\boldsymbol{\beta}_0
\right)
}.
\label{eq:scalar-residual-jacobian}
\end{equation}
Because \(U_t\) is independent of \(Y_{t-1}\), the generated-residual correction is
\begin{equation}
\begin{aligned}
\boldsymbol{r}\left(p\right)
&=
\mathbb{E}\left[
\boldsymbol{J}_t
\mathrel{\big|}
U_t=Q_0\left(p\right)
\right]\\
&=
-
\mathbb{E}\left[
\frac{
\partial
\mu\left(
Y_{t-1};
\boldsymbol{\beta}_0
\right)
/\partial\boldsymbol{\beta}
}{
\sigma\left(
Y_{t-1};
\boldsymbol{\beta}_0
\right)
}
\right]\\
&\quad
-
Q_0\left(p\right)
\mathbb{E}\left[
\frac{
\partial
\sigma\left(
Y_{t-1};
\boldsymbol{\beta}_0
\right)
/\partial\boldsymbol{\beta}
}{
\sigma\left(
Y_{t-1};
\boldsymbol{\beta}_0
\right)
}
\right].
\end{aligned}
\label{eq:scalar-generated-residual-correction}
\end{equation}
Assumption~\ref{ass:scalar-innovations} therefore provides the tail control required for this term. Appendix~\ref{app:residual-quantiles} uses equations~\eqref{eq:scalar-residual-jacobian} and~\eqref{eq:scalar-generated-residual-correction} to establish Assumption~\ref{ass:first-order}\textup{(iv)} for this model.
\end{proposition}

\subsection{Central ranks, tails, and optional clipping}
\label{app:assumptions-tails}

The weighted norm permits an exact population response with unbounded innovation support. The proofs nevertheless separate a central rank region from the tails. Let
\[
\pi_{\epsilon}\left(p\right)
=
\min\left\{
1-\epsilon,
\max\left\{
\epsilon,
p
\right\}
\right\},
\qquad
\tau_{\delta,\epsilon}\left(p\right)
=
\pi_{\epsilon}\left(
\tau_{\delta}\left(
\pi_{\epsilon}\left(p\right)
\right)
\right).
\]
Let \(D_h^{\epsilon}\) denote the path response obtained by applying these maps to every ordinary and shifted rank in equation~\eqref{eq:paired-paths}. A sequence \(\epsilon_T\downarrow0\) is admissible when
\begin{equation}
\frac{
T\epsilon_T
}{
\log T
}
\rightarrow
\infty,
\qquad
\sqrt{T}
\max_{1\leq m\leq M}
\mathbb{E}\left[
\left|
D^{\left(m\right),\epsilon_T}
-
D^{\left(m\right)}
\right|
\right]
\rightarrow
0.
\label{eq:tail-negligibility}
\end{equation}
The first condition leaves enough order statistics in the central region for a uniform quantile expansion. The second makes the omitted tail contribution negligible at the scale of the main theorem. If clipping is used in computation, equation~\eqref{eq:tail-negligibility} also ensures that the clipped and exact population responses are first-order equivalent.

For a practical sufficient condition, define
\[
\chi_{\mathcal{D}}\left(
\epsilon
\right)
=
\max_{\delta\in\mathcal{D}}
\left[
\rho_{\delta}\left(
\epsilon
\right)
+
1
-
\rho_{\delta}\left(
1-\epsilon
\right)
\right].
\]
This quantity is the largest probability, over the reported shock sizes, that a shifted rank lies outside
\(\left[\epsilon,1-\epsilon\right]\).
Suppose that, for each reported response, there are \(\kappa>0\) and a random variable \(W_h\) with
\(\mathbb{E}\left[W_h^{1+\kappa}\right]<\infty\)
such that
\[
\left|
D_h^{\epsilon}
-
D_h
\right|
\leq
W_h
\mathbf{1}\left\{
\text{an ordinary or shifted rank leaves }
\left[\epsilon,1-\epsilon\right]
\right\}.
\]
Then H\"older's inequality shows that the second condition in equation~\eqref{eq:tail-negligibility} follows from
\begin{equation}
\sqrt{T}
\left[
\epsilon_T
+
\chi_{\mathcal{D}}\left(
\epsilon_T
\right)
\right]^{\kappa/\left(1+\kappa\right)}
\rightarrow
0.
\label{eq:sufficient-tail-rate}
\end{equation}
For compactly supported innovations with bounded path derivatives, tail control reduces to the probability of the omitted ranks. With unbounded support, the moment of \(W_h\), the quantile growth, and the factor \(\rho_{\delta}'\) determine whether equations~\eqref{eq:tail-negligibility} and~\eqref{eq:sufficient-tail-rate} hold. The shifted-impact tail is generally the more demanding one.

\subsection{High-level conditions for the vector model}
\label{app:assumptions-vector}

The vector theorem does not impose one particular structural transition or first-stage estimator. It uses the following conditions.

\begin{assumption}[General structural model]
\label{ass:vector-high-level}

For the finite response collection introduced above, the following conditions hold.

\begin{enumerate}[label=(\roman*),leftmargin=2.2em]

\item The representation
\(\left(g,G,\boldsymbol{\beta}_0\right)\)
is identified under the stated economic or statistical restrictions. The structural innovations are independent over time and across components, and each component satisfies the normalization used by the structural model.

\item The process
\(\left\{\boldsymbol{Y}_t\right\}\)
is strictly stationary and ergodic. It is strongly mixing with coefficients
\(\alpha\left(\ell\right)\)
satisfying
\[
\sum_{\ell=1}^{\infty}
\alpha\left(\ell\right)^{\eta/\left(2+\eta\right)}
<
\infty,
\]
or it satisfies an alternative dependence condition that yields the same laws of large numbers, empirical-process expansion, and joint central limit theorem.

\item The maps \(g\) and \(G\) have continuous derivatives through order two on neighborhoods of the states, innovations, and parameter values reached by the reported responses. The derivatives of the observed process and of every paired path through horizon \(H\) are bounded by envelopes with finite \(2+\eta\) moments.

\item The transition estimator satisfies equation~\eqref{eq:beta-linearization}. The generated residuals satisfy equation~\eqref{eq:generated-quantile-expansion} under the norm in equation~\eqref{eq:weighted-quantile-norm}, jointly across components and jointly with the transition-estimator expansion.

\item The marginal innovation densities satisfy the smoothness condition in Assumption~\ref{ass:first-order}\textup{(iii)}, Assumption~\ref{ass:response-tail-continuity} holds, and the tail decomposition in equation~\eqref{eq:tail-negligibility} is valid whenever the support is unbounded or clipping is used.

\item The number of response specifications and \(H\) are fixed. The response-simulation draws are independent of the sample. The growth condition on \(S\) is the one stated in the result being applied.

\end{enumerate}
\end{assumption}

Under Assumption~\ref{ass:vector-high-level}, the finite-horizon response map is continuously differentiable in \(\boldsymbol{\beta}\) and in the component quantile functions under the product norm induced by equation~\eqref{eq:weighted-quantile-norm}. Consequently, Assumption~\ref{ass:first-order} follows directly. The high-level generated-residual condition is retained because its primitive form depends on the structural identification and transition estimator. Residual empirical-process expansions for autoregressive and heteroskedastic models provide the relevant model-specific ingredients \citep{Bai1994,KoulStute1999,Koul2002,KoulLing2006}.

\subsection{Why the horizon is fixed}
\label{app:assumptions-fixed-horizon}

For the scalar model, let
\(r\in\left\{0,\delta\right\}\)
index an unshocked or shocked path. Differentiating one step of the recursion gives
\begin{equation}
\begin{aligned}
\frac{
\partial Y_h^r
}{
\partial U_s^r
}
&=
\sigma\left(
Y_{s-1}^r;
\boldsymbol{\beta}_0
\right)
\prod_{\ell=s+1}^{h}
\left[
\frac{
\partial
\mu\left(
Y_{\ell-1}^r;
\boldsymbol{\beta}_0
\right)
}{
\partial y
}
+
\frac{
\partial
\sigma\left(
Y_{\ell-1}^r;
\boldsymbol{\beta}_0
\right)
}{
\partial y
}
U_{\ell}^r
\right],\\
\frac{
\partial Y_j^r
}{
\partial\boldsymbol{\beta}
}
&=
\left[
\frac{
\partial
\mu\left(
Y_{j-1}^r;
\boldsymbol{\beta}_0
\right)
}{
\partial y
}
+
\frac{
\partial
\sigma\left(
Y_{j-1}^r;
\boldsymbol{\beta}_0
\right)
}{
\partial y
}
U_j^r
\right]
\frac{
\partial Y_{j-1}^r
}{
\partial\boldsymbol{\beta}
}\\
&\quad
+
\frac{
\partial
\mu\left(
Y_{j-1}^r;
\boldsymbol{\beta}_0
\right)
}{
\partial\boldsymbol{\beta}
}
+
\frac{
\partial
\sigma\left(
Y_{j-1}^r;
\boldsymbol{\beta}_0
\right)
}{
\partial\boldsymbol{\beta}
}
U_j^r.
\end{aligned}
\label{eq:scalar-recursive-derivatives}
\end{equation}
The empty product in the first line equals one. Because \(h\leq H\) and \(H\) is fixed, equation~\eqref{eq:scalar-recursive-derivatives} contains only finitely many products and sums. The moment requirement in Assumption~\ref{ass:scalar-dynamics} can therefore be chosen once for the largest reported horizon. Allowing \(H\) to increase with \(T\) would require uniform control of these derivative products and of the persistence of the transition, which is outside the results in this paper.

\begin{remark}[Conditions used only by the smoothing and bootstrap results]
\label{rem:additional-theorem-conditions}

The second-order smoothing formula in equation~\eqref{eq:smoothing-bias-coefficient} additionally requires the two propagated integrals involving \(Q_{j0}^{\prime\prime}\) to be finite under the exact boundary correction used by \(\mathcal{S}_{b_T}\). This requirement is automatic on a compact interior rank region but need not follow from interior smoothness when the innovation support is unbounded. If the separate second-derivative integrals are not finite, the smoothing comparison remains valid with the operator-level bias
\[
\dot{\Psi}_h^{\mathrm{dist}}\left[
\left(
\mathcal{S}_{b_T}
-
\operatorname{Id}
\right)
\boldsymbol{Q}_0
\right]
+
\dot{\Psi}_h^{\mathrm{imp}}\left[
\left(
\mathcal{S}_{b_T}
-
\operatorname{Id}
\right)
\boldsymbol{Q}_0
\right]
\]
in place of \(b_T^2B_h^{\mathrm{S}}\), unless a distribution-specific boundary calculation establishes the latter expansion.

The recursive bootstrap requires the scalar contraction, moment, and differentiability conditions uniformly over a neighborhood of \(\boldsymbol{\beta}_0\) and over empirical innovation laws converging to \(F_0\). If
\(\varrho\in\left(0,1\right)\)
is a uniform contraction rate and \(\ell_T\) is the burn-in length, a sufficient initialization condition is
\[
\ell_T
\rightarrow
\infty,
\qquad
\sqrt{T}
\varrho^{\ell_T}
\rightarrow
0.
\]
The original and bootstrap calculations use the same residual normalization, endpoint convention, and clipping rule. The finite-simulation theorem additionally requires a finite \(2+\eta\) moment for the stacked path response, which follows from Assumption~\ref{ass:scalar-dynamics} in the scalar model.
\end{remark}

\section{Generated residuals and empirical quantiles}
\label{app:residual-quantiles}

This appendix derives the expansion in equation~\eqref{eq:generated-quantile-expansion}. The argument has three steps. First, estimating the transition parameter changes the residual empirical distribution. Next, inversion of that distribution gives the empirical-quantile expansion on a central rank region. Finally, the tail conditions in Appendix~\ref{app:assumptions-tails} extend the result to the ordinary and shifted ranks entering the impulse-response function. All statements hold jointly across the fixed number of innovation components.

\subsection{Residual empirical-process expansion}
\label{app:residual-quantiles-empirical-process}

For component \(j\), define the infeasible and residual empirical distribution functions by
\[
F_{j,T}\left(u\right)
=
\frac{1}{T}
\sum_{t=1}^{T}
\mathbf{1}\left\{U_{jt}\leq u\right\},
\qquad
\hat{F}_{j,T}\left(u\right)
=
\frac{1}{T}
\sum_{t=1}^{T}
\mathbf{1}\left\{\hat{U}_{jt}\leq u\right\}.
\]
Let
\[
\boldsymbol{\Delta}_T
=
\sqrt{T}
\left(
\boldsymbol{\hat{\beta}}-\boldsymbol{\beta}_0
\right),
\qquad
\boldsymbol{J}_{j,t}
=
\left.
\frac{\partial G_j\left(\boldsymbol{Y}_t,\boldsymbol{Y}_{t-1};\boldsymbol{\beta}\right)}{\partial\boldsymbol{\beta}}
\right|_{\boldsymbol{\beta}=\boldsymbol{\beta}_0}.
\]
For a function \(a:\mathbb{R}\rightarrow\mathbb{R}\), define the weighted distribution norm
\begin{equation}
\left\|a\right\|_{F_j,\zeta_j}
=
\sup_{u:\,0<F_{j0}\left(u\right)<1}
\frac{\left|a\left(u\right)\right|}{\left[F_{j0}\left(u\right)\left\{1-F_{j0}\left(u\right)\right\}\right]^{1/2-\zeta_j}}.
\label{eq:weighted-distribution-norm}
\end{equation}
The exponent is the same as in equation~\eqref{eq:weighted-quantile-norm}.

\begin{lemma}[Generated-residual empirical process]
\label{lem:generated-residual-empirical-process}
Suppose Assumptions~\ref{ass:scalar-dynamics}--\ref{ass:scalar-estimator} hold in the scalar model. In the vector model, suppose the componentwise residual-indicator classes are stochastically equicontinuous under \(T^{-1/2}\)-local parameter perturbations in the norm in equation~\eqref{eq:weighted-distribution-norm}, and suppose their population means are differentiable at \(\boldsymbol{\beta}_0\). Then, for every component \(j\),
\begin{equation}
\left\|
\sqrt{T}
\left(
\hat{F}_{j,T}-F_{j0}
\right)
-
\frac{1}{\sqrt{T}}
\sum_{t=1}^{T}
\left[
\mathbf{1}\left\{U_{jt}\leq\cdot\right\}
-
F_{j0}\left(\cdot\right)
\right]
-
\boldsymbol{\Gamma}_j\left(\cdot\right)^{\prime}
\boldsymbol{\Delta}_T
\right\|_{F_j,\zeta_j}
=
o_{\mathbb{P}}\left(1\right),
\label{eq:residual-empirical-process-expansion}
\end{equation}
where \(\boldsymbol{\Gamma}_j\) is defined in Assumption~\ref{ass:first-order}\textup{(iv)}. If the conditional expectation below is continuous in \(u\), then
\begin{equation}
\boldsymbol{\Gamma}_j\left(u\right)
=
-f_{j0}\left(u\right)
\mathbb{E}\left[
\boldsymbol{J}_{j,t}
\mathrel{\big|}
U_{jt}=u
\right].
\label{eq:residual-distribution-derivative}
\end{equation}
\end{lemma}

\begin{proof}
For \(u\in\mathbb{R}\) and a generic parameter value \(\boldsymbol{\beta}\), let
\[
m_{j,u,\boldsymbol{\beta},t}
=
\mathbf{1}\left\{
G_j\left(
\boldsymbol{Y}_t,
\boldsymbol{Y}_{t-1};
\boldsymbol{\beta}
\right)
\leq u
\right\}.
\]
Writing \(\mathbb{P}_T\) for the sample average and \(\mathbb{P}_0\) for expectation, decompose
\begin{align*}
\sqrt{T}
\left[
\hat{F}_{j,T}\left(u\right)-F_{j0}\left(u\right)
\right]
&=
\sqrt{T}
\left(
\mathbb{P}_T-\mathbb{P}_0
\right)
m_{j,u,\boldsymbol{\beta}_0,t}\\
&\quad+
\sqrt{T}
\left(
\mathbb{P}_T-\mathbb{P}_0
\right)
\left(
m_{j,u,\boldsymbol{\hat{\beta}},t}
-
m_{j,u,\boldsymbol{\beta}_0,t}
\right)\\
&\quad+
\sqrt{T}
\mathbb{P}_0
\left(
m_{j,u,\boldsymbol{\hat{\beta}},t}
-
m_{j,u,\boldsymbol{\beta}_0,t}
\right).
\end{align*}
The first term is the infeasible innovation empirical process. The weighted stochastic equicontinuity condition stated in the lemma, or its scalar consequence under Assumptions~\ref{ass:scalar-dynamics}--\ref{ass:scalar-estimator}, makes the second term \(o_{\mathbb{P}}\left(1\right)\) uniformly under the norm in equation~\eqref{eq:weighted-distribution-norm}. This is the step supplied by residual empirical-process results for dynamic models; see \citet[Theorem~1]{Bai1994} for ARMA residuals, \citet[Section~7.2]{Koul2002} for autoregressions, and \citet[Theorem~4.1 and Lemma~4.1]{KoulLing2006} for heteroskedastic models. \citet{KoulStute1999} develops related residual-process tools for nonlinear time-series specification.

For the last term, differentiability of the population residual distribution gives, uniformly in \(u\),
\[
\sqrt{T}
\mathbb{P}_0
\left(
m_{j,u,\boldsymbol{\hat{\beta}},t}
-
m_{j,u,\boldsymbol{\beta}_0,t}
\right)
=
\boldsymbol{\Gamma}_j\left(u\right)^{\prime}
\boldsymbol{\Delta}_T
+
o_{\mathbb{P}}\left(1\right).
\]
Combining the three terms proves equation~\eqref{eq:residual-empirical-process-expansion}.

To establish equation~\eqref{eq:residual-distribution-derivative}, fix a deterministic direction \(\boldsymbol{d}\). A first-order expansion of the residual map gives
\[
G_j\left(
\boldsymbol{Y}_t,
\boldsymbol{Y}_{t-1};
\boldsymbol{\beta}_0+s\boldsymbol{d}
\right)
=
U_{jt}
+
s\boldsymbol{J}_{j,t}^{\prime}\boldsymbol{d}
+
o\left(s\right).
\]
Differentiating the probability that this quantity is no larger than \(u\), conditional on \(\boldsymbol{J}_{j,t}\), yields
\[
\left.
\frac{\partial}{\partial s}
\Pr\left[
G_j\left(
\boldsymbol{Y}_t,
\boldsymbol{Y}_{t-1};
\boldsymbol{\beta}_0+s\boldsymbol{d}
\right)
\leq u
\right]
\right|_{s=0}
=
-f_{j0}\left(u\right)
\mathbb{E}\left[
\boldsymbol{J}_{j,t}^{\prime}\boldsymbol{d}
\mathrel{\big|}
U_{jt}=u
\right].
\]
Since the equality holds for every \(\boldsymbol{d}\), equation~\eqref{eq:residual-distribution-derivative} follows.
\end{proof}

Equation~\eqref{eq:residual-empirical-process-expansion} separates the ordinary innovation empirical process from the effect of estimating the transition before forming residual ranks. Substituting equation~\eqref{eq:beta-linearization} gives
\[
\sqrt{T}
\left[
\hat{F}_{j,T}\left(u\right)-F_{j0}\left(u\right)
\right]
=
\frac{1}{\sqrt{T}}
\sum_{t=1}^{T}
\left[
\mathbf{1}\left\{U_{jt}\leq u\right\}
-
F_{j0}\left(u\right)
+
\boldsymbol{\Gamma}_j\left(u\right)^{\prime}
\boldsymbol{L}_t
\right]
+
o_{\mathbb{P}}\left(1\right).
\]
The terms in the sum need not be serially independent because \(\boldsymbol{L}_t\) may depend on lagged states. The joint central limit theorem in Assumption~\ref{ass:first-order} therefore concerns the complete observation-level contribution.

\subsection{Empirical quantiles on central ranks}
\label{app:residual-quantiles-interior}

Let \(\epsilon_T\downarrow0\) satisfy the first condition in equation~\eqref{eq:tail-negligibility}, and define
\[
\mathcal{I}_T
=
\left[
\epsilon_T,
1-\epsilon_T
\right].
\]
The same interval can be used for every component; component-specific sequences give an equivalent result. Recall
\[
\Xi_{j,t}\left(p\right)
=
\frac{
p-
\mathbf{1}\left\{
U_{jt}
\leq
Q_{j0}\left(p\right)
\right\}
}{
f_{j0}\left(Q_{j0}\left(p\right)\right)
},
\qquad
\boldsymbol{r}_j\left(p\right)
=
-
\frac{
\boldsymbol{\Gamma}_j\left(Q_{j0}\left(p\right)\right)
}{
f_{j0}\left(Q_{j0}\left(p\right)\right)
}.
\]

\begin{lemma}[Uniform empirical-quantile expansion]
\label{lem:generated-quantile-interior}
Under the conditions of Lemma~\ref{lem:generated-residual-empirical-process} and the density conditions in Assumption~\ref{ass:first-order}\textup{(iii)},
\begin{equation}
\sup_{p\in\mathcal{I}_T}
\vartheta_j\left(p\right)
\left|
\sqrt{T}
\left[
\hat{Q}_{j}^{\mathrm{E}}\left(p\right)
-
Q_{j0}\left(p\right)
\right]
-
\frac{1}{\sqrt{T}}
\sum_{t=1}^{T}
\left[
\Xi_{j,t}\left(p\right)
+
\boldsymbol{r}_j\left(p\right)^{\prime}
\boldsymbol{L}_t
\right]
\right|
=
o_{\mathbb{P}}\left(1\right).
\label{eq:generated-quantile-interior-expansion}
\end{equation}
Moreover, equation~\eqref{eq:residual-distribution-derivative} implies
\begin{equation}
\boldsymbol{r}_j\left(p\right)
=
\mathbb{E}\left[
\boldsymbol{J}_{j,t}
\mathrel{\big|}
U_{jt}=Q_{j0}\left(p\right)
\right].
\label{eq:generated-quantile-residual-jacobian}
\end{equation}
\end{lemma}

\begin{proof}
The generalized inverse satisfies
\[
\hat{F}_{j,T}\left(
\hat{Q}_{j}^{\mathrm{E}}\left(p\right)-
\right)
\leq
p
\leq
\hat{F}_{j,T}\left(
\hat{Q}_{j}^{\mathrm{E}}\left(p\right)
\right),
\]
and the jump containing \(p\) has size at most \(T^{-1}\). Uniform consistency of \(\hat{F}_{j,T}\) and positivity of \(f_{j0}\) on the central region imply uniform consistency of \(\hat{Q}_{j}^{\mathrm{E}}\) there. A Taylor expansion of \(F_{j0}\) around \(Q_{j0}\left(p\right)\), together with the local stochastic equicontinuity of the residual empirical process, gives the weighted inverse relation
\[
\sup_{p\in\mathcal{I}_T}
\frac{
\vartheta_j\left(p\right)
}{
f_{j0}\left(Q_{j0}\left(p\right)\right)
}
\left|
\sqrt{T}
f_{j0}\left(Q_{j0}\left(p\right)\right)
\left[
\hat{Q}_{j}^{\mathrm{E}}\left(p\right)
-
Q_{j0}\left(p\right)
\right]
+
\sqrt{T}
\left[
\hat{F}_{j,T}\left(Q_{j0}\left(p\right)\right)-p
\right]
\right|
=
o_{\mathbb{P}}\left(1\right).
\]
The condition \(T\epsilon_T/\log T\rightarrow\infty\) makes the quantile discretization error negligible uniformly on \(\mathcal{I}_T\). Substituting equation~\eqref{eq:residual-empirical-process-expansion} at \(u=Q_{j0}\left(p\right)\) gives equation~\eqref{eq:generated-quantile-interior-expansion}. The signs follow from inversion: a positive residual-distribution error moves the estimated quantile downward. Finally, equations~\eqref{eq:residual-distribution-derivative} and the definition of \(\boldsymbol{r}_j\) give equation~\eqref{eq:generated-quantile-residual-jacobian}. The inverse-map argument is formalized in \citet[Example~20.5 and Chapter~21]{VanDerVaart1998}, \citet[Section~2.2.4 and Theorem~2.8]{Kosorok2008}, and \citet[Lemmas~3.10.21 and~3.10.24, and Example~3.10.25]{VanDerVaartWellner2023}.
\end{proof}

In the scalar location--scale model, equation~\eqref{eq:generated-quantile-residual-jacobian} reduces to equation~\eqref{eq:scalar-generated-residual-correction}. Thus, estimation of the conditional mean produces a rank-independent correction, whereas estimation of the conditional scale produces a correction that varies with \(Q_0\left(p\right)\). This distinction becomes relevant when the shock moves the impact rank into a tail of the innovation distribution.

\subsection{Weighted completion at the tails}
\label{app:residual-quantiles-weighted-tails}

The central-rank expansion must be extended before it can be applied to the exact population response. The weight in equation~\eqref{eq:weighted-quantile-norm} is chosen so that the empirical-quantile fluctuation remains bounded while the response derivatives remain integrable.

\begin{lemma}[Full weighted empirical-quantile expansion]
\label{lem:generated-quantile-full-range}
Under the conditions of Lemma~\ref{lem:generated-quantile-interior}, Assumption~\ref{ass:scalar-innovations} in the scalar model, and the corresponding componentwise tail conditions in the vector model,
\begin{equation}
\left\|
\sqrt{T}
\left(
\hat{Q}_{j}^{\mathrm{E}}-Q_{j0}
\right)
-
\frac{1}{\sqrt{T}}
\sum_{t=1}^{T}
\left[
\Xi_{j,t}
+
\boldsymbol{r}_j\left(\cdot\right)^{\prime}\boldsymbol{L}_t
\right]
\right\|_j
=
o_{\mathbb{P}}\left(1\right).
\label{eq:generated-quantile-full-range}
\end{equation}
The conclusion holds jointly over the fixed number of innovation components.
\end{lemma}

\begin{proof}
Lemma~\ref{lem:generated-quantile-interior} proves the result on \(\mathcal{I}_T\). It remains to control ranks below \(\epsilon_T\) and above \(1-\epsilon_T\). Condition~\eqref{eq:quantile-density-tail-condition} implies that the quantile density changes by at most a fixed factor over probability intervals whose width is proportional to \(p\left(1-p\right)\). Specifically, for every fixed \(c\in\left(0,1\right)\),
\[
\sup_{0<p<1}
\sup_{\substack{0<v<1\\
\left|v-p\right|\leq c p\left(1-p\right)}}
\frac{
f_{j0}\left(Q_{j0}\left(v\right)\right)
}{
f_{j0}\left(Q_{j0}\left(p\right)\right)
}
<
\infty,
\]
and the same bound holds after interchanging \(p\) and \(v\). This follows by integrating
\[
\frac{\mathrm{d}}{\mathrm{d}p}
\log f_{j0}\left(Q_{j0}\left(p\right)\right)
=
\frac{
f_{j0}^{\prime}\left(Q_{j0}\left(p\right)\right)
}{
f_{j0}\left(Q_{j0}\left(p\right)\right)^2
}.
\]

The weighted residual empirical-process expansion in equation~\eqref{eq:residual-empirical-process-expansion}, the preceding local density comparison, and the generalized-inverse inequalities used in Lemma~\ref{lem:generated-quantile-interior} give the same linearization on the two tail regions. The extreme order statistics are controlled by the endpoint moments and the weighted bounds in Assumption~\ref{ass:scalar-innovations}; the condition \(T\epsilon_T/\log T\rightarrow\infty\) makes the empirical-quantile discretization error negligible where the central and tail arguments meet. The generated-residual term is controlled by
\[
\sup_{0<p<1}
\vartheta_j\left(p\right)
\left\|
\boldsymbol{r}_j\left(p\right)
\right\|
<
\infty.
\]
Combining the lower tail, central region, and upper tail proves equation~\eqref{eq:generated-quantile-full-range}. This is the weighted inverse-map argument: the ordinary inverse map is treated in \citet[Lemmas~3.10.21 and~3.10.24]{VanDerVaartWellner2023}, while the extension to nonuniform sup-norms follows the quasi-Hadamard delta-method formulation in \citet{BeutnerZaehle2016}. Weighted empirical-process bounds for dynamic models are developed in \citet[Chapters~2 and~7]{Koul2002}.
\end{proof}

\begin{lemma}[Continuity at ordinary and shifted tails]
\label{lem:quantile-tail-control}
Suppose Assumption~\ref{ass:response-tail-continuity} holds. For \(\epsilon\in\left(0,1/2\right)\), let \(\dot{\Psi}_{h,\epsilon}^{\mathrm{dist}}\) denote \(\dot{\Psi}_{h}^{\mathrm{dist}}\) with every ordinary-rank integral restricted to \(\left[\epsilon,1-\epsilon\right]\). Let \(\dot{\Psi}_{h,\epsilon}^{\mathrm{imp}}\) denote the shifted-impact map after the change of variables \(r=\tau_{\delta}\left(p\right)\), with the integral in \(r\) restricted to the same interval. Then
\begin{equation}
\sup_{\left\|\boldsymbol{q}\right\|\leq1}
\left|
\left(
\dot{\Psi}_{h}^{\mathrm{dist}}
-
\dot{\Psi}_{h,\epsilon}^{\mathrm{dist}}
\right)
\left[\boldsymbol{q}\right]
\right|
+
\sup_{\left\|\boldsymbol{q}\right\|\leq1}
\left|
\left(
\dot{\Psi}_{h}^{\mathrm{imp}}
-
\dot{\Psi}_{h,\epsilon}^{\mathrm{imp}}
\right)
\left[\boldsymbol{q}\right]
\right|
\rightarrow
0
\label{eq:quantile-tail-operator-control}
\end{equation}
as \(\epsilon\downarrow0\).
\end{lemma}

\begin{proof}
For an ordinary-rank map and \(\left\|\boldsymbol{q}\right\|\leq1\), the absolute tail contribution is bounded by
\[
\sum_{j=1}^{n}
\left(
\int_{0}^{\epsilon}
+
\int_{1-\epsilon}^{1}
\right)
\frac{
\left|
\omega_{h,j}^{\mathrm{dist}}\left(p\right)
\right|
}{
\vartheta_j\left(p\right)
}
\,\dif p.
\]
For the shifted-impact map, changing variables from the original impact rank to the shifted rank gives the bound
\[
\left(
\int_{0}^{\epsilon}
+
\int_{1-\epsilon}^{1}
\right)
\frac{
\left|
\omega_{h,k}^{\mathrm{imp}}
\left(
\rho_{\delta}\left(p\right)
\right)
\right|
\rho_{\delta}^{\prime}\left(p\right)
}{
\vartheta_k\left(p\right)
}
\,\dif p.
\]
Both expressions converge to zero by equation~\eqref{eq:response-tail-integrability}.
\end{proof}

\subsection{Ordinary, future, and shifted impact ranks}
\label{app:residual-quantiles-response-parts}

The response derivative uses the empirical quantiles in three places. At impact, the unshocked path and the unchanged components of the shocked path use ordinary ranks. At later dates, both paths use common future ranks. The selected shocked component at impact uses the shifted rank. To display these contributions separately, define
\begin{equation}
\begin{aligned}
\omega_{h,j}^{\mathrm{unshift}}\left(p\right)
&=
\mathbb{E}\left[
\left(
\mathbf{1}\left\{j\neq k\right\}
\Lambda_{h,1,j}^{\delta}
-
\Lambda_{h,1,j}^{0}
\right)
\mathrel{\big|}
P_{j1}=p
\right],\\
\omega_{h,j}^{\mathrm{future}}\left(p\right)
&=
\sum_{s=2}^{h}
\mathbb{E}\left[
\Lambda_{h,s,j}^{\delta}
-
\Lambda_{h,s,j}^{0}
\mathrel{\big|}
P_{js}=p
\right].
\end{aligned}
\label{eq:unshifted-future-propagation-weights}
\end{equation}
Then
\[
\omega_{h,j}^{\mathrm{dist}}
=
\omega_{h,j}^{\mathrm{unshift}}
+
\omega_{h,j}^{\mathrm{future}}.
\]
For a vector of quantile perturbations \(\boldsymbol{q}\), let
\[
\dot{\Psi}_{h}^{\mathrm{unshift}}\left[\boldsymbol{q}\right]
=
\sum_{j=1}^{n}
\int_{0}^{1}
\omega_{h,j}^{\mathrm{unshift}}\left(p\right)
q_j\left(p\right)
\,\dif p,
\qquad
\dot{\Psi}_{h}^{\mathrm{future}}\left[\boldsymbol{q}\right]
=
\sum_{j=1}^{n}
\int_{0}^{1}
\omega_{h,j}^{\mathrm{future}}\left(p\right)
q_j\left(p\right)
\,\dif p.
\]
Consequently,
\[
\dot{\Psi}_{h}^{\mathrm{dist}}
=
\dot{\Psi}_{h}^{\mathrm{unshift}}
+
\dot{\Psi}_{h}^{\mathrm{future}}.
\]
The shifted-impact map remains \(\dot{\Psi}_{h}^{\mathrm{imp}}\) from equation~\eqref{eq:response-quantile-derivatives}.

\begin{lemma}[Projection through the response derivatives]
\label{lem:projected-generated-quantiles}
Suppose Assumption~\ref{ass:response-tail-continuity} holds. Let \(\boldsymbol{\Xi}_t\) and \(\boldsymbol{R}_t\) be defined as in Section~\ref{sec:main-theory}. Then
\begin{equation}
\begin{aligned}
\dot{\Psi}_{h}^{\mathrm{unshift}}
\left[
\sqrt{T}
\left(
\boldsymbol{\hat{Q}}^{\mathrm{E}}
-
\boldsymbol{Q}_0
\right)
\right]
&=
\frac{1}{\sqrt{T}}
\sum_{t=1}^{T}
\left\{
\dot{\Psi}_{h}^{\mathrm{unshift}}
\left[\boldsymbol{\Xi}_t\right]
+
\dot{\Psi}_{h}^{\mathrm{unshift}}
\left[\boldsymbol{R}_t\right]
\right\}
+
o_{\mathbb{P}}\left(1\right),\\
\dot{\Psi}_{h}^{\mathrm{future}}
\left[
\sqrt{T}
\left(
\boldsymbol{\hat{Q}}^{\mathrm{E}}
-
\boldsymbol{Q}_0
\right)
\right]
&=
\frac{1}{\sqrt{T}}
\sum_{t=1}^{T}
\left\{
\dot{\Psi}_{h}^{\mathrm{future}}
\left[\boldsymbol{\Xi}_t\right]
+
\dot{\Psi}_{h}^{\mathrm{future}}
\left[\boldsymbol{R}_t\right]
\right\}
+
o_{\mathbb{P}}\left(1\right),\\
\dot{\Psi}_{h}^{\mathrm{imp}}
\left[
\sqrt{T}
\left(
\boldsymbol{\hat{Q}}^{\mathrm{E}}
-
\boldsymbol{Q}_0
\right)
\right]
&=
\frac{1}{\sqrt{T}}
\sum_{t=1}^{T}
\left\{
\dot{\Psi}_{h}^{\mathrm{imp}}
\left[\boldsymbol{\Xi}_t\right]
+
\dot{\Psi}_{h}^{\mathrm{imp}}
\left[\boldsymbol{R}_t\right]
\right\}
+
o_{\mathbb{P}}\left(1\right).
\end{aligned}
\label{eq:projected-generated-quantile-expansions}
\end{equation}
\end{lemma}

\begin{proof}
Apply Lemma~\ref{lem:generated-quantile-full-range} componentwise. The ordinary-impact and future maps are continuous linear maps under equation~\eqref{eq:response-tail-integrability}. For the shifted-impact map, use \(r=\tau_{\delta}\left(p\right)\) to write
\[
\dot{\Psi}_{h}^{\mathrm{imp}}\left[\boldsymbol{q}\right]
=
\int_{0}^{1}
\omega_{h,k}^{\mathrm{imp}}
\left(
\rho_{\delta}\left(r\right)
\right)
\rho_{\delta}'\left(r\right)
q_k\left(r\right)
\,\dif r.
\]
The second integral in equation~\eqref{eq:response-tail-integrability} is the continuity bound for this map, and Lemma~\ref{lem:quantile-tail-control} records the corresponding tail approximation. Applying the three continuous maps to equation~\eqref{eq:generated-quantile-full-range} proves equation~\eqref{eq:projected-generated-quantile-expansions}.
\end{proof}

The empirical-process terms in the first two lines of equation~\eqref{eq:projected-generated-quantile-expansions} sum to \(Z_{h,t}^{\mathrm{dist}}\) in equation~\eqref{eq:influence-decomposition}. The empirical-process term in the third line is \(Z_{h,t}^{\mathrm{imp}}\). The three terms involving \(\boldsymbol{R}_t\) sum to \(Z_{h,t}^{\mathrm{res}}\). Thus, the main decomposition distinguishes the shifted impact innovation from the remaining uses of the innovation distribution, while retaining all effects of generated residuals in one term.

\subsection{Density cancellation}
\label{app:residual-quantiles-density-cancellation}

The feasible representation in Proposition~\ref{prop:density-free-influence} follows from two changes of variables. For any integrable ordinary-rank weight \(\omega\),
\begin{equation}
\begin{aligned}
\int_{0}^{1}
\omega\left(p\right)
\Xi_{j,t}\left(p\right)
\,\dif p
&=
\int_{\mathbb{R}}
\omega\left(F_{j0}\left(u\right)\right)
\left[
F_{j0}\left(u\right)
-
\mathbf{1}\left\{U_{jt}\leq u\right\}
\right]
\,\dif u,\\
\int_{0}^{1}
\omega\left(p\right)
\Xi_{k,t}\left(\tau_{\delta}\left(p\right)\right)
\,\dif p
&=
\int_{\mathbb{R}}
\omega\left(
\rho_{\delta}\left(F_{k0}\left(u\right)\right)
\right)
\rho_{\delta}'\left(F_{k0}\left(u\right)\right)
\left[
F_{k0}\left(u\right)
-
\mathbf{1}\left\{U_{kt}\leq u\right\}
\right]
\,\dif u.
\end{aligned}
\label{eq:quantile-density-cancellation}
\end{equation}
The first equality uses \(p=F_{j0}\left(u\right)\). The second first uses \(r=\tau_{\delta}\left(p\right)\), followed by \(r=F_{k0}\left(u\right)\). In both cases, the innovation density in \(\Xi_{j,t}\) cancels with the Jacobian of the quantile transformation.

Equation~\eqref{eq:generated-quantile-residual-jacobian} similarly gives
\begin{equation}
\begin{aligned}
\int_{0}^{1}
\omega\left(p\right)
\boldsymbol{r}_j\left(p\right)
\,\dif p
&=
\mathbb{E}\left[
\omega\left(P_{jt}\right)
\boldsymbol{J}_{j,t}
\right],\\
\int_{0}^{1}
\omega\left(p\right)
\boldsymbol{r}_k\left(\tau_{\delta}\left(p\right)\right)
\,\dif p
&=
\mathbb{E}\left[
\omega\left(
\rho_{\delta}\left(P_{kt}\right)
\right)
\rho_{\delta}'\left(P_{kt}\right)
\boldsymbol{J}_{k,t}
\right].
\end{aligned}
\label{eq:generated-residual-change-of-variables}
\end{equation}
Applying equations~\eqref{eq:quantile-density-cancellation} and~\eqref{eq:generated-residual-change-of-variables} with the propagation weights in equation~\eqref{eq:propagation-weights} proves equation~\eqref{eq:density-free-influence}.

\begin{proof}[Proof of Proposition~\ref{prop:generated-residual-quantile-expansion}]
For the scalar model, Lemma~\ref{lem:generated-quantile-full-range} and equation~\eqref{eq:beta-linearization} give equation~\eqref{eq:generated-quantile-expansion}. The fixed number of innovation components permits joint stacking. For the vector model, apply the same argument componentwise under the conditions stated in the proposition.
\end{proof}

\begin{remark}[Clipped implementation]
\label{rem:residual-quantile-clipping}
Suppose the numerical implementation uses the clipping sequence \(\epsilon_T\) in Appendix~\ref{app:assumptions-tails}. Let \(\hat{\psi}_{h,S}^{\mathrm{E},\epsilon_T}\) and \(\psi_h^{\epsilon_T}\) denote the resulting simulated estimator and population response. Lemma~\ref{lem:generated-quantile-interior} then applies directly to every quantile evaluation used by the clipped path simulator. Under equation~\eqref{eq:tail-negligibility} and the same path envelope uniformly over a neighborhood of \(\left(\boldsymbol{\beta}_0,\boldsymbol{Q}_0\right)\),
\[
\sqrt{T}
\left[
\left(
\hat{\psi}_{h,S}^{\mathrm{E}}
-
\psi_h
\right)
-
\left(
\hat{\psi}_{h,S}^{\mathrm{E},\epsilon_T}
-
\psi_h^{\epsilon_T}
\right)
\right]
=
o_{\mathbb{P}}\left(1\right)
\]
when simulation error is negligible at the \(\sqrt{T}\) scale. Thus, clipping may be used to stabilize endpoint calculations without changing the first-order distribution, provided its effect is verified through equation~\eqref{eq:tail-negligibility}.
\end{remark}

\section{Differentiating the recursive response}
\label{app:recursive-derivatives}

This appendix differentiates the finite-horizon path simulator and combines those derivatives with the generated-quantile expansion from Appendix~\ref{app:residual-quantiles}. The first part gives forward and reverse recursions that can be evaluated analytically or by automatic differentiation. The second part establishes the first-order expansion of the population response. The final part assembles the observation-level influence contributions and proves their joint Gaussian limit for a fixed collection of responses.

\subsection{Pathwise derivatives}
\label{app:recursive-derivatives-pathwise}

Fix one response specification \(\left(h,\boldsymbol{y},\boldsymbol{a},k,\delta\right)\). For \(r\in\left\{0,\delta\right\}\), write
\[
\boldsymbol{U}_j^r
=
\boldsymbol{Q}_0\left(\boldsymbol{P}_j^r\right),
\qquad
\boldsymbol{Y}_j^r
=
g\left(\boldsymbol{Y}_{j-1}^r,\boldsymbol{U}_j^r;\boldsymbol{\beta}_0\right),
\qquad
j=1,\ldots,h.
\]
At the arguments reached by this path, define the transition Jacobians
\[
\begin{aligned}
\boldsymbol{\mathcal{G}}_{y,j}^r
&=
\left.
\frac{\partial g\left(\boldsymbol{y},\boldsymbol{u};\boldsymbol{\beta}\right)}{\partial\boldsymbol{y}^{\prime}}
\right|_{\left(\boldsymbol{Y}_{j-1}^r,\boldsymbol{U}_j^r,\boldsymbol{\beta}_0\right)},
\\
\boldsymbol{\mathcal{G}}_{u,j}^r
&=
\left.
\frac{\partial g\left(\boldsymbol{y},\boldsymbol{u};\boldsymbol{\beta}\right)}{\partial\boldsymbol{u}^{\prime}}
\right|_{\left(\boldsymbol{Y}_{j-1}^r,\boldsymbol{U}_j^r,\boldsymbol{\beta}_0\right)},
\\
\boldsymbol{\mathcal{G}}_{\beta,j}^r
&=
\left.
\frac{\partial g\left(\boldsymbol{y},\boldsymbol{u};\boldsymbol{\beta}\right)}{\partial\boldsymbol{\beta}^{\prime}}
\right|_{\left(\boldsymbol{Y}_{j-1}^r,\boldsymbol{U}_j^r,\boldsymbol{\beta}_0\right)}.
\end{aligned}
\]
Their dimensions are \(d\times d\), \(d\times n\), and \(d\times p\), respectively.

Let
\[
\boldsymbol{H}_{\beta,j}^r
=
\left.
\frac{\partial\boldsymbol{Y}_j^r}{\partial\boldsymbol{\beta}^{\prime}}
\right|_{\left(\boldsymbol{\beta}_0,\boldsymbol{Q}_0\right)},
\qquad
\boldsymbol{H}_{u,j,s}^r
=
\left.
\frac{\partial\boldsymbol{Y}_j^r}{\partial\left(\boldsymbol{U}_s^r\right)^{\prime}}
\right|_{\left(\boldsymbol{\beta}_0,\boldsymbol{Q}_0\right)},
\qquad
1\leq s\leq j\leq h.
\]
Holding the innovation inputs fixed when differentiating with respect to \(\boldsymbol{\beta}\), the chain rule gives
\begin{equation}
\begin{aligned}
\boldsymbol{H}_{\beta,0}^r
&=
\boldsymbol{0}_{d\times p},
&
\boldsymbol{H}_{\beta,j}^r
&=
\boldsymbol{\mathcal{G}}_{y,j}^r
\boldsymbol{H}_{\beta,j-1}^r
+
\boldsymbol{\mathcal{G}}_{\beta,j}^r,
\\
\boldsymbol{H}_{u,s,s}^r
&=
\boldsymbol{\mathcal{G}}_{u,s}^r,
&
\boldsymbol{H}_{u,j,s}^r
&=
\boldsymbol{\mathcal{G}}_{y,j}^r
\boldsymbol{H}_{u,j-1,s}^r,
\qquad
j=s+1,\ldots,h.
\end{aligned}
\label{eq:forward-recursive-derivatives}
\end{equation}
Consequently,
\[
\boldsymbol{H}_{u,h,s}^r
=
\boldsymbol{\mathcal{G}}_{y,h}^r
\cdots
\boldsymbol{\mathcal{G}}_{y,s+1}^r
\boldsymbol{\mathcal{G}}_{u,s}^r,
\]
where the product preceding \(\boldsymbol{\mathcal{G}}_{u,s}^r\) is the identity matrix when \(s=h\). The innovation sensitivity used in the main text is therefore
\begin{equation}
\Lambda_{h,s,j}^r
=
\boldsymbol{a}^{\prime}
\boldsymbol{H}_{u,h,s}^r
\boldsymbol{e}_j,
\qquad
s=1,\ldots,h,
\qquad
j=1,\ldots,n,
\label{eq:path-innovation-sensitivity}
\end{equation}
where \(\boldsymbol{e}_j\) is the \(j\)th coordinate vector in \(\mathbb{R}^n\). The direct parameter derivative is
\begin{equation}
\boldsymbol{A}_h
=
\mathbb{E}\left[
\left(
\boldsymbol{H}_{\beta,h}^{\delta}
-
\boldsymbol{H}_{\beta,h}^{0}
\right)^{\prime}
\boldsymbol{a}
\right].
\label{eq:direct-transition-derivative-recursion}
\end{equation}
Equations~\eqref{eq:path-innovation-sensitivity} and~\eqref{eq:direct-transition-derivative-recursion} coincide with the derivative objects introduced before Theorem~\ref{thm:main-linearization}.

A reverse recursion computes the same quantities more efficiently when the response is scalar. Define the adjoint vectors by
\begin{equation}
\boldsymbol{C}_h^r
=
\boldsymbol{a},
\qquad
\boldsymbol{C}_{j-1}^r
=
\left(\boldsymbol{\mathcal{G}}_{y,j}^r\right)^{\prime}
\boldsymbol{C}_j^r,
\qquad
j=h,\ldots,1.
\label{eq:recursive-adjoint}
\end{equation}
Then
\begin{equation}
\boldsymbol{\Lambda}_{h,s}^r
=
\left(\boldsymbol{\mathcal{G}}_{u,s}^r\right)^{\prime}
\boldsymbol{C}_s^r,
\qquad
\frac{\partial\left(\boldsymbol{a}^{\prime}\boldsymbol{Y}_h^r\right)}{\partial\boldsymbol{\beta}}
=
\sum_{j=1}^{h}
\left(\boldsymbol{\mathcal{G}}_{\beta,j}^r\right)^{\prime}
\boldsymbol{C}_j^r,
\qquad
\Lambda_{h,s,j}^r
=
\boldsymbol{e}_j^{\prime}
\boldsymbol{\Lambda}_{h,s}^r.
\label{eq:reverse-recursive-derivatives}
\end{equation}
Thus, one forward evaluation of the path followed by one reverse pass returns the gradient with respect to \(\boldsymbol{\beta}\) and all innovation inputs. Reverse-mode automatic differentiation applied to equation~\eqref{eq:paired-paths} performs the recursion in equations~\eqref{eq:recursive-adjoint}--\eqref{eq:reverse-recursive-derivatives}. In the empirical implementation, the quantile outputs are treated as numerical innovation inputs during this pass. Sorting and evaluation of the empirical quantile function are not differentiated; their sampling effect is supplied by Appendix~\ref{app:residual-quantiles}. Replacing the population inputs in these recursions with \(\boldsymbol{\hat{\beta}}\), \(\boldsymbol{\hat{Q}}^{\mathrm{E}}\), and the simulated paths gives the derivatives used in Section~\ref{sec:main-theory-feasible}.

For the scalar location--scale model, equation~\eqref{eq:forward-recursive-derivatives} reduces to equation~\eqref{eq:scalar-recursive-derivatives}.

\subsection{Derivative of the response functional}
\label{app:recursive-derivatives-functional}

For a vector of quantile perturbations \(\boldsymbol{q}=\left(q_1,\ldots,q_n\right)^{\prime}\), write
\[
\boldsymbol{q}\left(\boldsymbol{p}\right)
=
\left(
q_1\left(p_1\right),
\ldots,
q_n\left(p_n\right)
\right)^{\prime}.
\]
Let \(\boldsymbol{b}\in\mathbb{R}^p\) be a parameter perturbation. The corresponding first-order change in path \(r\) satisfies
\begin{equation}
\dot{\boldsymbol{Y}}_0^r\left[\boldsymbol{b},\boldsymbol{q}\right]
=
\boldsymbol{0},
\qquad
\dot{\boldsymbol{Y}}_j^r\left[\boldsymbol{b},\boldsymbol{q}\right]
=
\boldsymbol{\mathcal{G}}_{y,j}^r
\dot{\boldsymbol{Y}}_{j-1}^r\left[\boldsymbol{b},\boldsymbol{q}\right]
+
\boldsymbol{\mathcal{G}}_{\beta,j}^r
\boldsymbol{b}
+
\boldsymbol{\mathcal{G}}_{u,j}^r
\boldsymbol{q}\left(\boldsymbol{P}_j^r\right),
\qquad
j=1,\ldots,h.
\label{eq:directional-state-recursion}
\end{equation}
The rank vectors in equation~\eqref{eq:directional-state-recursion} remain fixed. In particular, the selected component of \(\boldsymbol{q}\left(\boldsymbol{P}_1^{\delta}\right)\) is evaluated at \(\tau_{\delta}\left(P_{k1}\right)\).

Iterating equation~\eqref{eq:directional-state-recursion} gives the path-response derivative
\begin{equation}
\begin{aligned}
\dot{D}_h\left[\boldsymbol{b},\boldsymbol{q}\right]
&=
\boldsymbol{a}^{\prime}
\left(
\boldsymbol{H}_{\beta,h}^{\delta}
-
\boldsymbol{H}_{\beta,h}^{0}
\right)
\boldsymbol{b}
+
\sum_{j=1}^{n}
\left(
\mathbf{1}\left\{j\neq k\right\}
\Lambda_{h,1,j}^{\delta}
-
\Lambda_{h,1,j}^{0}
\right)
q_j\left(P_{j1}\right)
\\
&\quad+
\Lambda_{h,1,k}^{\delta}
q_k\left(
\tau_{\delta}\left(P_{k1}\right)
\right)
+
\sum_{s=2}^{h}
\sum_{j=1}^{n}
\left(
\Lambda_{h,s,j}^{\delta}
-
\Lambda_{h,s,j}^{0}
\right)
q_j\left(P_{js}\right).
\end{aligned}
\label{eq:path-response-directional-derivative}
\end{equation}
The terms on the first line after the parameter derivative cover the ordinary impact ranks. The second line contains the shifted impact rank and the future ranks.

\begin{lemma}[First-order differentiability of the finite-horizon response]
\label{lem:recursive-response-differentiability}
Suppose the twice-differentiability and moment-envelope conditions in Assumption~\ref{ass:scalar-dynamics}\textup{(i)--(ii)} or Assumption~\ref{ass:vector-high-level}\textup{(iii)} hold, together with Assumption~\ref{ass:response-tail-continuity}. Let \(t_{\ell}\downarrow0\), let \(\boldsymbol{b}_{\ell}\rightarrow\boldsymbol{b}\), and let \(\boldsymbol{Q}_{\ell}\) be admissible collections of quantile functions satisfying
\[
\left\|
\frac{\boldsymbol{Q}_{\ell}-\boldsymbol{Q}_0}{t_{\ell}}
-
\boldsymbol{q}
\right\|
\rightarrow
0.
\]
Then the difference quotient of the path response converges to equation~\eqref{eq:path-response-directional-derivative} in \(L^1\). Under the \(L^{2+\eta}\) envelope in Assumption~\ref{ass:response-tail-continuity}, the convergence also holds in \(L^2\). Consequently,
\begin{equation}
\frac{
\Psi_h\left(
\boldsymbol{\beta}_0+t_{\ell}\boldsymbol{b}_{\ell},
\boldsymbol{Q}_{\ell}
\right)
-
\Psi_h\left(
\boldsymbol{\beta}_0,
\boldsymbol{Q}_0
\right)
}{t_{\ell}}
\rightarrow
\boldsymbol{A}_h^{\prime}\boldsymbol{b}
+
\dot{\Psi}_h^{\mathrm{dist}}\left[\boldsymbol{q}\right]
+
\dot{\Psi}_h^{\mathrm{imp}}\left[\boldsymbol{q}\right].
\label{eq:population-response-derivative}
\end{equation}
The convergence in equation~\eqref{eq:population-response-derivative} holds jointly for any fixed collection of response specifications.
\end{lemma}

\begin{proof}
For path \(r\), let \(\boldsymbol{Y}_{j,\ell}^r\) denote the state generated by \(\boldsymbol{\beta}_0+t_{\ell}\boldsymbol{b}_{\ell}\) and \(\boldsymbol{Q}_{\ell}\). A first-order Taylor expansion of one transition step gives
\[
\frac{
\boldsymbol{Y}_{j,\ell}^r
-
\boldsymbol{Y}_j^r
}{t_{\ell}}
=
\boldsymbol{\mathcal{G}}_{y,j}^r
\frac{
\boldsymbol{Y}_{j-1,\ell}^r
-
\boldsymbol{Y}_{j-1}^r
}{t_{\ell}}
+
\boldsymbol{\mathcal{G}}_{\beta,j}^r
\boldsymbol{b}_{\ell}
+
\boldsymbol{\mathcal{G}}_{u,j}^r
\frac{
\boldsymbol{Q}_{\ell}\left(\boldsymbol{P}_j^r\right)
-
\boldsymbol{Q}_0\left(\boldsymbol{P}_j^r\right)
}{t_{\ell}}
+
\boldsymbol{o}_{\ell,j}^r.
\]
The final remainder is bounded by a second-derivative envelope times the square of the state, parameter, and innovation perturbations. Because \(h\) is fixed, induction over \(j=1,\ldots,h\) and the maintained moment conditions imply
\[
\frac{
\boldsymbol{Y}_{j,\ell}^r
-
\boldsymbol{Y}_j^r
}{t_{\ell}}
-
\dot{\boldsymbol{Y}}_j^r\left[\boldsymbol{b},\boldsymbol{q}\right]
\rightarrow
\boldsymbol{0}
\]
in \(L^1\), and in \(L^2\) under the stronger envelope. Assumption~\ref{ass:response-tail-continuity} controls the ordinary and shifted quantile perturbations near the endpoints. Multiplication by \(\boldsymbol{a}^{\prime}\) and subtraction of the two paths yield equation~\eqref{eq:path-response-directional-derivative}. Taking expectations gives equation~\eqref{eq:population-response-derivative}. The fixed number of responses permits componentwise stacking.
\end{proof}

Applying Lemma~\ref{lem:recursive-response-differentiability} to the estimated parameter and empirical quantile functions gives
\begin{equation}
\begin{aligned}
&\sqrt{T}
\left[
\Psi_h\left(
\boldsymbol{\hat{\beta}},
\boldsymbol{\hat{Q}}^{\mathrm{E}}
\right)
-
\Psi_h\left(
\boldsymbol{\beta}_0,
\boldsymbol{Q}_0
\right)
\right]
\\
&\quad=
\boldsymbol{A}_h^{\prime}
\sqrt{T}
\left(
\boldsymbol{\hat{\beta}}
-
\boldsymbol{\beta}_0
\right)
+
\dot{\Psi}_h^{\mathrm{dist}}
\left[
\sqrt{T}
\left(
\boldsymbol{\hat{Q}}^{\mathrm{E}}
-
\boldsymbol{Q}_0
\right)
\right]
\\
&\qquad+
\dot{\Psi}_h^{\mathrm{imp}}
\left[
\sqrt{T}
\left(
\boldsymbol{\hat{Q}}^{\mathrm{E}}
-
\boldsymbol{Q}_0
\right)
\right]
+
o_{\mathbb{P}}\left(1\right).
\end{aligned}
\label{eq:recursive-plug-in-expansion}
\end{equation}
The quantile maps in equation~\eqref{eq:recursive-plug-in-expansion} include the generated-residual effect; it is therefore not added to the direct parameter derivative in equation~\eqref{eq:direct-transition-derivative-recursion}.

\subsection{Observation-level representation and joint limit}
\label{app:recursive-derivatives-limit}

We first remove the numerical integration error under the rate used in Theorem~\ref{thm:main-linearization}.

\begin{lemma}[Negligible simulation error]
\label{lem:negligible-simulation-error}
Suppose the conditions of Lemma~\ref{lem:recursive-response-differentiability} hold, the fitted path responses have a uniformly bounded second moment with probability approaching one, and \(S/T\rightarrow\infty\). Then, jointly for any fixed collection of responses,
\begin{equation}
\sqrt{T}
\left[
\Psi_{h,S}\left(
\boldsymbol{\hat{\beta}},
\boldsymbol{\hat{Q}}^{\mathrm{E}}
\right)
-
\Psi_h\left(
\boldsymbol{\hat{\beta}},
\boldsymbol{\hat{Q}}^{\mathrm{E}}
\right)
\right]
=
o_{\mathbb{P}}\left(1\right).
\label{eq:negligible-simulation-remainder}
\end{equation}
\end{lemma}

\begin{proof}
Conditional on the observed sample, the path responses are independent across simulation draws. Their conditional covariance matrix is \(O_{\mathbb{P}}\left(1\right)\) by the maintained moment envelope and consistency of \(\boldsymbol{\hat{\beta}}\) and \(\boldsymbol{\hat{Q}}^{\mathrm{E}}\). The conditional second moment of the left-hand side of equation~\eqref{eq:negligible-simulation-remainder} is therefore \(O_{\mathbb{P}}\left(T/S\right)\), which converges to zero.
\end{proof}

\begin{proposition}[Observation-level representation]
\label{prop:recursive-observation-level-representation}
Suppose Assumption~\ref{ass:first-order}, Assumption~\ref{ass:response-tail-continuity}, and \(S/T\rightarrow\infty\) hold. Then equations~\eqref{eq:empirical-irf-linearization}--\eqref{eq:joint-empirical-irf-linearization} hold with the influence contributions in equation~\eqref{eq:influence-decomposition}.
\end{proposition}

\begin{proof}
By Lemma~\ref{lem:negligible-simulation-error}, the simulated estimator can be replaced by the exactly integrated fitted response at the \(\sqrt{T}\) scale. Apply equation~\eqref{eq:recursive-plug-in-expansion}. Equation~\eqref{eq:beta-linearization} turns the direct parameter term into
\[
\frac{1}{\sqrt{T}}
\sum_{t=1}^{T}
\boldsymbol{A}_h^{\prime}
\boldsymbol{L}_t
+
o_{\mathbb{P}}\left(1\right).
\]
Next, Lemma~\ref{lem:projected-generated-quantiles} expands the ordinary-impact, future, and shifted-impact quantile terms. Their empirical-process parts are \(Z_{h,t}^{\mathrm{dist}}\) and \(Z_{h,t}^{\mathrm{imp}}\), while their generated-residual parts sum to \(Z_{h,t}^{\mathrm{res}}\). Adding the direct term gives equation~\eqref{eq:influence-decomposition} and hence equation~\eqref{eq:empirical-irf-linearization}. Applying the same argument componentwise proves the stacked representation in equation~\eqref{eq:joint-empirical-irf-linearization}.
\end{proof}

\begin{proposition}[Joint Gaussian limit for fixed responses]
\label{prop:recursive-joint-clt}
Let \(m=1,\ldots,M\) index a fixed collection of response specifications, and let
\[
\boldsymbol{Z}_t
=
\left(
Z_t^{\left(1\right)},
\ldots,
Z_t^{\left(M\right)}
\right)^{\prime}
\]
stack their observation-level influence contributions. Suppose \(\left\{\boldsymbol{Z}_t\right\}\) is strictly stationary and strongly mixing with coefficients \(\alpha\left(\ell\right)\) satisfying
\[
\sum_{\ell=1}^{\infty}
\alpha\left(\ell\right)^{\eta/\left(2+\eta\right)}
<
\infty,
\qquad
\mathbb{E}\left[
\left\|\boldsymbol{Z}_t\right\|^{2+\eta}
\right]
<
\infty
\]
for some \(\eta>0\). Then the long-run covariance matrix in equation~\eqref{eq:empirical-irf-limit} is finite, and
\begin{equation}
\frac{1}{\sqrt{T}}
\sum_{t=1}^{T}
\boldsymbol{Z}_t
\xrightarrow{\mathrm{d}}
\mathcal{N}\left(
\boldsymbol{0},
\boldsymbol{\Omega}
\right).
\label{eq:recursive-joint-clt}
\end{equation}
The conclusion also holds under any alternative weak-dependence condition that yields the same multivariate central limit theorem and absolute convergence of the covariance series.
\end{proposition}

\begin{proof}
First, \(\mathbb{E}\left[\boldsymbol{L}_t\right]=\boldsymbol{0}\), and for every component and rank,
\[
\mathbb{E}\left[
\Xi_{j,t}\left(p\right)
\right]
=
0.
\]
The derivative maps are deterministic continuous linear maps. Hence each direct, generated-residual, future-distribution, and shifted-impact contribution has mean zero, so \(\mathbb{E}\left[\boldsymbol{Z}_t\right]=\boldsymbol{0}\). The \(L^{2+\eta}\) part of Assumption~\ref{ass:response-tail-continuity}, together with the moment condition on \(\boldsymbol{L}_t\), gives the stated moment bound for the stacked contribution.

For any \(\boldsymbol{c}\in\mathbb{R}^M\), the scalar sequence \(\left\{\boldsymbol{c}^{\prime}\boldsymbol{Z}_t\right\}\) has the same mixing coefficients and a finite \(2+\eta\) moment. The mixing summability condition implies the tail-and-dependence condition used by the central limit theorem for strongly mixing sequences. Therefore, \(T^{-1/2}\sum_{t=1}^{T}\boldsymbol{c}^{\prime}\boldsymbol{Z}_t\) converges to a centered normal variable; see \citet[Corollary~1.2 and Corollary~4.1]{Rio2017}. The same covariance inequality implies absolute summability of the covariance series. Applying the Cram\'er--Wold device over the fixed dimension \(M\) proves equation~\eqref{eq:recursive-joint-clt}.
\end{proof}

Under the scalar primitive conditions in Appendix~\ref{app:assumptions-scalar}, the estimating-equation contribution and the projected residual empirical process satisfy the joint central limit theorem imposed in Assumption~\ref{ass:scalar-estimator}. Under Assumption~\ref{ass:vector-high-level}, Proposition~\ref{prop:recursive-joint-clt} applies whenever the transition-estimator influence contribution is measurable with respect to the same strongly mixing process, as occurs for the smooth estimating-equation procedures considered here. Combining Propositions~\ref{prop:recursive-observation-level-representation} and~\ref{prop:recursive-joint-clt} proves Corollary~\ref{cor:main-gaussian-limit}. Growing-horizon inference would require uniform control of the derivative products in equation~\eqref{eq:forward-recursive-derivatives}; it is not covered by this argument.

\section{Proofs of the main results}
\label{app:proofs}

\subsection{Consistency and first-order expansion}
\label{app:proofs-main}

\begin{proof}[Proof of Proposition~\ref{prop:consistency}]
Fix a response specification indexed by \(m\), and suppress its fixed arguments. Let
\[
\mathcal{Y}_T
=
\sigma\left(
\boldsymbol{Y}_0,
\ldots,
\boldsymbol{Y}_T
\right).
\]

Equation~\eqref{eq:beta-linearization} and tightness of the transition-estimator influence sum give
\[
\left\|
\boldsymbol{\hat{\beta}}
-
\boldsymbol{\beta}_0
\right\|
=
O_{\mathbb{P}}\left(T^{-1/2}\right)
=
o_{\mathbb{P}}\left(1\right).
\]
Likewise, equation~\eqref{eq:generated-quantile-expansion} and Proposition~\ref{prop:generated-residual-quantile-expansion} imply
\[
\left\|
\boldsymbol{\hat{Q}}^{\mathrm{E}}
-
\boldsymbol{Q}_0
\right\|
=
O_{\mathbb{P}}\left(T^{-1/2}\right)
=
o_{\mathbb{P}}\left(1\right).
\]
For the smoothed estimator, the maintained condition gives directly
\[
\left\|
\boldsymbol{\hat{Q}}_{b_T}^{\mathrm{S}}
-
\boldsymbol{Q}_0
\right\|
=
o_{\mathbb{P}}\left(1\right).
\]

Let \(\boldsymbol{Q}_T\) denote either \(\boldsymbol{\hat{Q}}^{\mathrm{E}}\) or \(\boldsymbol{\hat{Q}}_{b_T}^{\mathrm{S}}\). The finite-horizon derivative in Lemma~\ref{lem:recursive-response-differentiability} and the uniform derivative envelope in Assumption~\ref{ass:first-order}\textup{(i)} yield
\[
\begin{aligned}
&\Psi_m\left(
\boldsymbol{\hat{\beta}},
\boldsymbol{Q}_T
\right)
-
\Psi_m\left(
\boldsymbol{\beta}_0,
\boldsymbol{Q}_T
\right)
\\
&\qquad=
\left\{
\int_{0}^{1}
\frac{\partial}{\partial\boldsymbol{\beta}}
\Psi_m\left(
\boldsymbol{\beta}_0
+
s\left(
\boldsymbol{\hat{\beta}}
-
\boldsymbol{\beta}_0
\right),
\boldsymbol{Q}_T
\right)
\,\dif s
\right\}^{\prime}
\left(
\boldsymbol{\hat{\beta}}
-
\boldsymbol{\beta}_0
\right)
\\
&\qquad=
o_{\mathbb{P}}\left(1\right).
\end{aligned}
\]
The final equality follows because the integrated derivative is \(O_{\mathbb{P}}\left(1\right)\) and the parameter error is \(o_{\mathbb{P}}\left(1\right)\).

The quantile derivative in equation~\eqref{eq:population-response-derivative} is continuous under the weighted norm in equation~\eqref{eq:weighted-quantile-norm}. Assumption~\ref{ass:first-order}\textup{(iii)}, stated explicitly as Assumption~\ref{ass:response-tail-continuity}, therefore gives a finite constant \(C_m\) such that, locally around \(\boldsymbol{Q}_0\),
\[
\left|
\dot{\Psi}_m^{\mathrm{dist}}\left[
\boldsymbol{q}
\right]
\right|
+
\left|
\dot{\Psi}_m^{\mathrm{imp}}\left[
\boldsymbol{q}
\right]
\right|
\leq
C_m
\left\|
\boldsymbol{q}
\right\|.
\]
Applying Lemma~\ref{lem:recursive-response-differentiability} with the parameter fixed at \(\boldsymbol{\beta}_0\) gives
\[
\Psi_m\left(
\boldsymbol{\beta}_0,
\boldsymbol{Q}_T
\right)
-
\Psi_m\left(
\boldsymbol{\beta}_0,
\boldsymbol{Q}_0
\right)
=
O_{\mathbb{P}}\left(
\left\|
\boldsymbol{Q}_T
-
\boldsymbol{Q}_0
\right\|
\right)
=
o_{\mathbb{P}}\left(1\right).
\]

Conditional on the observed sample, the simulated path responses are independent across \(s\) and have conditional mean
\(
\Psi_m\left(
\boldsymbol{\hat{\beta}},
\boldsymbol{Q}_T
\right)
\).
Hence
\[
\begin{aligned}
&\mathbb{E}\left[
\left|
\Psi_{m,S}\left(
\boldsymbol{\hat{\beta}},
\boldsymbol{Q}_T
\right)
-
\Psi_m\left(
\boldsymbol{\hat{\beta}},
\boldsymbol{Q}_T
\right)
\right|^2
\mathrel{\Big|}
\mathcal{Y}_T
\right]
\\
&\qquad=
\frac{1}{S}
\operatorname{Var}\left[
D_m\left(
\boldsymbol{P}_{1:h_m};
\boldsymbol{\hat{\beta}},
\boldsymbol{Q}_T
\right)
\mathrel{\Big|}
\mathcal{Y}_T
\right]
=
O_{\mathbb{P}}\left(S^{-1}\right).
\end{aligned}
\]
Conditional Chebyshev's inequality then gives
\[
\Psi_{m,S}\left(
\boldsymbol{\hat{\beta}},
\boldsymbol{Q}_T
\right)
-
\Psi_m\left(
\boldsymbol{\hat{\beta}},
\boldsymbol{Q}_T
\right)
=
o_{\mathbb{P}}\left(1\right)
\]
when \(S\rightarrow\infty\).

For the final decomposition, write
\[
\boldsymbol{\hat{Q}}^{r}
=
\begin{cases}
\boldsymbol{\hat{Q}}^{\mathrm{E}}, & r=\mathrm{E},\\
\boldsymbol{\hat{Q}}_{b_T}^{\mathrm{S}}, & r=\mathrm{S}.
\end{cases}
\]
Combining the preceding steps,
\[
\begin{aligned}
\hat{\psi}_{m,S}^{r}
-
\psi_m
&=
\left[
\Psi_{m,S}\left(
\boldsymbol{\hat{\beta}},
\boldsymbol{\hat{Q}}^{r}
\right)
-
\Psi_m\left(
\boldsymbol{\hat{\beta}},
\boldsymbol{\hat{Q}}^{r}
\right)
\right]
\\
&\quad+
\left[
\Psi_m\left(
\boldsymbol{\hat{\beta}},
\boldsymbol{\hat{Q}}^{r}
\right)
-
\Psi_m\left(
\boldsymbol{\beta}_0,
\boldsymbol{\hat{Q}}^{r}
\right)
\right]
\\
&\quad+
\left[
\Psi_m\left(
\boldsymbol{\beta}_0,
\boldsymbol{\hat{Q}}^{r}
\right)
-
\Psi_m\left(
\boldsymbol{\beta}_0,
\boldsymbol{Q}_0
\right)
\right]
\\
&=
o_{\mathbb{P}}\left(1\right),
\qquad
r\in\left\{\mathrm{E},\mathrm{S}\right\}.
\end{aligned}
\]
Because \(M\) is fixed,
\[
\Pr\left(
\max_{1\leq m\leq M}
\left|
\hat{\psi}_{m,S}^{r}
-
\psi_m
\right|
>\varepsilon
\right)
\leq
\sum_{m=1}^{M}
\Pr\left(
\left|
\hat{\psi}_{m,S}^{r}
-
\psi_m
\right|
>\varepsilon
\right)
\rightarrow
0.
\]
This proves both conclusions.
\end{proof}

\begin{proof}[Proof of Theorem~\ref{thm:main-linearization}]
Fix one response specification and suppress its fixed arguments.

Add and subtract the exactly integrated fitted response:
\begin{equation}
\begin{aligned}
\sqrt{T}
\left(
\hat{\psi}_{h,S}^{\mathrm{E}}
-
\psi_h
\right)
&=
\sqrt{T}
\left[
\Psi_{h,S}\left(
\boldsymbol{\hat{\beta}},
\boldsymbol{\hat{Q}}^{\mathrm{E}}
\right)
-
\Psi_h\left(
\boldsymbol{\hat{\beta}},
\boldsymbol{\hat{Q}}^{\mathrm{E}}
\right)
\right]
\\
&\quad+
\sqrt{T}
\left[
\Psi_h\left(
\boldsymbol{\hat{\beta}},
\boldsymbol{\hat{Q}}^{\mathrm{E}}
\right)
-
\Psi_h\left(
\boldsymbol{\beta}_0,
\boldsymbol{Q}_0
\right)
\right].
\end{aligned}
\label{eq:proof-main-sampling-simulation-decomposition}
\end{equation}
Lemma~\ref{lem:negligible-simulation-error} applies because \(S/T\rightarrow\infty\). Therefore,
\begin{equation}
\sqrt{T}
\left[
\Psi_{h,S}\left(
\boldsymbol{\hat{\beta}},
\boldsymbol{\hat{Q}}^{\mathrm{E}}
\right)
-
\Psi_h\left(
\boldsymbol{\hat{\beta}},
\boldsymbol{\hat{Q}}^{\mathrm{E}}
\right)
\right]
=
o_{\mathbb{P}}\left(1\right).
\label{eq:proof-main-negligible-simulation}
\end{equation}

Equation~\eqref{eq:recursive-plug-in-expansion}, obtained from Lemma~\ref{lem:recursive-response-differentiability}, gives
\begin{equation}
\begin{aligned}
&\sqrt{T}
\left[
\Psi_h\left(
\boldsymbol{\hat{\beta}},
\boldsymbol{\hat{Q}}^{\mathrm{E}}
\right)
-
\Psi_h\left(
\boldsymbol{\beta}_0,
\boldsymbol{Q}_0
\right)
\right]
\\
&\qquad=
\boldsymbol{A}_h^{\prime}
\sqrt{T}
\left(
\boldsymbol{\hat{\beta}}
-
\boldsymbol{\beta}_0
\right)
+
\dot{\Psi}_h^{\mathrm{dist}}\left[
\sqrt{T}
\left(
\boldsymbol{\hat{Q}}^{\mathrm{E}}
-
\boldsymbol{Q}_0
\right)
\right]
\\
&\qquad\quad+
\dot{\Psi}_h^{\mathrm{imp}}\left[
\sqrt{T}
\left(
\boldsymbol{\hat{Q}}^{\mathrm{E}}
-
\boldsymbol{Q}_0
\right)
\right]
+
o_{\mathbb{P}}\left(1\right).
\end{aligned}
\label{eq:proof-main-functional-expansion}
\end{equation}

Substituting equation~\eqref{eq:beta-linearization} into the first term on the right-hand side of equation~\eqref{eq:proof-main-functional-expansion} gives
\begin{equation}
\boldsymbol{A}_h^{\prime}
\sqrt{T}
\left(
\boldsymbol{\hat{\beta}}
-
\boldsymbol{\beta}_0
\right)
=
\frac{1}{\sqrt{T}}
\sum_{t=1}^{T}
\boldsymbol{A}_h^{\prime}
\boldsymbol{L}_t
+
o_{\mathbb{P}}\left(1\right)
=
\frac{1}{\sqrt{T}}
\sum_{t=1}^{T}
Z_{h,t}^{\mathrm{tr}}
+
o_{\mathbb{P}}\left(1\right).
\label{eq:proof-main-transition-term}
\end{equation}

Because
\(
\dot{\Psi}_h^{\mathrm{dist}}
=
\dot{\Psi}_h^{\mathrm{unshift}}
+
\dot{\Psi}_h^{\mathrm{future}}
\),
Lemma~\ref{lem:projected-generated-quantiles} and equation~\eqref{eq:projected-generated-quantile-expansions} imply
\begin{equation}
\begin{aligned}
\dot{\Psi}_h^{\mathrm{dist}}\left[
\sqrt{T}
\left(
\boldsymbol{\hat{Q}}^{\mathrm{E}}
-
\boldsymbol{Q}_0
\right)
\right]
&=
\frac{1}{\sqrt{T}}
\sum_{t=1}^{T}
\left\{
\dot{\Psi}_h^{\mathrm{dist}}\left[
\boldsymbol{\Xi}_t
\right]
+
\dot{\Psi}_h^{\mathrm{dist}}\left[
\boldsymbol{R}_t
\right]
\right\}
+
o_{\mathbb{P}}\left(1\right),
\\
\dot{\Psi}_h^{\mathrm{imp}}\left[
\sqrt{T}
\left(
\boldsymbol{\hat{Q}}^{\mathrm{E}}
-
\boldsymbol{Q}_0
\right)
\right]
&=
\frac{1}{\sqrt{T}}
\sum_{t=1}^{T}
\left\{
\dot{\Psi}_h^{\mathrm{imp}}\left[
\boldsymbol{\Xi}_t
\right]
+
\dot{\Psi}_h^{\mathrm{imp}}\left[
\boldsymbol{R}_t
\right]
\right\}
+
o_{\mathbb{P}}\left(1\right).
\end{aligned}
\label{eq:proof-main-quantile-terms}
\end{equation}
Using equation~\eqref{eq:influence-decomposition}, the sum of the two lines in equation~\eqref{eq:proof-main-quantile-terms} is
\begin{equation}
\frac{1}{\sqrt{T}}
\sum_{t=1}^{T}
\left(
Z_{h,t}^{\mathrm{res}}
+
Z_{h,t}^{\mathrm{dist}}
+
Z_{h,t}^{\mathrm{imp}}
\right)
+
o_{\mathbb{P}}\left(1\right).
\label{eq:proof-main-quantile-decomposition}
\end{equation}

\medskip\noindent\textit{Step 5: collect the observation-level terms.}\par
Substituting equations~\eqref{eq:proof-main-transition-term} and~\eqref{eq:proof-main-quantile-decomposition} into equation~\eqref{eq:proof-main-functional-expansion}, and then using equations~\eqref{eq:proof-main-sampling-simulation-decomposition}--\eqref{eq:proof-main-negligible-simulation}, gives
\[
\begin{aligned}
\sqrt{T}
\left(
\hat{\psi}_{h,S}^{\mathrm{E}}
-
\psi_h
\right)
&=
\frac{1}{\sqrt{T}}
\sum_{t=1}^{T}
\left(
Z_{h,t}^{\mathrm{tr}}
+
Z_{h,t}^{\mathrm{res}}
+
Z_{h,t}^{\mathrm{dist}}
+
Z_{h,t}^{\mathrm{imp}}
\right)
+
o_{\mathbb{P}}\left(1\right)
\\
&=
\frac{1}{\sqrt{T}}
\sum_{t=1}^{T}
Z_{h,t}
+
o_{\mathbb{P}}\left(1\right).
\end{aligned}
\]
This proves equation~\eqref{eq:empirical-irf-linearization}.

For a fixed collection of \(M\) responses, apply the preceding five steps componentwise. If \(r_{m,T}=o_{\mathbb{P}}\left(1\right)\) is the remainder for response \(m\), then
\[
\max_{1\leq m\leq M}
\left|
r_{m,T}
\right|
=
o_{\mathbb{P}}\left(1\right)
\]
because \(M\) is fixed. Stacking the componentwise expansions gives equation~\eqref{eq:joint-empirical-irf-linearization}.
\end{proof}

\begin{proof}[Proof of Corollary~\ref{cor:main-gaussian-limit}]

Equation~\eqref{eq:joint-empirical-irf-linearization} gives
\begin{equation}
\sqrt{T}
\left(
\boldsymbol{\hat{\psi}}_{S}^{\mathrm{E}}
-
\boldsymbol{\psi}
\right)
=
\frac{1}{\sqrt{T}}
\sum_{t=1}^{T}
\boldsymbol{Z}_t
+
o_{\mathbb{P}}\left(1\right).
\label{eq:proof-main-joint-linearization}
\end{equation}

Under Assumption~\ref{ass:first-order}\textup{(i)}, Proposition~\ref{prop:recursive-joint-clt} applies. Under its primitive strong-mixing conditions, the polynomial-moment calculation following \citet[Corollary~1.2, pp.~11--12]{Rio2017} verifies condition \(\mathrm{DMR}\). Corollary~1.2 then gives absolute convergence of the covariance series, and the multivariate central limit theorem in \citet[Corollary~4.1]{Rio2017} gives
\[
\frac{1}{\sqrt{T}}
\sum_{t=1}^{T}
\boldsymbol{Z}_t
\xrightarrow{\mathrm{d}}
\mathcal{N}\left(
\boldsymbol{0},
\boldsymbol{\Omega}
\right),
\qquad
\boldsymbol{\Omega}
=
\sum_{\ell=-\infty}^{\infty}
\operatorname{Cov}\left(
\boldsymbol{Z}_0,
\boldsymbol{Z}_{\ell}
\right).
\]
Slutsky's lemma, in the form stated by \citet[Lemma~2.8]{VanDerVaart1998}, applied to equation~\eqref{eq:proof-main-joint-linearization} proves equation~\eqref{eq:empirical-irf-limit}.

Suppose additionally that \(\left\{\boldsymbol{Z}_t,\mathcal{F}_t\right\}\) is a martingale difference sequence. For \(\ell\geq1\),
\[
\begin{aligned}
\operatorname{Cov}\left(
\boldsymbol{Z}_0,
\boldsymbol{Z}_{\ell}
\right)
&=
\mathbb{E}\left[
\boldsymbol{Z}_0
\boldsymbol{Z}_{\ell}^{\prime}
\right]
\\
&=
\mathbb{E}\left[
\boldsymbol{Z}_0
\mathbb{E}\left[
\left.
\boldsymbol{Z}_{\ell}^{\prime}
\right|
\mathcal{F}_{\ell-1}
\right]
\right]
=
\boldsymbol{0}.
\end{aligned}
\]
Stationarity gives the same conclusion for negative lags. Therefore,
\[
\boldsymbol{\Omega}
=
\operatorname{Var}\left(
\boldsymbol{Z}_t
\right)
=
\mathbb{E}\left[
\boldsymbol{Z}_t
\boldsymbol{Z}_t^{\prime}
\right].
\]
\end{proof}

\subsection{Smoothing comparison}
\label{app:proofs-smoothing}

\begin{proof}[Proof of Theorem~\ref{thm:smoothing-expansion}]
Fix one response specification and suppress its fixed arguments. Define
\[
a_T
=
T^{-1/2}+b_T^2,
\qquad
\mathcal{Y}_T
=
\sigma\left(
\boldsymbol{Y}_0,
\ldots,
\boldsymbol{Y}_T
\right),
\]
and let
\[
\boldsymbol{\Delta}_T
=
\boldsymbol{\hat{Q}}_{b_T}^{\mathrm{S}}
-
\boldsymbol{\hat{Q}}^{\mathrm{E}}
\]
denote the componentwise difference between the two estimated quantile functions.

Let
\[
\boldsymbol{Q}_0^{\prime\prime}
=
\left(
Q_{10}^{\prime\prime},
\ldots,
Q_{n0}^{\prime\prime}
\right)^{\prime}.
\]
Equation~\eqref{eq:quantile-smoothing-decomposition} gives
\[
\begin{aligned}
\boldsymbol{\Delta}_T
-
\frac{b_T^2}{2}
\boldsymbol{Q}_0^{\prime\prime}
&=
\left[
\mathcal{S}_{b_T}\boldsymbol{Q}_0
-
\boldsymbol{Q}_0
-
\frac{b_T^2}{2}
\boldsymbol{Q}_0^{\prime\prime}
\right] \\
&\quad+
\left(
\mathcal{S}_{b_T}
-
\operatorname{Id}
\right)
\left(
\boldsymbol{\hat{Q}}^{\mathrm{E}}
-
\boldsymbol{Q}_0
\right).
\end{aligned}
\]
Assumption~\ref{ass:smoothing} therefore implies
\[
\left\|
\boldsymbol{\Delta}_T
-
\frac{b_T^2}{2}
\boldsymbol{Q}_0^{\prime\prime}
\right\|
=
o\left(
b_T^2
\right)
+
o_{\mathbb{P}}\left(
T^{-1/2}
\right)
=
o_{\mathbb{P}}\left(
a_T
\right).
\]
Consequently,
\begin{equation}
\left\|
\boldsymbol{\Delta}_T
\right\|
=
O_{\mathbb{P}}\left(
a_T
\right),
\qquad
\left\|
\boldsymbol{\hat{Q}}_{b_T}^{\mathrm{S}}
-
\boldsymbol{Q}_0
\right\|
=
O_{\mathbb{P}}\left(
a_T
\right).
\label{eq:proof-smoothing-quantile-order}
\end{equation}

Lemma~\ref{lem:recursive-response-differentiability}, together with the second-order remainder bound used in its proof, gives
\[
\begin{aligned}
&\Psi_h\left(
\boldsymbol{\hat{\beta}},
\boldsymbol{\hat{Q}}_{b_T}^{\mathrm{S}}
\right)
-
\Psi_h\left(
\boldsymbol{\beta}_0,
\boldsymbol{Q}_0
\right) \\
&\quad=
\boldsymbol{A}_h^{\prime}
\left(
\boldsymbol{\hat{\beta}}
-
\boldsymbol{\beta}_0
\right)
+
\dot{\Psi}_h^{\mathrm{dist}}
\left[
\boldsymbol{\hat{Q}}_{b_T}^{\mathrm{S}}
-
\boldsymbol{Q}_0
\right]
+
\dot{\Psi}_h^{\mathrm{imp}}
\left[
\boldsymbol{\hat{Q}}_{b_T}^{\mathrm{S}}
-
\boldsymbol{Q}_0
\right]
+
o_{\mathbb{P}}\left(
a_T
\right)
\end{aligned}
\]
and
\[
\begin{aligned}
&\Psi_h\left(
\boldsymbol{\hat{\beta}},
\boldsymbol{\hat{Q}}^{\mathrm{E}}
\right)
-
\Psi_h\left(
\boldsymbol{\beta}_0,
\boldsymbol{Q}_0
\right) \\
&\quad=
\boldsymbol{A}_h^{\prime}
\left(
\boldsymbol{\hat{\beta}}
-
\boldsymbol{\beta}_0
\right)
+
\dot{\Psi}_h^{\mathrm{dist}}
\left[
\boldsymbol{\hat{Q}}^{\mathrm{E}}
-
\boldsymbol{Q}_0
\right]
+
\dot{\Psi}_h^{\mathrm{imp}}
\left[
\boldsymbol{\hat{Q}}^{\mathrm{E}}
-
\boldsymbol{Q}_0
\right]
+
o_{\mathbb{P}}\left(
T^{-1/2}
\right).
\end{aligned}
\]
Subtracting the second expansion from the first cancels the direct transition-parameter term:
\begin{equation}
\begin{aligned}
&\Psi_h\left(
\boldsymbol{\hat{\beta}},
\boldsymbol{\hat{Q}}_{b_T}^{\mathrm{S}}
\right)
-
\Psi_h\left(
\boldsymbol{\hat{\beta}},
\boldsymbol{\hat{Q}}^{\mathrm{E}}
\right) \\
&\quad=
\dot{\Psi}_h^{\mathrm{dist}}
\left[
\boldsymbol{\Delta}_T
\right]
+
\dot{\Psi}_h^{\mathrm{imp}}
\left[
\boldsymbol{\Delta}_T
\right]
+
o_{\mathbb{P}}\left(
a_T
\right).
\end{aligned}
\label{eq:proof-smoothing-integrated-linearization}
\end{equation}

The derivative maps are continuous under the weighted norm by Assumption~\ref{ass:response-tail-continuity}. Combining this continuity with the preceding expansion for \(\boldsymbol{\Delta}_T\) gives
\[
\begin{aligned}
&\dot{\Psi}_h^{\mathrm{dist}}
\left[
\boldsymbol{\Delta}_T
\right]
+
\dot{\Psi}_h^{\mathrm{imp}}
\left[
\boldsymbol{\Delta}_T
\right] \\
&\quad=
\frac{b_T^2}{2}
\left\{
\dot{\Psi}_h^{\mathrm{dist}}
\left[
\boldsymbol{Q}_0^{\prime\prime}
\right]
+
\dot{\Psi}_h^{\mathrm{imp}}
\left[
\boldsymbol{Q}_0^{\prime\prime}
\right]
\right\}
+
o_{\mathbb{P}}\left(
a_T
\right).
\end{aligned}
\]
Using equation~\eqref{eq:propagation-weights},
\[
\begin{aligned}
&\frac{1}{2}
\left\{
\dot{\Psi}_h^{\mathrm{dist}}
\left[
\boldsymbol{Q}_0^{\prime\prime}
\right]
+
\dot{\Psi}_h^{\mathrm{imp}}
\left[
\boldsymbol{Q}_0^{\prime\prime}
\right]
\right\} \\
&\quad=
\frac{1}{2}
\left[
\sum_{j=1}^{n}
\int_{0}^{1}
\omega_{h,j}^{\mathrm{dist}}\left(
p
\right)
Q_{j0}^{\prime\prime}\left(
p
\right)
\,\dif p
+
\int_{0}^{1}
\omega_{h,k}^{\mathrm{imp}}\left(
p
\right)
Q_{k0}^{\prime\prime}\left(
\tau_{\delta}\left(
p
\right)
\right)
\,\dif p
\right] \\
&\quad=
B_h^{\mathrm{S}}.
\end{aligned}
\]
Substituting this identity into equation~\eqref{eq:proof-smoothing-integrated-linearization} yields
\begin{equation}
\Psi_h\left(
\boldsymbol{\hat{\beta}},
\boldsymbol{\hat{Q}}_{b_T}^{\mathrm{S}}
\right)
-
\Psi_h\left(
\boldsymbol{\hat{\beta}},
\boldsymbol{\hat{Q}}^{\mathrm{E}}
\right)
=
b_T^2B_h^{\mathrm{S}}
+
o_{\mathbb{P}}\left(
a_T
\right).
\label{eq:proof-smoothing-integrated-bias}
\end{equation}
Stopping before the second-order substitution gives the operator-level bias stated in Remark~\ref{rem:additional-theorem-conditions}.

For simulation draw \(s\), define
\[
\begin{aligned}
\Delta D_{h,s,T}
&=
D_h\left(
\boldsymbol{P}_{s,1:h};
\boldsymbol{\hat{\beta}},
\boldsymbol{\hat{Q}}_{b_T}^{\mathrm{S}},
\boldsymbol{y},
\boldsymbol{a},
k,
\delta
\right) \\
&\quad-
D_h\left(
\boldsymbol{P}_{s,1:h};
\boldsymbol{\hat{\beta}},
\boldsymbol{\hat{Q}}^{\mathrm{E}},
\boldsymbol{y},
\boldsymbol{a},
k,
\delta
\right).
\end{aligned}
\]
Because both terms use the same rank path,
\[
\begin{aligned}
\mathbb{E}\left[
\left.
\Delta D_{h,s,T}
\right|
\mathcal{Y}_T
\right]
&=
\Psi_h\left(
\boldsymbol{\hat{\beta}},
\boldsymbol{\hat{Q}}_{b_T}^{\mathrm{S}}
\right)
-
\Psi_h\left(
\boldsymbol{\hat{\beta}},
\boldsymbol{\hat{Q}}^{\mathrm{E}}
\right).
\end{aligned}
\]
The \(L^2\) part of Lemma~\ref{lem:recursive-response-differentiability}, applied along the segment joining the two quantile estimates, gives a \(\mathcal{Y}_T\)-measurable sequence \(C_T=O_{\mathbb{P}}\left(1\right)\) such that
\[
\mathbb{E}\left[
\left.
\left(
\Delta D_{h,s,T}
\right)^2
\right|
\mathcal{Y}_T
\right]
\leq
C_T
\left\|
\boldsymbol{\Delta}_T
\right\|^2.
\]
Conditional independence across simulation draws and equation~\eqref{eq:proof-smoothing-quantile-order} imply
\[
\begin{aligned}
&\operatorname{Var}\left[
\left.
\frac{1}{S}
\sum_{s=1}^{S}
\left\{
\Delta D_{h,s,T}
-
\mathbb{E}\left[
\left.
\Delta D_{h,s,T}
\right|
\mathcal{Y}_T
\right]
\right\}
\right|
\mathcal{Y}_T
\right] \\
&\quad\leq
\frac{C_T}{S}
\left\|
\boldsymbol{\Delta}_T
\right\|^2
=
O_{\mathbb{P}}\left(
S^{-1}a_T^2
\right).
\end{aligned}
\]
Conditional Chebyshev's inequality therefore gives
\begin{equation}
\frac{1}{S}
\sum_{s=1}^{S}
\left\{
\Delta D_{h,s,T}
-
\mathbb{E}\left[
\left.
\Delta D_{h,s,T}
\right|
\mathcal{Y}_T
\right]
\right\}
=
O_{\mathbb{P}}\left(
S^{-1/2}a_T
\right).
\label{eq:proof-smoothing-paired-simulation}
\end{equation}
Combining equations~\eqref{eq:proof-smoothing-integrated-bias} and~\eqref{eq:proof-smoothing-paired-simulation} gives
\[
\begin{aligned}
\hat{\psi}_{h,S}^{\mathrm{S}}
-
\hat{\psi}_{h,S}^{\mathrm{E}}
&=
b_T^2B_h^{\mathrm{S}}
+
o_{\mathbb{P}}\left(
T^{-1/2}+b_T^2
\right) \\
&\quad+
O_{\mathbb{P}}\left[
S^{-1/2}
\left(
T^{-1/2}+b_T^2
\right)
\right].
\end{aligned}
\]
This proves equations~\eqref{eq:smoothing-expansion} and~\eqref{eq:smoothing-remainder}.

For a fixed collection of \(M\) responses, the preceding bounds apply componentwise. Since \(M\) is fixed, the maximum of the componentwise \(o_{\mathbb{P}}\left(a_T\right)\) remainders remains \(o_{\mathbb{P}}\left(a_T\right)\), while the conditional covariance matrix of the paired Monte Carlo term is \(O_{\mathbb{P}}\left(S^{-1}a_T^2\right)\). Stacking the componentwise expansions proves the joint result.
\end{proof}

\begin{proof}[Proof of Corollary~\ref{cor:smoothing-regimes}]
Let \(\boldsymbol{R}_{T,S}^{\mathrm{S}}\) stack the remainders in equation~\eqref{eq:smoothing-remainder}. Multiplying the joint version of equation~\eqref{eq:smoothing-expansion} by \(\sqrt{T}\) gives
\begin{equation}
\sqrt{T}
\left(
\boldsymbol{\hat{\psi}}_{S}^{\mathrm{S}}
-
\boldsymbol{\hat{\psi}}_{S}^{\mathrm{E}}
\right)
=
\lambda_T
\boldsymbol{B}^{\mathrm{S}}
+
\sqrt{T}
\boldsymbol{R}_{T,S}^{\mathrm{S}},
\label{eq:proof-smoothing-root-t-difference}
\end{equation}
where
\begin{equation}
\sqrt{T}
\boldsymbol{R}_{T,S}^{\mathrm{S}}
=
o_{\mathbb{P}}\left(
1+\lambda_T
\right)
+
O_{\mathbb{P}}\left[
S^{-1/2}
\left(
1+\lambda_T
\right)
\right].
\label{eq:proof-smoothing-root-t-remainder}
\end{equation}

If \(\lambda_T\rightarrow0\), equation~\eqref{eq:proof-smoothing-root-t-remainder} and \(S\rightarrow\infty\) give
\[
\sqrt{T}
\boldsymbol{R}_{T,S}^{\mathrm{S}}
=
o_{\mathbb{P}}\left(
1
\right).
\]
Equation~\eqref{eq:proof-smoothing-root-t-difference} therefore implies
\[
\sqrt{T}
\left(
\boldsymbol{\hat{\psi}}_{S}^{\mathrm{S}}
-
\boldsymbol{\hat{\psi}}_{S}^{\mathrm{E}}
\right)
\xrightarrow{\mathbb{P}}
\boldsymbol{0}.
\]
Consequently,
\[
\sqrt{T}
\left(
\boldsymbol{\hat{\psi}}_{S}^{\mathrm{S}}
-
\boldsymbol{\psi}
\right)
=
\sqrt{T}
\left(
\boldsymbol{\hat{\psi}}_{S}^{\mathrm{E}}
-
\boldsymbol{\psi}
\right)
+
o_{\mathbb{P}}\left(
1
\right).
\]
Corollaries~\ref{cor:main-gaussian-limit} and~\ref{cor:feasible-wald-inference} then give the same first-order distribution and the same consistent covariance estimator.

If \(\lambda_T\rightarrow\lambda\in\left(0,\infty\right)\), equation~\eqref{eq:proof-smoothing-root-t-remainder} gives
\[
\sqrt{T}
\boldsymbol{R}_{T,S}^{\mathrm{S}}
=
o_{\mathbb{P}}\left(
1
\right).
\]
Hence,
\[
\begin{aligned}
\sqrt{T}
\left(
\boldsymbol{\hat{\psi}}_{S}^{\mathrm{S}}
-
\boldsymbol{\psi}
\right)
&=
\sqrt{T}
\left(
\boldsymbol{\hat{\psi}}_{S}^{\mathrm{E}}
-
\boldsymbol{\psi}
\right)
+
\lambda_T
\boldsymbol{B}^{\mathrm{S}}
+
o_{\mathbb{P}}\left(
1
\right).
\end{aligned}
\]
Corollary~\ref{cor:main-gaussian-limit} and Slutsky's lemma, in the form stated by \citet[Lemma~2.8]{VanDerVaart1998}, give
\[
\sqrt{T}
\left(
\boldsymbol{\hat{\psi}}_{S}^{\mathrm{S}}
-
\boldsymbol{\psi}
\right)
\xrightarrow{\mathrm{d}}
\mathcal{N}\left(
\lambda\boldsymbol{B}^{\mathrm{S}},
\boldsymbol{\Omega}
\right),
\]
which proves equation~\eqref{eq:smoothed-local-bias-limit}.

Suppose \(\lambda_T\rightarrow\infty\) and \(B_h^{\mathrm{S}}\neq0\). Dividing equation~\eqref{eq:smoothing-remainder} by \(b_T^2\) gives
\[
\begin{aligned}
b_T^{-2}R_{h,T,S}^{\mathrm{S}}
&=
o_{\mathbb{P}}\left(
1+\lambda_T^{-1}
\right)
+
O_{\mathbb{P}}\left[
S^{-1/2}
\left(
1+\lambda_T^{-1}
\right)
\right] \\
&=
o_{\mathbb{P}}\left(
1
\right).
\end{aligned}
\]
Corollary~\ref{cor:main-gaussian-limit} gives
\[
\sqrt{T}
\left(
\hat{\psi}_{h,S}^{\mathrm{E}}
-
\psi_h
\right)
=
O_{\mathbb{P}}\left(
1
\right),
\]
and therefore
\[
b_T^{-2}
\left(
\hat{\psi}_{h,S}^{\mathrm{E}}
-
\psi_h
\right)
=
\lambda_T^{-1}
O_{\mathbb{P}}\left(
1
\right)
=
o_{\mathbb{P}}\left(
1
\right).
\]
Using equation~\eqref{eq:smoothing-expansion},
\[
\begin{aligned}
b_T^{-2}
\left(
\hat{\psi}_{h,S}^{\mathrm{S}}
-
\psi_h
\right)
&=
B_h^{\mathrm{S}}
+
b_T^{-2}R_{h,T,S}^{\mathrm{S}}
+
b_T^{-2}
\left(
\hat{\psi}_{h,S}^{\mathrm{E}}
-
\psi_h
\right) \\
&=
B_h^{\mathrm{S}}
+
o_{\mathbb{P}}\left(
1
\right),
\end{aligned}
\]
which proves equation~\eqref{eq:smoothing-bias-dominates}.
\end{proof}

\begin{proof}[Derivation of equation~\ref{eq:smoothed-wald-coverage}]
Let
\[
\hat{\sigma}_h^2
=
\hat{\Omega}_{hh}.
\]
Corollary~\ref{cor:feasible-wald-inference} gives
\[
\hat{\sigma}_h
\xrightarrow{\mathbb{P}}
\sigma_h.
\]
Under \(\lambda_T\rightarrow\lambda\in\left(0,\infty\right)\), Part~\textup{(ii)} of Corollary~\ref{cor:smoothing-regimes} and Slutsky's lemma, as stated by \citet[Lemma~2.8]{VanDerVaart1998}, give
\[
\frac{
\sqrt{T}
\left(
\hat{\psi}_{h,S}^{\mathrm{S}}
-
\psi_h
\right)
}{
\hat{\sigma}_h
}
\xrightarrow{\mathrm{d}}
Z
+
\frac{
\lambda B_h^{\mathrm{S}}
}{
\sigma_h
},
\qquad
Z
\sim
\mathcal{N}\left(
0,
1
\right).
\]
The Wald interval covers \(\psi_h\) when the absolute value of the statistic on the left does not exceed \(z_{1-\alpha/2}\). Its limiting coverage is therefore
\[
\begin{aligned}
&\Pr\left(
-z_{1-\alpha/2}
\leq
Z
+
\frac{
\lambda B_h^{\mathrm{S}}
}{
\sigma_h
}
\leq
z_{1-\alpha/2}
\right) \\
&\quad=
\Phi\left(
z_{1-\alpha/2}
-
\frac{
\lambda B_h^{\mathrm{S}}
}{
\sigma_h
}
\right)
-
\Phi\left(
-z_{1-\alpha/2}
-
\frac{
\lambda B_h^{\mathrm{S}}
}{
\sigma_h
}
\right),
\end{aligned}
\]
which proves equation~\eqref{eq:smoothed-wald-coverage}.

Let
\[
c
=
\frac{
\lambda B_h^{\mathrm{S}}
}{
\sigma_h
}.
\]
The coverage expression is even in \(c\). For \(c>0\), its derivative is
\[
\phi\left(
z_{1-\alpha/2}+c
\right)
-
\phi\left(
z_{1-\alpha/2}-c
\right)
<
0.
\]
Hence the coverage is strictly below \(1-\alpha\) whenever \(c\neq0\).

Under \(\lambda_T\rightarrow\infty\), equation~\eqref{eq:smoothing-bias-dominates} gives
\[
\left|
\frac{
\sqrt{T}
\left(
\hat{\psi}_{h,S}^{\mathrm{S}}
-
\psi_h
\right)
}{
\hat{\sigma}_h
}
\right|
=
\frac{
\lambda_T
}{
\hat{\sigma}_h
}
\left|
B_h^{\mathrm{S}}
+
o_{\mathbb{P}}\left(
1
\right)
\right|
\xrightarrow{\mathbb{P}}
\infty.
\]
Therefore, the coverage probability of the uncorrected interval converges to zero.
\end{proof}

\begin{proof}[Proof of Corollary~\ref{cor:smoothing-special-cases}]

For a quantile perturbation \(\boldsymbol{q}\), the pathwise response derivative is
\[
\begin{aligned}
\dot{D}_h^{\xi}\left[
\boldsymbol{q}
\right]
&=
\sum_{j=1}^{n}
\left(
\Lambda_{h,1,j}^{\xi}
-
\Lambda_{h,1,j}^{0}
\right)
q_j\left(
P_{j1}
\right) \\
&\quad+
\sum_{s=2}^{h}
\sum_{j=1}^{n}
\left(
\Lambda_{h,s,j}^{\xi}
-
\Lambda_{h,s,j}^{0}
\right)
q_j\left(
P_{js}
\right).
\end{aligned}
\]
The impact displacement \(\xi\boldsymbol{e}_k\) is fixed, so no quantile perturbation is evaluated at a shifted rank. Taking expectations and conditioning on the corresponding rank gives
\[
\begin{aligned}
\dot{\Psi}_h^{\xi}\left[
\boldsymbol{q}
\right]
&=
\sum_{j=1}^{n}
\int_{0}^{1}
\left[
\mathbb{E}\left[
\left.
\Lambda_{h,1,j}^{\xi}
-
\Lambda_{h,1,j}^{0}
\right|
P_{j1}=p
\right] \right. \\
&\qquad\left.
+
\sum_{s=2}^{h}
\mathbb{E}\left[
\left.
\Lambda_{h,s,j}^{\xi}
-
\Lambda_{h,s,j}^{0}
\right|
P_{js}=p
\right]
\right]
q_j\left(
p
\right)
\,\dif p.
\end{aligned}
\]
Applying Assumption~\ref{ass:smoothing} with
\[
q_j
=
\frac{1}{2}
Q_{j0}^{\prime\prime}
\]
gives the leading smoothing coefficient for a fixed additive shock.

At impact,
\[
\boldsymbol{Y}_1^{\delta}
-
\boldsymbol{Y}_1^{0}
=
\boldsymbol{B}_0
\boldsymbol{e}_k
\left[
Q_{k0}\left(
\tau_{\delta}\left(
P_{k1}
\right)
\right)
-
Q_{k0}\left(
P_{k1}
\right)
\right].
\]
For \(j=2,\ldots,h\), the common future innovations cancel:
\[
\boldsymbol{Y}_j^{\delta}
-
\boldsymbol{Y}_j^{0}
=
\boldsymbol{A}_0
\left(
\boldsymbol{Y}_{j-1}^{\delta}
-
\boldsymbol{Y}_{j-1}^{0}
\right).
\]
Iteration gives
\[
\begin{aligned}
&D_h\left(
\boldsymbol{P}_{1:h};
\boldsymbol{\beta}_0,
\boldsymbol{Q}_0
\right) \\
&\quad=
\boldsymbol{a}^{\prime}
\boldsymbol{A}_0^{h-1}
\boldsymbol{B}_0
\boldsymbol{e}_k
\left[
Q_{k0}\left(
\tau_{\delta}\left(
P_{k1}
\right)
\right)
-
Q_{k0}\left(
P_{k1}
\right)
\right].
\end{aligned}
\]
Therefore, for a quantile perturbation \(\boldsymbol{q}\),
\[
\dot{\Psi}_h\left[
\boldsymbol{q}
\right]
=
\boldsymbol{a}^{\prime}
\boldsymbol{A}_0^{h-1}
\boldsymbol{B}_0
\boldsymbol{e}_k
\int_{0}^{1}
\left[
q_k\left(
\tau_{\delta}\left(
p
\right)
\right)
-
q_k\left(
p
\right)
\right]
\,\dif p.
\]
Substituting
\[
q_k
=
\frac{1}{2}
Q_{k0}^{\prime\prime}
\]
proves equation~\eqref{eq:affine-smoothing-bias}.

Under a fixed additive shock, the path response at a generic common transition parameter is
\[
D_h^{\xi}
=
\xi
\boldsymbol{a}^{\prime}
\boldsymbol{A}^{h-1}
\boldsymbol{B}
\boldsymbol{e}_k.
\]
This expression does not depend on the innovation quantile functions or the simulation ranks. Hence the empirical and smoothed estimators coincide when they use the same transition estimate.
\end{proof}

\subsection{Bootstrap and finite-simulation results}
\label{app:proofs-inference}

\begin{proof}[Proof of Theorem~\ref{thm:bootstrap-validity}]
Fix one response specification. The vector result follows by stacking the same argument over the fixed response collection. All conditional stochastic orders below are understood in probability with respect to the original sample.

Add and subtract the exactly integrated original and bootstrap responses:
\[
\begin{aligned}
&\sqrt{T}
\left(
\hat{\psi}_{h,S}^{\mathrm{E},*}
-
\hat{\psi}_{h,S}^{\mathrm{E}}
\right)
\\
&\quad=
\sqrt{T}
\left[
\Psi_{h,S}\left(
\boldsymbol{\hat{\beta}}^{*},
\boldsymbol{\hat{Q}}^{\mathrm{E},*}
\right)
-
\Psi_h\left(
\boldsymbol{\hat{\beta}}^{*},
\boldsymbol{\hat{Q}}^{\mathrm{E},*}
\right)
\right]
\\
&\qquad-
\sqrt{T}
\left[
\Psi_{h,S}\left(
\boldsymbol{\hat{\beta}},
\boldsymbol{\hat{Q}}^{\mathrm{E}}
\right)
-
\Psi_h\left(
\boldsymbol{\hat{\beta}},
\boldsymbol{\hat{Q}}^{\mathrm{E}}
\right)
\right]
\\
&\qquad+
\sqrt{T}
\left[
\Psi_h\left(
\boldsymbol{\hat{\beta}}^{*},
\boldsymbol{\hat{Q}}^{\mathrm{E},*}
\right)
-
\Psi_h\left(
\boldsymbol{\hat{\beta}},
\boldsymbol{\hat{Q}}^{\mathrm{E}}
\right)
\right].
\end{aligned}
\]
The proof of Lemma~\ref{lem:negligible-simulation-error}, applied uniformly over the neighborhoods in Assumption~\ref{ass:bootstrap}\textup{(i)}, gives
\[
\sqrt{T}
\left[
\Psi_{h,S}\left(
\boldsymbol{\hat{\beta}}^{*},
\boldsymbol{\hat{Q}}^{\mathrm{E},*}
\right)
-
\Psi_h\left(
\boldsymbol{\hat{\beta}}^{*},
\boldsymbol{\hat{Q}}^{\mathrm{E},*}
\right)
\right]
=
o_{\mathbb{P}^{*}}\left(1\right)
\]
in probability. The same lemma gives
\[
\sqrt{T}
\left[
\Psi_{h,S}\left(
\boldsymbol{\hat{\beta}},
\boldsymbol{\hat{Q}}^{\mathrm{E}}
\right)
-
\Psi_h\left(
\boldsymbol{\hat{\beta}},
\boldsymbol{\hat{Q}}^{\mathrm{E}}
\right)
\right]
=
o_{\mathbb{P}}\left(1\right).
\]
Both conclusions use \(S/T\rightarrow\infty\). Therefore,
\begin{equation}
\sqrt{T}
\left(
\hat{\psi}_{h,S}^{\mathrm{E},*}
-
\hat{\psi}_{h,S}^{\mathrm{E}}
\right)
=
\sqrt{T}
\left[
\Psi_h\left(
\boldsymbol{\hat{\beta}}^{*},
\boldsymbol{\hat{Q}}^{\mathrm{E},*}
\right)
-
\Psi_h\left(
\boldsymbol{\hat{\beta}},
\boldsymbol{\hat{Q}}^{\mathrm{E}}
\right)
\right]
+
o_{\mathbb{P}^{*}}\left(1\right).
\label{eq:proof-bootstrap-remove-simulation}
\end{equation}

Assumption~\ref{ass:bootstrap}\textup{(ii)--(iii)} implies
\[
\sqrt{T}
\left[
\left\|
\boldsymbol{\hat{\beta}}^{*}
-
\boldsymbol{\hat{\beta}}
\right\|
+
\left\|
\boldsymbol{\hat{Q}}^{\mathrm{E},*}
-
\boldsymbol{\hat{Q}}^{\mathrm{E}}
\right\|
\right]
=
O_{\mathbb{P}^{*}}\left(1\right).
\]
Let \(\boldsymbol{A}_{h,T}\), \(\dot{\Psi}_{h,T}^{\mathrm{dist}}\), and \(\dot{\Psi}_{h,T}^{\mathrm{imp}}\) denote the derivative objects evaluated at \(\left(\boldsymbol{\hat{\beta}},\boldsymbol{\hat{Q}}^{\mathrm{E}}\right)\). Lemma~\ref{lem:recursive-response-differentiability} establishes the required Hadamard differentiability, and the conditional delta method in \citet[Theorem~3.10.11]{VanDerVaartWellner2023} gives
\begin{equation}
\begin{aligned}
&\sqrt{T}
\left[
\Psi_h\left(
\boldsymbol{\hat{\beta}}^{*},
\boldsymbol{\hat{Q}}^{\mathrm{E},*}
\right)
-
\Psi_h\left(
\boldsymbol{\hat{\beta}},
\boldsymbol{\hat{Q}}^{\mathrm{E}}
\right)
\right]
\\
&\quad=
\boldsymbol{A}_{h,T}^{\prime}
\sqrt{T}
\left(
\boldsymbol{\hat{\beta}}^{*}
-
\boldsymbol{\hat{\beta}}
\right)
+
\dot{\Psi}_{h,T}^{\mathrm{dist}}
\left[
\sqrt{T}
\left(
\boldsymbol{\hat{Q}}^{\mathrm{E},*}
-
\boldsymbol{\hat{Q}}^{\mathrm{E}}
\right)
\right]
\\
&\qquad+
\dot{\Psi}_{h,T}^{\mathrm{imp}}
\left[
\sqrt{T}
\left(
\boldsymbol{\hat{Q}}^{\mathrm{E},*}
-
\boldsymbol{\hat{Q}}^{\mathrm{E}}
\right)
\right]
+
o_{\mathbb{P}^{*}}\left(1\right).
\end{aligned}
\label{eq:proof-bootstrap-functional-expansion}
\end{equation}
Uniform differentiability and consistency imply
\[
\boldsymbol{A}_{h,T}
\xrightarrow{\mathbb{P}}
\boldsymbol{A}_h,
\qquad
\left\|
\dot{\Psi}_{h,T}^{\mathrm{dist}}
-
\dot{\Psi}_{h}^{\mathrm{dist}}
\right\|_{\mathrm{op}}
+
\left\|
\dot{\Psi}_{h,T}^{\mathrm{imp}}
-
\dot{\Psi}_{h}^{\mathrm{imp}}
\right\|_{\mathrm{op}}
\xrightarrow{\mathbb{P}}
0.
\]

Collect the bootstrap empirical-quantile contributions in
\[
\boldsymbol{\Xi}_t^{*}
=
\left(
\Xi_{1,t}^{*},
\ldots,
\Xi_{n,t}^{*}
\right)^{\prime}
\]
and define
\[
\boldsymbol{R}_t^{*}
=
\left(
p\mapsto
\boldsymbol{r}_{1}^{*}\left(p\right)^{\prime}
\boldsymbol{L}_t^{*},
\ldots,
p\mapsto
\boldsymbol{r}_{n}^{*}\left(p\right)^{\prime}
\boldsymbol{L}_t^{*}
\right)^{\prime}.
\]
Assumption~\ref{ass:bootstrap}\textup{(ii)--(iii)} gives
\[
\sqrt{T}
\left(
\boldsymbol{\hat{\beta}}^{*}
-
\boldsymbol{\hat{\beta}}
\right)
=
\frac{1}{\sqrt{T}}
\sum_{t=1}^{T}
\boldsymbol{L}_t^{*}
+
o_{\mathbb{P}^{*}}\left(1\right)
\]
and
\[
\sqrt{T}
\left(
\boldsymbol{\hat{Q}}^{\mathrm{E},*}
-
\boldsymbol{\hat{Q}}^{\mathrm{E}}
\right)
=
\frac{1}{\sqrt{T}}
\sum_{t=1}^{T}
\left(
\boldsymbol{\Xi}_t^{*}
+
\boldsymbol{R}_t^{*}
\right)
+
o_{\mathbb{P}^{*}}\left(1\right)
\]
in the product weighted norm. Substituting these expansions into equation~\eqref{eq:proof-bootstrap-functional-expansion} yields
\[
\begin{aligned}
&\sqrt{T}
\left[
\Psi_h\left(
\boldsymbol{\hat{\beta}}^{*},
\boldsymbol{\hat{Q}}^{\mathrm{E},*}
\right)
-
\Psi_h\left(
\boldsymbol{\hat{\beta}},
\boldsymbol{\hat{Q}}^{\mathrm{E}}
\right)
\right]
\\
&\quad=
\frac{1}{\sqrt{T}}
\sum_{t=1}^{T}
\left\{
\boldsymbol{A}_{h,T}^{\prime}
\boldsymbol{L}_t^{*}
+
\dot{\Psi}_{h,T}^{\mathrm{dist}}
\left[
\boldsymbol{R}_t^{*}
\right]
+
\dot{\Psi}_{h,T}^{\mathrm{imp}}
\left[
\boldsymbol{R}_t^{*}
\right]
\right.
\\
&\qquad\left.
+
\dot{\Psi}_{h,T}^{\mathrm{dist}}
\left[
\boldsymbol{\Xi}_t^{*}
\right]
+
\dot{\Psi}_{h,T}^{\mathrm{imp}}
\left[
\boldsymbol{\Xi}_t^{*}
\right]
\right\}
+
o_{\mathbb{P}^{*}}\left(1\right).
\end{aligned}
\]
Define
\[
\begin{aligned}
Z_{h,t}^{\mathrm{tr},*}
&=
\boldsymbol{A}_{h,T}^{\prime}
\boldsymbol{L}_t^{*},
&
Z_{h,t}^{\mathrm{res},*}
&=
\dot{\Psi}_{h,T}^{\mathrm{dist}}
\left[
\boldsymbol{R}_t^{*}
\right]
+
\dot{\Psi}_{h,T}^{\mathrm{imp}}
\left[
\boldsymbol{R}_t^{*}
\right],
\\
Z_{h,t}^{\mathrm{dist},*}
&=
\dot{\Psi}_{h,T}^{\mathrm{dist}}
\left[
\boldsymbol{\Xi}_t^{*}
\right],
&
Z_{h,t}^{\mathrm{imp},*}
&=
\dot{\Psi}_{h,T}^{\mathrm{imp}}
\left[
\boldsymbol{\Xi}_t^{*}
\right].
\end{aligned}
\]
Combining the preceding display with equation~\eqref{eq:proof-bootstrap-remove-simulation} proves equation~\eqref{eq:bootstrap-linearization} for one response. Componentwise stacking proves the stated expansion for the fixed response vector.

Let \(\operatorname{BL}_1\left(\mathbb{R}^{M}\right)\) denote the class of functions bounded by one with Lipschitz constant at most one. Assumption~\ref{ass:bootstrap}\textup{(iii)} gives
\[
\sup_{f\in\operatorname{BL}_1\left(\mathbb{R}^{M}\right)}
\left|
\mathbb{E}^{*}
\left[
f\left(
\frac{1}{\sqrt{T}}
\sum_{t=1}^{T}
\boldsymbol{Z}_t^{*}
\right)
\right]
-
\mathbb{E}
\left[
f\left(
\boldsymbol{G}
\right)
\right]
\right|
\xrightarrow{\mathbb{P}}
0,
\qquad
\boldsymbol{G}
\sim
\mathcal{N}\left(
\boldsymbol{0},
\boldsymbol{\Omega}
\right).
\]
Let \(\boldsymbol{r}_T^{*}=o_{\mathbb{P}^{*}}\left(1\right)\) denote the remainder in equation~\eqref{eq:bootstrap-linearization}. Then
\[
\sup_{f\in\operatorname{BL}_1\left(\mathbb{R}^{M}\right)}
\left|
\mathbb{E}^{*}
\left[
f\left(
\frac{1}{\sqrt{T}}
\sum_{t=1}^{T}
\boldsymbol{Z}_t^{*}
+
\boldsymbol{r}_T^{*}
\right)
-
f\left(
\frac{1}{\sqrt{T}}
\sum_{t=1}^{T}
\boldsymbol{Z}_t^{*}
\right)
\right]
\right|
\leq
\mathbb{E}^{*}
\left[
\min\left\{
\left\|
\boldsymbol{r}_T^{*}
\right\|,
2
\right\}
\right]
=
o_{\mathbb{P}}\left(1\right).
\]
This proves equation~\eqref{eq:bootstrap-conditional-limit}.

Corollary~\ref{cor:feasible-wald-inference} gives
\[
\hat{\sigma}_h^2
\xrightarrow{\mathbb{P}}
\Omega_{hh}.
\]
Assumption~\ref{ass:bootstrap}\textup{(iii)} and the bootstrap version of the covariance-consistency argument give
\[
\left(
\hat{\sigma}_h^{*}
\right)^2
\xrightarrow{\mathbb{P}^{*}}
\Omega_{hh}
\]
in probability. If \(\Omega_{hh}>0\), the conditional delta method in \citet[Theorem~3.10.11]{VanDerVaartWellner2023} applied to the division map, together with Slutsky's lemma for the original statistic as stated in \citet[Lemma~2.8]{VanDerVaart1998}, gives
\[
R_{h,T}^{*}
\xrightarrow{\mathrm{d}^{*}}
\mathcal{N}\left(0,1\right)
\quad
\text{in probability},
\qquad
R_{h,T}
\xrightarrow{\mathrm{d}}
\mathcal{N}\left(0,1\right).
\]
Because the standard normal distribution function is continuous, the conditional analogue of P\'olya's lemma, obtained by applying the monotonicity argument in \citet[Lemma~2.11]{VanDerVaart1998} along almost-surely convergent subsequences, gives
\[
\sup_{x\in\mathbb{R}}
\left|
\Pr^{*}\left(
R_{h,T}^{*}
\leq x
\right)
-
\Phi\left(x\right)
\right|
\xrightarrow{\mathbb{P}}
0.
\]
The unconditional form of \citet[Lemma~2.11]{VanDerVaart1998} gives
\[
\sup_{x\in\mathbb{R}}
\left|
\Pr\left(
R_{h,T}
\leq x
\right)
-
\Phi\left(x\right)
\right|
\rightarrow
0.
\]
The triangle inequality proves the scalar studentization claim.

Suppose every diagonal element of \(\boldsymbol{\Omega}\) is positive. Define
\[
\mathcal{R}_T
=
\max_{1\leq m\leq M}
\left|
\frac{
\sqrt{T}
\left(
\hat{\psi}_{m,S}^{\mathrm{E}}
-
\psi_m
\right)
}{
\hat{\sigma}_m
}
\right|,
\qquad
\mathcal{R}_T^{*}
=
\max_{1\leq m\leq M}
\left|
\frac{
\sqrt{T}
\left(
\hat{\psi}_{m,S}^{\mathrm{E},*}
-
\hat{\psi}_{m,S}^{\mathrm{E}}
\right)
}{
\hat{\sigma}_m^{*}
}
\right|.
\]
The continuous-mapping theorem applied to the joint original and conditional bootstrap limits gives convergence of both statistics to
\[
\max_{1\leq m\leq M}
\left|
\frac{G_m}{\sqrt{\Omega_{mm}}}
\right|,
\qquad
\boldsymbol{G}
\sim
\mathcal{N}\left(
\boldsymbol{0},
\boldsymbol{\Omega}
\right).
\]
For each \(c\geq0\), the boundary of the event that this maximum is at most \(c\) is contained in the finite union
\[
\bigcup_{m=1}^{M}
\left[
\left\{
G_m
=
c\sqrt{\Omega_{mm}}
\right\}
\cup
\left\{
G_m
=
-c\sqrt{\Omega_{mm}}
\right\}
\right],
\]
which has probability zero. The limiting distribution is therefore continuous. Applying the preceding conditional P\'olya argument to \(\mathcal{R}_T^{*}\) and the unconditional argument to \(\mathcal{R}_T\) proves validity of the maximum absolute studentized statistic.
\end{proof}

\begin{proof}[Proof of Proposition~\ref{prop:path-only-resampling}]
Let
\[
D_{h,T}^{\dagger}
=
D_h\left(
\boldsymbol{P}_{1:h}^{\dagger};
\boldsymbol{\hat{\beta}},
\boldsymbol{\hat{Q}}^{\mathrm{E}},
\boldsymbol{y},
\boldsymbol{a},
k,
\delta
\right)
\]
and let \(D_{h,0}^{\dagger}\) denote the same path response evaluated at \(\boldsymbol{\beta}_0\) and \(\boldsymbol{Q}_0\). Conditional on \(\mathcal{F}_T\), the path-only draws are independent copies of \(D_{h,T}^{\dagger}\), with
\[
\mathbb{E}^{\dagger}
\left[
D_{h,T}^{\dagger}
\mathrel{\big|}
\mathcal{F}_T
\right]
=
\hat{\psi}_h^{\infty}.
\]

The \(L^2\) part of Lemma~\ref{lem:recursive-response-differentiability}, consistency of \(\boldsymbol{\hat{\beta}}\) and \(\boldsymbol{\hat{Q}}^{\mathrm{E}}\), and the maintained moment envelope give
\[
\mathbb{E}^{\dagger}
\left[
\left.
\left|
D_{h,T}^{\dagger}
-
D_{h,0}^{\dagger}
\right|^2
\right|
\mathcal{F}_T
\right]
\xrightarrow{\mathbb{P}}
0.
\]
Consequently,
\[
\hat{\psi}_h^{\infty}
\xrightarrow{\mathbb{P}}
\psi_h
\]
and
\begin{equation}
\operatorname{Var}^{\dagger}
\left(
D_{h,T}^{\dagger}
\mathrel{\big|}
\mathcal{F}_T
\right)
\xrightarrow{\mathbb{P}}
\operatorname{Var}\left(
D_{h,0}^{\dagger}
\right)
=
\mathcal{V}_h.
\label{eq:proof-path-only-variance}
\end{equation}

Define
\[
X_{h,T}^{\dagger}
=
D_{h,T}^{\dagger}
-
\hat{\psi}_h^{\infty},
\qquad
X_{h,0}^{\dagger}
=
D_{h,0}^{\dagger}
-
\psi_h.
\]
The preceding \(L^2\) convergence implies
\[
\mathbb{E}^{\dagger}
\left[
\left.
\left|
X_{h,T}^{\dagger}
-
X_{h,0}^{\dagger}
\right|^2
\right|
\mathcal{F}_T
\right]
\xrightarrow{\mathbb{P}}
0.
\]
By conditional Cauchy--Schwarz,
\[
\begin{aligned}
&\mathbb{E}^{\dagger}
\left[
\left.
\left|
\left(
X_{h,T}^{\dagger}
\right)^2
-
\left(
X_{h,0}^{\dagger}
\right)^2
\right|
\right|
\mathcal{F}_T
\right]
\\
&\quad\leq
\left\{
\mathbb{E}^{\dagger}
\left[
\left.
\left|
X_{h,T}^{\dagger}
-
X_{h,0}^{\dagger}
\right|^2
\right|
\mathcal{F}_T
\right]
\right\}^{1/2}
\left\{
\mathbb{E}^{\dagger}
\left[
\left.
\left(
\left|
X_{h,T}^{\dagger}
\right|
+
\left|
X_{h,0}^{\dagger}
\right|
\right)^2
\right|
\mathcal{F}_T
\right]
\right\}^{1/2}
\xrightarrow{\mathbb{P}}
0.
\end{aligned}
\]
Thus, the conditional squares are uniformly integrable in probability. Since \(\varepsilon\sqrt{S}\rightarrow\infty\), for every \(\varepsilon>0\),
\[
\mathbb{E}^{\dagger}
\left[
\left.
\left(
X_{h,T}^{\dagger}
\right)^2
\mathbf{1}
\left\{
\left|
X_{h,T}^{\dagger}
\right|
>
\varepsilon\sqrt{S}
\right\}
\right|
\mathcal{F}_T
\right]
=
o_{\mathbb{P}}\left(1\right).
\]
Conditional on \(\mathcal{F}_T\), the Lindeberg--Feller theorem in \citet[Proposition~2.27]{VanDerVaart1998} and equation~\eqref{eq:proof-path-only-variance} therefore give
\[
\sqrt{S}
\left(
\hat{\psi}_{h,S}^{\dagger}
-
\hat{\psi}_h^{\infty}
\right)
\xrightarrow{\mathrm{d}^{\dagger}}
\mathcal{N}\left(
0,
\mathcal{V}_h
\right)
\quad
\text{in probability}.
\]
This proves equation~\eqref{eq:path-only-limit}. The same argument includes the degenerate case \(\mathcal{V}_h=0\), in which the limiting distribution is concentrated at zero.

The identity
\[
\sqrt{T}
\left(
\hat{\psi}_{h,S}^{\dagger}
-
\hat{\psi}_h^{\infty}
\right)
=
\sqrt{\frac{T}{S}}
\sqrt{S}
\left(
\hat{\psi}_{h,S}^{\dagger}
-
\hat{\psi}_h^{\infty}
\right)
\]
gives the remaining conclusions. If \(S/T\rightarrow\infty\), the first factor converges to zero and the second is conditionally tight, so the product converges to zero in conditional probability. If \(S/T\rightarrow\kappa\in\left(0,\infty\right)\), conditional Slutsky gives
\[
\sqrt{T}
\left(
\hat{\psi}_{h,S}^{\dagger}
-
\hat{\psi}_h^{\infty}
\right)
\xrightarrow{\mathrm{d}^{\dagger}}
\mathcal{N}\left(
0,
\frac{\mathcal{V}_h}{\kappa}
\right)
\quad
\text{in probability}.
\]
Both limits condition on the fitted transition and residual quantiles and therefore contain no term involving \(\Omega_{hh}\).
\end{proof}

\begin{proof}[Proof of Theorem~\ref{thm:finite-simulation-limit}]
For the fixed collection of responses, define
\[
\boldsymbol{\Psi}\left(
\boldsymbol{\beta},
\boldsymbol{Q}
\right)
=
\mathbb{E}
\left[
\boldsymbol{D}\left(
\boldsymbol{P}_{1:H};
\boldsymbol{\beta},
\boldsymbol{Q}
\right)
\right].
\]
Let
\[
\boldsymbol{D}_T\left(
\boldsymbol{P}_{1:H}
\right)
=
\boldsymbol{D}\left(
\boldsymbol{P}_{1:H};
\boldsymbol{\hat{\beta}},
\boldsymbol{\hat{Q}}^{\mathrm{E}}
\right),
\qquad
\boldsymbol{D}_0\left(
\boldsymbol{P}_{1:H}
\right)
=
\boldsymbol{D}\left(
\boldsymbol{P}_{1:H};
\boldsymbol{\beta}_0,
\boldsymbol{Q}_0
\right).
\]
For any square-integrable vector-valued function \(\boldsymbol{f}\) of a master rank path, define
\[
\mathbb{G}_S\boldsymbol{f}
=
\frac{1}{\sqrt{S}}
\sum_{s=1}^{S}
\left\{
\boldsymbol{f}\left(
\boldsymbol{P}_{s,1:H}
\right)
-
\mathbb{E}
\left[
\boldsymbol{f}\left(
\boldsymbol{P}_{1:H}
\right)
\right]
\right\}.
\]

Equations~\eqref{eq:proof-main-functional-expansion}--\eqref{eq:proof-main-quantile-decomposition} concern \(\Psi_h\) and therefore do not use the rate of \(S\). Stacking the fixed response collection gives
\begin{equation}
\sqrt{T}
\left[
\boldsymbol{\Psi}\left(
\boldsymbol{\hat{\beta}},
\boldsymbol{\hat{Q}}^{\mathrm{E}}
\right)
-
\boldsymbol{\psi}
\right]
=
\frac{1}{\sqrt{T}}
\sum_{t=1}^{T}
\boldsymbol{Z}_t
+
o_{\mathbb{P}}\left(1\right).
\label{eq:proof-finite-s-exact-integration}
\end{equation}

By definition,
\[
\sqrt{T}
\left(
\boldsymbol{\hat{\psi}}_{S}^{\mathrm{E}}
-
\boldsymbol{\psi}
\right)
=
\sqrt{T}
\left[
\boldsymbol{\Psi}\left(
\boldsymbol{\hat{\beta}},
\boldsymbol{\hat{Q}}^{\mathrm{E}}
\right)
-
\boldsymbol{\psi}
\right]
+
\sqrt{\frac{T}{S}}
\mathbb{G}_S\boldsymbol{D}_T.
\]
Substituting equation~\eqref{eq:proof-finite-s-exact-integration} gives
\begin{equation}
\sqrt{T}
\left(
\boldsymbol{\hat{\psi}}_{S}^{\mathrm{E}}
-
\boldsymbol{\psi}
\right)
=
\frac{1}{\sqrt{T}}
\sum_{t=1}^{T}
\boldsymbol{Z}_t
+
\sqrt{\frac{T}{S}}
\mathbb{G}_S\boldsymbol{D}_T
+
o_{\mathbb{P}}\left(1\right).
\label{eq:proof-finite-s-initial-decomposition}
\end{equation}

Conditional on the observed sample, the simulation ranks are independent and
\[
\mathbb{E}
\left[
\left.
\left\|
\mathbb{G}_S
\left(
\boldsymbol{D}_T
-
\boldsymbol{D}_0
\right)
\right\|^2
\right|
\boldsymbol{Y}_0,
\ldots,
\boldsymbol{Y}_T
\right]
\leq
\mathbb{E}
\left[
\left.
\left\|
\boldsymbol{D}_T\left(
\boldsymbol{P}_{1:H}
\right)
-
\boldsymbol{D}_0\left(
\boldsymbol{P}_{1:H}
\right)
\right\|^2
\right|
\boldsymbol{Y}_0,
\ldots,
\boldsymbol{Y}_T
\right].
\]
The right-hand side converges to zero in probability by the \(L^2\) part of Lemma~\ref{lem:recursive-response-differentiability}, equation~\eqref{eq:beta-linearization}, equation~\eqref{eq:generated-quantile-expansion}, and the fixed response collection. Conditional Chebyshev's inequality gives
\[
\mathbb{G}_S
\left(
\boldsymbol{D}_T
-
\boldsymbol{D}_0
\right)
=
o_{\mathbb{P}}\left(1\right).
\]
Since \(T/S\rightarrow\kappa^{-1}\), equation~\eqref{eq:proof-finite-s-initial-decomposition} becomes
\[
\begin{aligned}
\sqrt{T}
\left(
\boldsymbol{\hat{\psi}}_{S}^{\mathrm{E}}
-
\boldsymbol{\psi}
\right)
&=
\frac{1}{\sqrt{T}}
\sum_{t=1}^{T}
\boldsymbol{Z}_t
\\
&\quad+
\sqrt{\frac{T}{S}}
\frac{1}{\sqrt{S}}
\sum_{s=1}^{S}
\left[
\boldsymbol{D}_0\left(
\boldsymbol{P}_{s,1:H}
\right)
-
\boldsymbol{\psi}
\right]
+
o_{\mathbb{P}}\left(1\right),
\end{aligned}
\]
which proves equation~\eqref{eq:finite-simulation-expansion}.

For every \(\boldsymbol{c}\in\mathbb{R}^{M}\), the variables
\[
\boldsymbol{c}^{\prime}
\left[
\boldsymbol{D}_0\left(
\boldsymbol{P}_{s,1:H}
\right)
-
\boldsymbol{\psi}
\right],
\qquad
s=1,\ldots,S,
\]
are independent and identically distributed, centered, and have a finite \(2+\eta\) moment. The Lindeberg--Feller theorem in \citet[Proposition~2.27]{VanDerVaart1998} and the Cram\'er--Wold device give
\[
\frac{1}{\sqrt{S}}
\sum_{s=1}^{S}
\left[
\boldsymbol{D}_0\left(
\boldsymbol{P}_{s,1:H}
\right)
-
\boldsymbol{\psi}
\right]
\xrightarrow{\mathrm{d}}
\mathcal{N}\left(
\boldsymbol{0},
\boldsymbol{\Omega}_{\mathrm{MC}}
\right).
\]
Proposition~\ref{prop:recursive-joint-clt} gives
\[
\frac{1}{\sqrt{T}}
\sum_{t=1}^{T}
\boldsymbol{Z}_t
\xrightarrow{\mathrm{d}}
\mathcal{N}\left(
\boldsymbol{0},
\boldsymbol{\Omega}
\right).
\]

For every \(T\) and \(S\), the first sum in equation~\eqref{eq:finite-simulation-expansion} is measurable with respect to the observed sample, while the second is measurable with respect to the independently generated rank array. Thus, the two sums are independent. For \(\boldsymbol{c},\boldsymbol{d}\in\mathbb{R}^{M}\), their joint characteristic function factors as
\[
\begin{aligned}
&\mathbb{E}
\exp
\left\{
\mathrm{i}
\boldsymbol{c}^{\prime}
\frac{1}{\sqrt{T}}
\sum_{t=1}^{T}
\boldsymbol{Z}_t
+
\mathrm{i}
\boldsymbol{d}^{\prime}
\frac{1}{\sqrt{S}}
\sum_{s=1}^{S}
\left[
\boldsymbol{D}_0\left(
\boldsymbol{P}_{s,1:H}
\right)
-
\boldsymbol{\psi}
\right]
\right\}
\\
&\quad=
\mathbb{E}
\exp
\left\{
\mathrm{i}
\boldsymbol{c}^{\prime}
\frac{1}{\sqrt{T}}
\sum_{t=1}^{T}
\boldsymbol{Z}_t
\right\}
\mathbb{E}
\exp
\left\{
\mathrm{i}
\boldsymbol{d}^{\prime}
\frac{1}{\sqrt{S}}
\sum_{s=1}^{S}
\left[
\boldsymbol{D}_0\left(
\boldsymbol{P}_{s,1:H}
\right)
-
\boldsymbol{\psi}
\right]
\right\}.
\end{aligned}
\]
The limiting characteristic function is that of two independent centered Gaussian vectors with covariance matrices \(\boldsymbol{\Omega}\) and \(\boldsymbol{\Omega}_{\mathrm{MC}}\). Since \(\sqrt{T/S}\rightarrow\kappa^{-1/2}\), Slutsky's lemma, as stated in \citet[Lemma~2.8]{VanDerVaart1998}, and equation~\eqref{eq:finite-simulation-expansion} give
\[
\sqrt{T}
\left(
\boldsymbol{\hat{\psi}}_{S}^{\mathrm{E}}
-
\boldsymbol{\psi}
\right)
\xrightarrow{\mathrm{d}}
\mathcal{N}\left(
\boldsymbol{0},
\boldsymbol{\Omega}
+
\kappa^{-1}
\boldsymbol{\Omega}_{\mathrm{MC}}
\right).
\]
The characteristic-function factorization also proves asymptotic independence of the sampling and simulation components.
\end{proof}

\begin{proof}[Consistency of the correction in equation~\eqref{eq:finite-simulation-correction}]
Let
\[
\boldsymbol{D}_{s,0}
=
\boldsymbol{D}_0\left(
\boldsymbol{P}_{s,1:H}
\right).
\]
The \(L^2\) convergence used in Step~3 of the preceding proof and the conditional law of large numbers give
\[
\frac{1}{S}
\sum_{s=1}^{S}
\left\|
\boldsymbol{\hat{D}}_s^{\mathrm{E}}
-
\boldsymbol{D}_{s,0}
\right\|^2
=
o_{\mathbb{P}}\left(1\right).
\]
By Cauchy--Schwarz,
\[
\left\|
\frac{1}{S}
\sum_{s=1}^{S}
\boldsymbol{\hat{D}}_s^{\mathrm{E}}
\left(
\boldsymbol{\hat{D}}_s^{\mathrm{E}}
\right)^{\prime}
-
\frac{1}{S}
\sum_{s=1}^{S}
\boldsymbol{D}_{s,0}
\boldsymbol{D}_{s,0}^{\prime}
\right\|
=
o_{\mathbb{P}}\left(1\right).
\]
The ordinary law of large numbers gives
\[
\frac{1}{S}
\sum_{s=1}^{S}
\boldsymbol{D}_{s,0}
\xrightarrow{\mathbb{P}}
\boldsymbol{\psi},
\qquad
\frac{1}{S}
\sum_{s=1}^{S}
\boldsymbol{D}_{s,0}
\boldsymbol{D}_{s,0}^{\prime}
\xrightarrow{\mathbb{P}}
\mathbb{E}
\left[
\boldsymbol{D}_0\left(
\boldsymbol{P}_{1:H}
\right)
\boldsymbol{D}_0\left(
\boldsymbol{P}_{1:H}
\right)^{\prime}
\right].
\]
Therefore,
\[
\boldsymbol{\hat{\Omega}}_{\mathrm{MC}}^{\mathrm{E}}
\xrightarrow{\mathbb{P}}
\boldsymbol{\Omega}_{\mathrm{MC}}.
\]
Together with \(\boldsymbol{\hat{\Omega}}\xrightarrow{\mathbb{P}}\boldsymbol{\Omega}\) and \(T/S\rightarrow\kappa^{-1}\), this yields
\[
\boldsymbol{\hat{\Omega}}_{T,S}^{\mathrm{tot}}
\xrightarrow{\mathbb{P}}
\boldsymbol{\Omega}
+
\kappa^{-1}
\boldsymbol{\Omega}_{\mathrm{MC}}.
\]
Thus, equation~\eqref{eq:finite-simulation-correction} consistently estimates the covariance matrix in equation~\eqref{eq:finite-simulation-limit}.
\end{proof}

\begin{proof}[Common-random-number covariance in equation~\eqref{eq:paired-mc-covariance}]
For \(r\in\left\{\mathrm{E},\mathrm{S}\right\}\), let
\[
\boldsymbol{D}_T^r\left(
\boldsymbol{P}_{1:H}
\right)
\]
denote the corresponding fitted path-response vector. Consistency of both quantile estimators and the \(L^2\) continuity in Lemma~\ref{lem:recursive-response-differentiability} imply
\[
\mathbb{E}
\left[
\left.
\left\|
\boldsymbol{D}_T^r\left(
\boldsymbol{P}_{1:H}
\right)
-
\boldsymbol{D}_0\left(
\boldsymbol{P}_{1:H}
\right)
\right\|^2
\right|
\boldsymbol{Y}_0,
\ldots,
\boldsymbol{Y}_T
\right]
\xrightarrow{\mathbb{P}}
0.
\]
Consequently, for \(r,q\in\left\{\mathrm{E},\mathrm{S}\right\}\),
\[
\operatorname{Cov}
\left[
\boldsymbol{D}_T^r\left(
\boldsymbol{P}_{1:H}
\right),
\boldsymbol{D}_T^q\left(
\boldsymbol{P}_{1:H}
\right)
\mathrel{\big|}
\boldsymbol{Y}_0,
\ldots,
\boldsymbol{Y}_T
\right]
\xrightarrow{\mathbb{P}}
\boldsymbol{\Omega}_{\mathrm{MC}}.
\]
The conditional law of large numbers gives
\[
\boldsymbol{\hat{\Omega}}_{\mathrm{MC}}^{rq}
\xrightarrow{\mathbb{P}}
\boldsymbol{\Omega}_{\mathrm{MC}}.
\]
Using the same rank paths for both estimators,
\[
\boldsymbol{\hat{\Omega}}_{\mathrm{MC}}^{\Delta}
=
\boldsymbol{\hat{\Omega}}_{\mathrm{MC}}^{\mathrm{SS}}
+
\boldsymbol{\hat{\Omega}}_{\mathrm{MC}}^{\mathrm{EE}}
-
\boldsymbol{\hat{\Omega}}_{\mathrm{MC}}^{\mathrm{SE}}
-
\boldsymbol{\hat{\Omega}}_{\mathrm{MC}}^{\mathrm{ES}}
\xrightarrow{\mathbb{P}}
\boldsymbol{0}.
\]
If the estimators use independent rank arrays, their conditional cross-covariances are zero. The simulation covariance of their difference then converges to
\[
\boldsymbol{\Omega}_{\mathrm{MC}}
+
\boldsymbol{\Omega}_{\mathrm{MC}}
=
2\boldsymbol{\Omega}_{\mathrm{MC}}.
\]
\end{proof}
\section{Implementation and reporting conventions}
\label{app:implementation}

The theoretical results allow different transition estimators, residual normalizations, numerical differentiation routines, and rank smoothers, provided the corresponding assumptions are verified. Comparisons across estimators require a stricter convention: every element of the response specification other than the quantile estimator must remain fixed.

\begin{enumerate}[label=\textbf{\arabic*.},leftmargin=3.0em]
\item \textbf{Transition and residuals.} Estimate \(\boldsymbol{\beta}_0\) by the procedure that justifies equation~\eqref{eq:beta-linearization}. Recover the structural residuals with the same residual map and impose the same location, scale, sign, and ordering normalizations in the original sample and in every bootstrap replication.

\item \textbf{Empirical quantiles and endpoint treatment.} Construct each component of \(\boldsymbol{\hat{Q}}^{\mathrm{E}}\) from the corresponding residual order statistics. Any clipping or endpoint rule must be applied identically to all estimators and bootstrap replications. Its effect on the reported response must satisfy the negligibility condition in Remark~\ref{rem:residual-quantile-clipping}. For a smoothed comparison, report the smoothing operator, bandwidth, and boundary treatment and verify Assumption~\ref{ass:smoothing}.

\item \textbf{Response simulation.} Generate one master array of rank paths through the largest reported horizon and use it for every response in the fixed collection. Use the same array for empirical and smoothed estimates. The path count \(S\) and the estimated numerical-integration shares should be reported. When these shares are material, use the covariance correction in equation~\eqref{eq:finite-simulation-correction}.

\item \textbf{Sampling inference.} Influence-function inference uses the observation-level contributions in equation~\eqref{eq:estimated-influence-contributions} and the long-run covariance estimator in equation~\eqref{eq:feasible-long-run-covariance}. Bootstrap inference uses Algorithm~\ref{alg:recursive-residual-bootstrap} and requires full regeneration and re-estimation. The bootstrap theorem conditions on fixed response-simulation ranks and assumes \(S/T\rightarrow\infty\). When \(S/T\) has a finite positive limit, the analytical correction in equation~\eqref{eq:finite-simulation-correction} accounts for numerical integration; Theorem~\ref{thm:bootstrap-validity} alone does not supply a finite-\(S\) bootstrap correction.

\item \textbf{Reported information.} A replication archive should record the transition estimator, innovation normalizations, initial state, response variable, shocked component, shock map, future innovation law, horizons, random-number seeds, number of response paths, covariance or bootstrap procedure, and any endpoint or smoothing choices. These records are sufficient to verify that an estimator comparison preserves the population response defined in equation~\eqref{eq:population-irf}.
\end{enumerate}

\bibliographystyle{apalike}
\bibliography{nonlinear_irf_references}

\end{document}